\documentclass[fleqn,nonblindrev,opre]{informs4}

\OneAndAHalfSpacedXI 

\usepackage{endnotes}
\let\footnote=\endnote

\let\enotesize=\normalsize
\def\notesname{Endnotes}%
\def\makeenmark{$^{\theenmark}$}
\def\enoteformat{\rightskip0pt\leftskip0pt\parindent=1.75em
	\leavevmode\llap{\theenmark.\enskip}}

\usepackage{graphicx}
\usepackage{eqndefns-left}
\usepackage{multirow}
\usepackage{hhline}
\usepackage{framed}
\usepackage{float}
\usepackage{enumitem}
\usepackage[small, margin=1cm]{caption}
\usepackage{algorithm}
\usepackage{algpseudocode}
\usepackage{appendix}

\expandafter\def\expandafter\normalsize\expandafter{%
    \normalsize%
    \setlength\abovedisplayskip{5pt}%
    \setlength\belowdisplayskip{5pt}%
    \setlength\abovedisplayshortskip{3pt}%
    \setlength\belowdisplayshortskip{3pt}%
}

\usepackage{bm}
\usepackage{color}
\definecolor{strcolor}{rgb}{0.6, 0.2, 0.6}
\definecolor{commentcolor}{rgb}{0.3125, 0.5, 0.3125}
\definecolor{keycol}{rgb}{0, 0, 1}

\usepackage{bbm}

\usepackage{listings}
\usepackage{url}

\newcommand {\bea}{\begin{eqnarray}}
	\newcommand {\eea}{\end{eqnarray}}
\newcommand {\E}[1]{\mathrm{E}\left( #1 \right)}

\newcommand {\p}{{\rm P}}

\def\blot{\quad \mbox{$\vcenter{ \vbox{ \hrule height.4pt
				\hbox{\vrule width.4pt height.9ex \kern.9ex \vrule width.4pt}
				\hrule height.4pt}}$}}

\usepackage{natbib}
\bibpunct[, ]{(}{)}{,}{a}{}{,}%
\def\bibfont{\fontsize{8}{9.5}\selectfont}%
\TheoremsNumberedThrough     
\ECRepeatTheorems

\EquationsNumberedThrough    

\gdef\AQ#1{}
\gdef\CQ#1{}

\usepackage{xcolor}
\definecolor{darkblue}{rgb}{0, 0, 0.65} 

\def \implies{\Rightarrow}
\def \pd#1#2{\frac{\partial #1}{\partial #2}}

\def\b{\mathbf{b}}
\def\v{\mathbf{v}}
\def\a{\mathbf{a}}
\def\p{\mathbf{p}}

\def\apr{(\mathbf{a},\mathbf{p},r)}

\def\cite#1{\citet{#1}}

\def\UUB{{\tt U-UB}}
\def\UFT{{\tt U-FT}}
\def\MFT{{\tt M-FT}}
\def\MTD{{\tt M-TD}}
\def\MUUB#1{{\tt MU$^#1$-UB}}

\def\UB{{\tt UB}}
\def\FT{{\tt FT}}
\def\TD{{\tt TD}}

\def\bmtheta{{\bm{\theta}}}

\def\E{\mathbb{E}}

\def\com#1{\textcolor{blue}{[#1]}}

\def\deleted#1{}
\def\comment#1{}

\def\D{\mathbf{D}}

\newtheorem{observation}{Observation}

\newfloat{Mechanism}{ht}{ext}

\def \R{\mathbb{R}}

\definecolor{DSgray}{cmyk}{0,1,0,0}

\MANUSCRIPTNO{OPRE-2024-03-0865.R1}
\begin{document}
	
\def\COPYRIGHTHOLDER{INFORMS}%
\def\COPYRIGHTYEAR{2017}%
\def\DOI{\fontsize{7.5}{9.5}\selectfont\sf\bfseries\noindent https://doi.org/10.1287/opre.2017.1714\CQ{Word count = 9740}}

	\RUNAUTHOR{Chen et~al.} %

	\RUNTITLE{Bayesian Mechanism Design for Blockchain Transaction Fee Allocation}

\TITLE{Bayesian Mechanism Design for Blockchain Transaction Fee Allocation}


	\ARTICLEAUTHORS{%
\AUTHOR{Xi Chen$^*$}
\AFF{Leonard N. Stern School of Business, New York University, New York, NY 10012, USA, \EMAIL{xc13@stern.nyu.edu}}

\AUTHOR{David Simchi-Levi$^*$}
\AFF{Institute for Data, Systems and Society, Operations Research Center, Department of Civil and Environmental Engineering,
Massachusetts Institute of Technology, Cambridge, MA 02139, USA, 
\EMAIL{dslevi@mit.edu}}

\AUTHOR{Zishuo Zhao$^*$}
\AFF{Department of Industrial and Enterprise Systems Engineering, University of Illinois Urbana-Champaign, Urbana, IL 61801, USA, \EMAIL{zishuoz2@illinois.edu}} 

\AUTHOR{Yuan Zhou$^*$}
\AFF{Yau Mathematical Sciences Center, Tsinghua University; Beijing Institute of Mathematical Sciences and Application; Department of Mathematical Sciences, Tsinghua University; Beijing 100084, China, \EMAIL{yuan-zhou@tsinghua.edu.cn}
}
} 
	

\ABSTRACT{%
In blockchain systems, the design of transaction fee mechanisms is essential for stability and satisfaction for both miners and users. A recent work has proven the impossibility of collusion-proof mechanisms that achieve both non-zero miner revenue and {Dominant-Strategy-Incentive-Compatible} (DSIC) for users. However, a positive miner revenue is important in practice to motivate miners. To address this challenge,  we consider a \emph{Bayesian game} setting and relax the DSIC  requirement for users to Bayesian-Nash-Incentive-Compatibility (BNIC). In particular, we propose an auxiliary mechanism method that makes connections between BNIC and DSIC mechanisms. With the auxiliary mechanism method, we design a transaction fee mechanism (TFM) based on the multinomial logit (MNL) choice model, and prove that the TFM has both BNIC and collusion-proof properties with an asymptotic constant-factor approximation of optimal miner revenue for \emph{i.i.d.} bounded valuations. Our result breaks the zero-revenue barrier while preserving truthfulness and collusion-proof properties. 
\ifdefined\EnableBurning
We additionally prove the necessity to either burn transaction fees or have variable block sizes to satisfy both BNIC and collusion proofness guarantees.
\fi
}%


\SUBJECTCLASS{Blockchain, mechanism design, Bayesian game}

\AREAOFREVIEW{Markets, Platforms, and Revenue Management.}


	
	%
	
\maketitle
\def\thefootnote{*}\footnotetext{Authors listed in alphabetical order.}\def\thefootnote{\arabic{footnote}}

\section{Introduction}


The blockchain, as a new decentralized technology, is becoming an interesting research object for the Operations community (see, e.g., \cite{davydiuk2023crypto,iyengar2022economics,manzoor2022blockchain,whitaker2020fractional} and references therein). Just like the emergence of the ridesharing topic ten years ago, the special structure in blockchain poses many unique challenges in auction theory, game theory, scheduling, and optimization; in turn, the blockchain technology also has applications that foster traditional aspects of operation research, e.g., newsvendors and supply chains \citep{keskin2023blockchain}. Particularly, in the scope of game theory, \cite{liu2019survey} surveys a variety of its applications in blockchain systems.

Let us zoom in and briefly discuss the structure of a standard blockchain. A blockchain is essentially a linked list (or a chain) of blocks, where each block stores a number of transactions. There are two types of agents participating in a blockchain: \emph{users} and \emph{miners}. Users propose to put \emph{transactions} on the chain, and miners pack transactions into a block and then send blocks to the chain. Once the block has been finalized on the chain, the miner will receive tokens (e.g., Bitcoin) as a reward.  Generally, each block may contain multiple transactions, but only one miner claims ownership of the block and obtains the corresponding reward. An illustration of the generation process of each block is shown in Figure~\ref{fig:gen}.

\begin{figure}[!t]
    \centering
    \includegraphics[width = 0.65\textwidth]{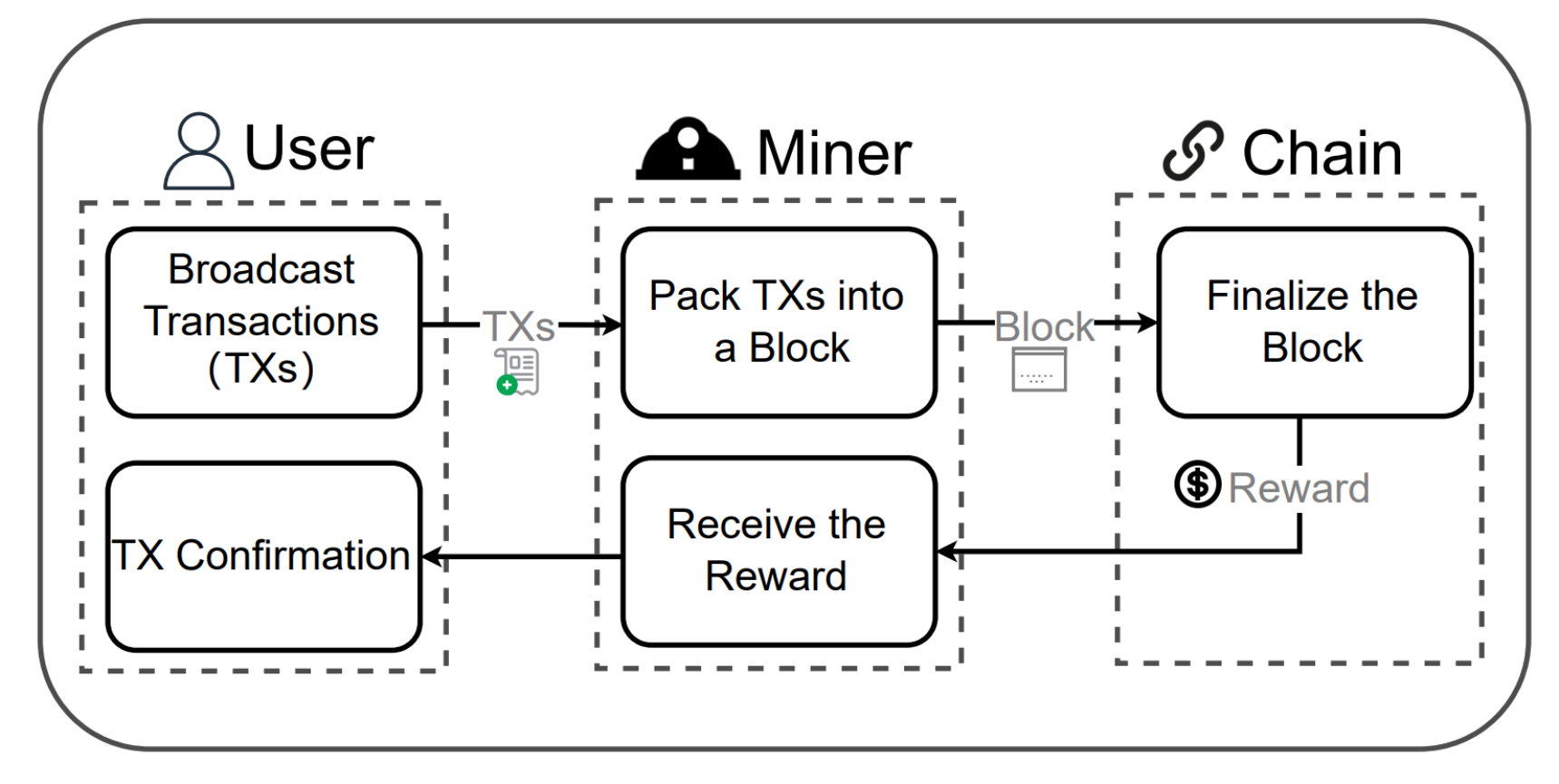}
    \caption{The role of different parties in generating one block in a blockchain.}
    \label{fig:gen}
\end{figure}

The blockchain stores the blocks sequentially in the time order. After a miner creates a new block, the block is appended to the chain via a reference to the latest existing block.
For the efficiency of the block space, each block only contains a pre-specified limited number of transactions. In order to retrieve the status (e.g., balances of each user), we have to track from the beginning of the chain and go through all previous transactions to determine the current status at any block. We show the structure of the blockchain in Figure~\ref{fig:chain}. 

\begin{figure}[!t]
    \centering
    \includegraphics[width = 0.9\textwidth]{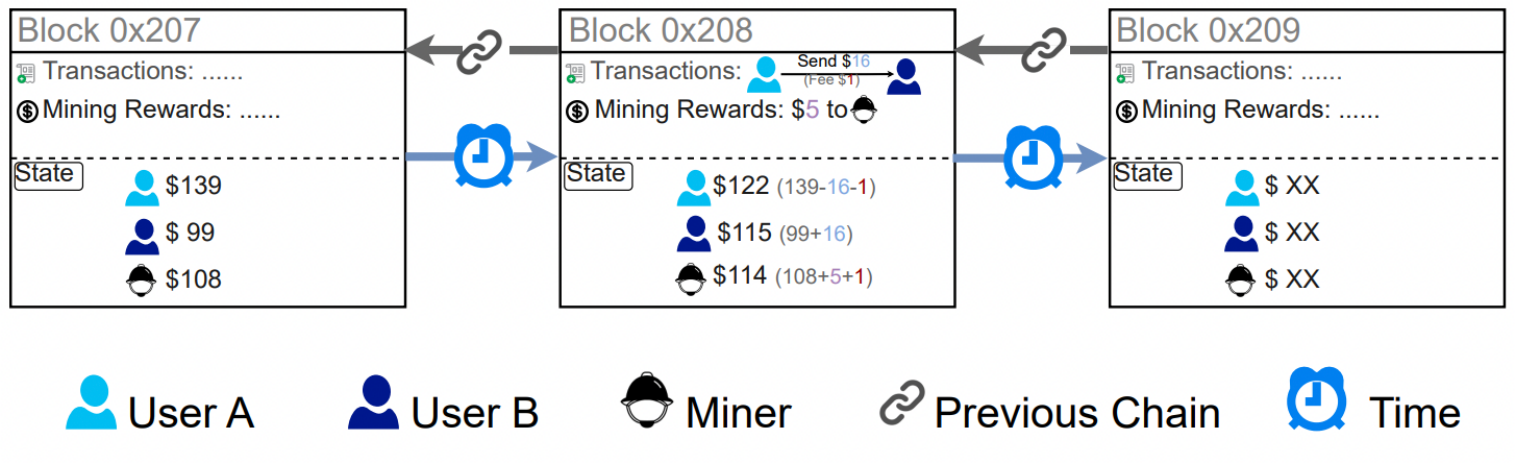}
    \caption{Illustration of the blockchain structure.
    Above the dashed lines are the blocks and their contents (transactions and rewards). The blocks are arranged from left to right in the time order, and each block is linked to the previous one on its left. Below the dashed lines are the states of the system -- the amount of money owned by each party at the time.}
    \label{fig:chain}
\end{figure}

As each block only contains a limited number of transactions, {the major bottleneck of limited-space design in the blockchain systems} draws wide research interest in the field of mechanism design (e.g., \cite{wang2019consensus}). With this bottleneck, users need to compete to win a transaction space in the block. Such competition can naturally be implemented via an auction. However, as we will explain later, the auction design in a blockchain exhibits unique challenges on how to charge the users properly and how to reward the miner. This problem is usually referred to as \emph{``transaction fee mechanism (TFM) design''}, which has been modeled by a seminal paper of  \cite{roughgarden2021transaction}.  In more detail,  to incentivize the miners to mine the block, the blockchain systems adopt economic mechanisms to pay miners via cryptocurrency. Such payments usually consist of a mining reward, and an additional reward extracted from transaction fees paid by users {named as the \emph{miner revenue}}.  As users benefit from transactions being confirmed on the blockchain and miners need the incentivization, it is reasonable to charge transaction fees from users for confirmed transactions. As described by \cite{roughgarden2021transaction}, the on-chain space is a scarce resource, so to facilitate the social efficiency of the system, we want to confirm transactions of high values. Therefore, many blockchains adopt bidding-confirmation transaction fee mechanisms (TFM) such as auctions.

\deleted{Blockchains, like Bitcoin \citep{nakamoto2008bitcoin} and Ethereum \citep{wood2014ethereum}, are essentially distributed databases built by consensus and growing over time with history data saved in ``blocks'' as linked lists created by miners. In each block, users leverage blockchains to either store or verify information, such as money (cryptocurrency) transfers, texts, and modern data such as smart contracts. This on-chain information is usually referred to as ``transactions''. Miners, in turn, try to ``mine'' a block by getting access to write a block at a cost (computational power in Proof-of-Work (PoW) \citep{shi2016new} or cryptocurrency deposit in Proof-of-Stake (PoS) \citep{kiayias2017ouroboros}) and then putting information into the block. Generally, each block may contain multiple transactions, but only one winning miner has the access to own the block.}

However, due to the online and anonymous properties of blockchains, the design of mechanisms for blockchain systems faces a major concern of \emph{credibility} \citep{he2023blockchain}. Compared to traditional auctions, the miner has a wider strategy space (than a traditional auctioneer) to conduct dishonest activities, including injecting fake transactions, concealing users' bids, and colluding with users. Therefore, it is important to address these unique challenges raised in the blockchain setting and develop a desirable TFM that discourages all possible dishonest activities to make sure the whole on-chain economic system can operate correctly. 

In a blockchain system, there are mainly three types of dishonest behaviors: \emph{Untruthful Bids  ({\tt UB})},  \emph{Fake Transactions ({\tt FT})}, and \emph{Transaction Deletion ({\tt TD})}. The agents who might conduct these dishonest behaviors include \emph{individual user ({\tt U})},  \emph{individual miner ({\tt M})}, \emph{miner colluding with $c$ users ({\tt MU$^c$})}, and \emph{collusion among $c$ users  ({\tt U$^c$}) } with  $c \geq 2$.  We provide a table to summarize all types of deviations at the end of Section \ref{sec:contri}. 

In our research, we prevent these types of dishonest behavior in a \emph{mixed} way consisting of cryptographic and economic techniques. In the cryptographic part, we introduce a \emph{commitment scheme} from \citep{credibility} and adapt it into our proposed mechanism (as shown in Appendix~\ref{sec:crypto}) that both essentially runs a sealed-bid auction that ensures fairness among users' information sets and restricts the miner's strategic space, also resolving the {MEV (miner-extractable-value or maximal-extractable-value)  issue, in which miners may gain additional revenue via strategically injecting, excluding or re-ordering the transactions \citep{daian2020flash}}, via pinning down the transaction orders before they are revealed. As shown in Appendix~\ref{sec:crypto}, some types of dishonest behavior (particularly deviations by an individual miner) can be effectively prevented via cryptographic protocols. Furthermore, the anonymity property of the blockchain system also brings intrinsic difficulty to collusions among users.
Therefore, in the economic part, we mainly focus on the \emph{prevention of \textbf{individual user's deviation} and \textbf{miner-user collusion}.}

\subsection{Research Question: How to Design Truthful and Collusion-Proof TFMs}
\label{sec:research_q}


To understand our results, we provide the necessary background on the truthfulness of users and miners. In the standard auction theory, a strong version of the truthfulness of users can be specified as the {User-Dominant-Strategy-Incentive-Compatibility} (U-DSIC), which means that any individual user will not benefit from deviation from truthful bidding even if she knows all bids of other users (as in Definition~\ref{def:udsic}). In many practical scenarios, it would be difficult to satisfy such a strong notion of truthfulness. Instead, a weaker version of the truthfulness of users is studied in this paper, i.e., the \emph{User-Bayesian-Nash-Incentive-Compatibility (U-BNIC)} (also known as Bayesian-Incentive-Compatibility (BIC) in some literature). In particular, U-BNIC means that when each user only knows the distribution of others' valuations, the game achieves a Bayesian Nash equilibrium when all users truthfully bid their valuations 
 (see Definition~\ref{def:ubnic}).  The truthfulness of the miner can be specified as the Miner-Incentive-Compatibility (MIC), which means that the miner will not benefit from untruthful behavior, e.g., injecting fake transactions or ignoring existing transactions. For the issue of collusion, the paper by \cite{shi} formulates collusion-proofness as $c$-Side-Contract-Proof ($c$-SCP): when the miner colludes with at most $c$ users by asking them to change their bids, the coalition cannot increase the total utility by deviations from truthfully bidding their valuations (as in Definition~\ref{def:cscp}). 

The key question is how to design TFMs to guarantee incentive compatibility and collusion-proof requirements, {and the existing works on TFM design can be roughly classified into two families: \emph{auction-like} mechanisms in which confirmed users' payments are dependent on the bids of the current block, and \emph{posted-price} mechanisms in which their payments are completely based on statistics of previous blocks.} The most intuitive form of the single-round auction mechanism is the \emph{first-price auction} \citep{vickrey61}, in which the auctioneer collects all users' bids, and sells the item to the user who bids the highest at the price she bids. When we generalize the first-price auction to the setting of multiple identical items, the auctioneer sells the items to users with the $k$-highest bids, charging the users what they bid. The Bitcoin blockchain essentially uses the multi-item first-price auction, but it is not truthful: users tend to bid lower than their valuations. A famous DSIC auction mechanism is the \emph{(multi-item) second-price auction} \citep{vickrey61}, where the winners are also the $k$-highest bidders but their payment is the $(k+1)$-th highest bid.  However, the second-price auction is susceptible to miner-user collusion as the miner may collude with the $(k+1)$-th highest bidder by asking her to raise her bid as long as it is still lower than the $k$-th highest bid. In such a way, the $(k+1)$-th highest bidder still gets the same utility $0$ while the miner gets a higher revenue, increasing the total utility of the colluding party. 

To solve the issue of collusion, the EIP-1559 mechanism \citep{roughgarden2020transaction} in Ethereum seeks to avoid collusion by adopting {a dynamic} \emph{posted-price} mechanism, as long as the posted price is ``well chosen'' {from historical demands, in expectation that there is no congestion on the block size --- if there is nothing to bid, there is significantly less space for strategic bidding (i.e., \UB). }

\color{black}

{However, in the TFM of EIP-1559, the miner typically gets no revenue from the transaction fees} as the fees have to be ``\emph{burnt}'' and removed from the blockchain to maintain collusion-proof properties {(details discussed in Section~\ref{sec:posted:price})}. While EIP-1559 prevents collusion between users and miners, it is economically inefficient for miners {as miners get no rewards from the transaction fees}. Therefore, a natural question would be the following:

\begin{center}
\begin{framed}
{\it
Can we design a TFM that satisfies both truthfulness \\
and collusion-proof conditions, and has a desirable miner revenue?
}
\end{framed}
\end{center}

To answer this question, the paper by \cite{shi} proves a negative result (Theorem~\ref{thm:shi:impossible}) under the \emph{{User-Dominant-Strategy-Incentive-Compatibility} (U-DSIC)}. 
In particular, even if we only consider the deviation set that contains individual-user untruthful bids and the miner collusion with one user, it is impossible to make positive revenue in the complete-information setting. To address this issue, the paper by  \cite{shi} introduces a so-called $\gamma$-strict utility{, in which the parameter $\gamma$ roughly depicts the probability that a currently unconfirmed transaction would be confirmed in future blocks} (see details in Section~\ref{sec:related:strategy}). However, such relaxation involving confirming a transaction in future blocks, brings an additional layer of difficulty. Indeed,  the probability that a currently unconfirmed transaction gets confirmed in future blocks is not a universal constant, as unconfirmed transactions with higher bids are more likely to be confirmed in the future than those with lower bids. Thus, finding an accurate $\gamma$ can be difficult, if not impossible. Therefore, we would like to ask the following question in this paper.

\begin{center}
\begin{framed}
{\it
Are there other reasonable relaxations of the model and incentive compatibility specifications to circumvent the impossibility result?
}
\end{framed}
\end{center}


For this question, a series of related works (e.g., \cite{yotam,cryptoeprint:2022/1294,wu2023maximizing}) have studied the problem of revenue optimization for blockchain transaction fee mechanisms with different models. Although it might be argued that the absence of miner revenue, or burning of \emph{money}, may not directly undermine \emph{global} social welfare \citep{damle2024no}, a non-zero additional reward for transaction confirmation besides the static block reward would indeed incentivize the miner to create the blocks honestly, in order to earn more money {in addition to} the basic block reward, {especially for blockchains like Bitcoin in which the block rewards gradually go to zero. 
Besides, while the burning of tokens is not necessarily value-destroying, excessive burning may lead to a problem of \emph{deflation}. While an important motivation of Bitcoin is to prevent inflation, from an economic perspective,  deflation could be even worse than inflation in the long term because it may discourage spending and investment \citep{bewley1979optimum,baig2003deflation}, as people would prefer holding onto their tokens rather than using them in transactions. To maintain a thriving ecosystem in the blockchain community, we are indeed motivated to increase miner revenue via a decreased level of burning.}

In our study, we address the above open question by relaxing the U-DSIC requirement to U-BNIC, which assumes the users only have information of distributions of other users' valuations instead of all their bids. This relaxation is reasonable because,  in the distributed network of blockchain, it is impossible for a user to actually know all other users' bids, especially those who propose transactions after them but compete for the same block. Besides, a blockchain system, when combined with a commitment scheme (as discussed in Appendix~\ref{sec:crypto}), can essentially work as a \emph{sealed-bid auction}\footnote{In na\"ive implementations users may be able to see bids proposed before them but not after, leading to certain \emph{unfairness} and MEV issues. With a commitment scheme (see Appendix~\ref{sec:crypto}), we make it fair as all bids are sealed until the bidding process completes.} in which users' bids are not revealed until the bidding process finishes. Hence, users only have distributional knowledge about others, making the U-BNIC a natural requirement to prevent users' deviation. 
With the awareness that the MIC property is not the most necessary requirement (Remark~\ref{remark:no:mic} in Appendix A), the main goal of the paper is to design a TFM that satisfies U-BNIC and 1-SCP for bounded \emph{i.i.d.}\ valuation distributions with a constant-factor approximation of the optimal revenue.

Interestingly, besides bypassing the negative result in blockchain TFMs via Bayesian mechanism design, on the other hand, our paper also bypasses a major negative result in the scope of Bayesian mechanism design via \emph{burning}, a feature in blockchain systems. The papers by \cite{manelli2010bayesian, gershkov2013equivalence} show that in conventional auction settings in which all bidders' payments are rewarded to the auctioneer, the BNIC and DSIC conditions are equivalent. However, our research shows that with the incorporation of burning, the additional freedom to allow partial payment to be rewarded to the miner actually makes it possible to design essentially different mechanisms and gain increased revenue via the relaxation from DSIC to BNIC. (see discussion in Section~\ref{sec:bmd})

In the rest of this paper, we assume the valuation distributions of users are \emph{i.i.d.} and bounded. Without loss of generality, we assume the valuations are in the range of $[0,1]$.


\ifdefined\EnableBurning
We also look into the issue of \emph{burning}, which means that some of the transaction fees are removed from the system instead of paid to the miner. The existence of burning leads to economical inefficiency of the blockchain system. On the other hand,  the paper by \cite{roughgarden2020transaction} has shown that the ``burning'' is essential in the Ethereum EIP-1559 mechanism to avoid collusion. As the mechanism proposed by \cite{shi} also incorporates burning in their design, we further ask the following  question: 
\begin{center}
\begin{framed}
{\it 
Is burning \emph{generally necessary} for the collusion-proof TFM design? 
}
\end{framed}
\end{center}
For this question, we give a conditional positive answer in the scope of U-BNIC and 1-SCP mechanisms. In particular, under mild assumptions, we show that any non-trivial TFM that satisfies U-BNIC and 1-SCP and has constant block size must burn a positive amount of the collected fees.
\fi

\subsection{Summary of Contributions}
\label{sec:contri}

Following the previous discussion, we summarize the main contribution of the paper below.

\begin{enumerate}
    \item We propose an {\emph{auxiliary mechanism method}} as our main tool to study U-BNIC TFMs by establishing connections between BNIC and DSIC auction mechanisms. 
    In general, the method enables us to decompose any TFM into an \emph{auxiliary} U-DSIC mechanism and a \emph{variation term} and design them separately. This method can be a versatile tool in the design of BNIC mechanisms in the paradigm of {{relaxing the DSIC condition to BNIC and utilizing the information asymmetry for higher revenue (or other desired properties, e.g. welfare)}}. The auxiliary mechanism method will be described in Section \ref{sec:roadmap}.

    \item For ease of illustrating our main idea, we first study the case where each block only contains one transaction (i.e., the block size is one). To design the TFM, we first construct a so-called \emph{soft second-price mechanism} as our auxiliary mechanism, also referred to as \emph{exponential mechanism}, based on the logit choice model.  
    Via the auxiliary mechanism method, we design our mechanism that exploits the maximum extent of the information asymmetry between the miner and users to extract maximum revenue for the miner (Section~\ref{sec:alg:one_size}).
    

    \item We further extend our mechanism to general block size $k$, and prove that the constant-fraction approximation of optimal revenue still holds in this general case as long as the number of users $n$ is larger than $\lambda_0 k$ for any fixed $\lambda_0 >\frac{e}{e-1}$ 
    (see Section~\ref{section:size:k}).
    
    \item We further explore new properties of miner incentives in our TFM, for both size 1 and general block size $k$. 
    Our results show that a reasonable level of miner deviations would not substantially benefit the miner, even if the miner knows all the bids. 
    We also show a negative result that any TFM that is U-BNIC, 1-SCP and (strict) MIC cannot have a positive \emph{expected} miner revenue (see Section~\ref{subsec:almost:mic}). Furthermore, we 
    establish a key stability result on the miner's revenue from our TFM over the distribution of users' bids. This result is important in practice, as the stability of revenue is critically important for miners (see Section~\ref{subsec:revenue:stability}).
    \ifdefined\EnableBurning
    \item \com{TODO: 4 not updated yet} We look into the necessity of burning, showing an impossible result that under mild assumptions, no non-trivial Transaction Fee Mechanism can achieve both of the following properties (Section~\ref{sec:burning}): 
    
    \begin{itemize}
        \item U-BNIC and 1-SCP (truthfulness \& collusion-proofness guarantees),
        \item No burning (no waste on transaction fee).
    \end{itemize}
    \fi
\end{enumerate}

As we described in the paragraph Section \ref{sec:research_q}, we provide a classification of the dishonest behaviors among different possible agents (or groups of agents). Here, we summarize this classification in Table \ref{table:dishonst}. Based on this classification, for different classes of deviations, we compare the strategy-proof properties and miner revenue of our proposed mechanism with different designs in the existing literature. The detailed comparison is provided in  Table~\ref{table:comparison}.  In this paper (as well as \citep{shi}) we are particularly interested in preventing dishonest behaviors \UUB, \UFT, \MUUB{c}, \MFT, and \MTD. These dishonest behaviors are respectively prevented by the above-mentioned strategy-proof properties U-BNIC, U-SP, $c$-SCP, and MIC (where U-BNIC, U-SP and $c$-SCP will be formally defined in Section~\ref{subsec:strategy:proof} and MIC will be formally defined in Section~\ref{subsec:almost:mic}). In contrast, the \UUB~dishonest behavior is dealt with U-DSIC in \citep{shi} under the complete-information setting; {besides, a recent work by \citep{wu2023maximizing} considers a different multi-party-computation (MPC) model and develops another variant of posted-price mechanisms with comparable incentive and revenue guarantees as our work, but via different methodologies. We would discuss on the comparison in Section~\ref{sec:posted:price}.}

We finally note that this paper only considers the case that the miner colludes with one user (i.e., the case $c=1$ in {\tt MU$^c$-UB}) and it would be interesting to study for general $c$. We would like to leave it as a future work. 



\begin{table}[tb]
    \centering
    \scalebox{0.9}{
    \begin{tabular}[htb]{|c|c|c|c|}
        \hline
          & Untruthful Bids  ({\tt UB}) & Fake Transactions ({\tt FT}) & Transaction Deletion ({\tt TD})\\
        \hline
        Individual User ({\tt U}) & {\tt U-UB} & {\tt U-FT} & --- \\
        \hline
        Individual Miner ({\tt M}) & --- & {\tt M-FT} & {\tt M-TD} \\
        \hline
        Miner-$c$-User Collusion ({\tt MU$^c$}, $c \ge 1$) & {\tt MU$^c$-UB} & {\tt MU$^c$-FT} & {\tt MU$^c$-TD} \\
        \hline
        $c$-User Collusion ({\tt U$^c$}, $c\ge 2$) & {\tt U$^c$-UB} & {\tt U$^c$-FT} & --- \\
        \hline
    \end{tabular}
    }
    \caption{Classification of Dishonest Behaviors.}
    \label{table:dishonst}
\end{table}

\begin{small}
\begin{table}[htb]
\centering
\scalebox{0.9}{
\begin{tabular}{|c|c|c|c|c|c|}
\hline
              & Setting                 & \UUB, \UFT                                & \MUUB{c}                    & \MFT, \MTD          & Revenue$/ OPT$                             \\ \hline
Our Mechanism & Bayesian game           & \checkmark  & $c=1$                     & Approx.$^{*}$              & $\Theta(1)$                 \\ \hline
\cite{shi}  & $\gamma$-strict utility & \checkmark  & \checkmark & \checkmark & $\approx O(\gamma ^2 / c)$ \\ \hline
{\cite{wu2023maximizing}}  & {MPC model} & {\checkmark}  & {$c=1$, Fixed$^{**}$, Approx.} & {Approx.} & {$\Theta(1)$} \\ \hline
Bitcoin       & N/A                     & $\times$                                      & $\times$                  & \checkmark & No analysis                          \\ \hline
EIP-1559      & Deterministic           & Approx.$^{***}$                          & \checkmark & \checkmark & $\approx 0$                                    \\ \hline
\end{tabular}}
\caption{Comparison of different TFMs. $^*$ See detailed definition in Section~\ref{subsec:almost:mic}. 
{$^{**}$ Their 1-SCP notion is substantially weaker as they only allow the miner to collude with a \emph{fixed} user.}  $^{***}$ This property is guaranteed only when the ``posted price'' is well set to prevent congestion (see the paper of \cite{roughgarden2020transaction}.)}
\label{table:comparison}
\end{table}
\end{small}

In summary, the multi-item first-price auction mechanism adopted by Bitcoin has bad strategy-proof properties, although it expects to have good miner revenue because it charges and awards transaction fees in a ``greedy'' way. The EIP-1559 mechanism \citep{roughgarden2020transaction} has almost the best strategy-proof properties but has zero miner revenue. The mechanism proposed by \cite{shi} can prevent general \MUUB{c}, but it only has good miner revenue for small $c$ and large $\gamma$ (meaning that every transaction has a high probability to be eventually confirmed in the future even if the bid is low, which is different from the common practice in the blockchain). Our mechanism uses the Bayesian setting and has decent strategy-proof properties while achieving a constant-fraction approximation of the optimal miner revenue.

We list the meaning of abbreviations appearing in our paper in Table~\ref{table:abbr}.

\begin{table}[!h]
\centering

\begin{tabular}{|c|c|}
\hline
Abbreviation & Meaning                                          \\ \hline
TFM          & Transaction Fee Mechanism                        \\ \hline
(U-)DSIC     & (User) Dominant Strategy Incentive Compatibility \\ \hline
(U-)BNIC     & (User) Bayesian Nash Incentive Compatibility     \\ \hline
U-SP         & User Sybil Proofness                             \\ \hline
$c$-SCP      & $c$-Side Contract Proofness                      \\ \hline
MIC          & Miner Incentive Compatibility                    \\ \hline
OCA          & Off-Chain Agreement                              \\ \hline
BF           & Budget Feasibility                               \\ \hline
NFL          & No Free Lunch                                    \\ \hline
UIR          & User Individual Rationality                      \\ \hline
MIR          & Miner Individual Rationality                     \\ 
\hline
LP          & Linear Programming                     \\ \hline
MPC          & Multi-Party Computation                          \\ \hline
MEV          & Miner/Maximal Extractable Value                  \\ \hline
\end{tabular}
\caption{List of abbreviations.}
\label{table:abbr}
\color{black}
\end{table}

\section{Related Work}
\label{sec:related}

\subsection{{Auction-Like} TFM Design in Literature} \label{sec:related:strategy}

{Since the main motivation of TFM design is to allocate the scarce block space to users, an intuitive idea is to design auction-like TFMs.} While it is common to assume that the auctioneer in a traditional auction is \emph{trusted}, it is not true in blockchain systems. To address this new challenge, the papers of \cite{roughgarden2021transaction} and \cite{shi} split the incentive compatibility into two parts: User Incentive Compatibility (UIC) and Miner Incentive Compatibility (MIC), and 
consider the \emph{complete-information setting} in which users have complete information of others' bids. 
%
Essentially, their papers define the term UIC equivalent to \{\UUB, \UFT \}-proofness, MIC equivalent to \{\MFT, \MTD \}-proofness (see Table \ref{table:dishonst}). 
Furthermore, the paper by \cite{shi} specifies the notion of $c$-SCP, which is equivalent to \{\MUUB{c}\}-proofness. Essentially, \cite{shi} show a seminal impossibility result as follows: 

\begin{theorem}[\cite{shi}] \label{thm:shi:impossible}
    Any TFM which is \{\UUB, \MUUB{1}\}-proof in the complete-information (a.k.a. deterministic) setting has zero miner revenue. 
\end{theorem}

Note that the original theorem of \cite{shi} states that ``any TFM which satisfies UIC (\{\UUB, \UFT\}-proof) and 1-SCP (\{\MUUB{1}\}-proof) has zero miner revenue.'' However in their proof, they only consider the deviations of \UUB~and \MUUB{1} but not \UFT. So as we consider the deviations in a more refined way, they actually prove this slightly stronger impossibility result.

To overcome the issue of zero miner revenue, \cite{shi}  introduce the  ``$\gamma$-strict utility'' to make unconfirmed over-bidder still pay a $\gamma$ fraction of the worst-case cost. In particular,  if a bidder $i$ has valuation $v_i$ and her bid $b_i>v_i$, even if the transaction is not confirmed, she gets a utility of $-\gamma(b_i-v_i)$. 
Thus, if the confirmation probability is $a_i(b_i,\b_{-i})$ ($\b_{-i}$ denotes bids of all other users than $i$) and the bidder pays $p_i(b_i,\b_{-i})$ if the transaction gets confirmed, the utility of the bidder takes the following form:
\begin{small}
\begin{align*}
    u_i^{(\gamma)}(b_i,\b_{-i};v_i) &= a_i(b_i,\b_{-i})(v_i-p_i(b_i,\b_{-i})) \\
    &\quad- \gamma (1-a_i(b_i,\b_{-i})) \max\{b_i-v_i, 0\}.
\end{align*}
\end{small}

This relaxed utility function is justified by \cite{shi} by considering the bidding process of more than one block in a blockchain: even if an overbidding transaction is not confirmed in the current block, the authors assume that the bid could still be collected and confirmed into future blocks, and in the worst case, the over-bidder would have to pay their full bid and get a utility of $-(b_i-v_i)$. In this setting, the authors have further developed a \emph{burning second price} TFM that satisfies U-DSIC, MIC and $c$-SCP in the notion of $\gamma$-strict utility.

The $\gamma$-strict utility that considers the multi-block setting is critically sensitive to the parameter $\gamma$, but in practice, it is difficult to determine the value of $\gamma$ due to the unpredictable nature of future blocks. On the other hand, as users cannot see others' bids in a blockchain system, the requirement of a complete-information setting is too strong in practice as compared to the Bayesian setting. In this perspective, our research focuses on the single-block setting in which a proposed transaction is only valid for the current block, but we assume that each user only knows the distributions of other users' valuations. Based on the distributional information, we consider a different relaxation and develop our mechanism in the Bayesian game setting.

Besides strict notions of incentive compatibility, previous research on transaction fee design also considers nearly-incentive-compatibility properties. \cite{yao2018incentive} shows that the \emph{monopolistic price} mechanism proposed by \cite{lavi2022redesigning} is nearly incentive compatible, i.e., strategic behavior can only gain a small advantage in utility. On the other hand, there is another parallel paradigm of collusion-proofness named as OCA-proofness, as proposed by \cite{roughgarden2021transaction}, which only considers collusion of \emph{confirmed} users instead of all users and has different properties. The detailed difference has been discussed by \cite{yotam} and \cite{chung2024collusion}.

{\textbf{Collusion-proofness notions: OCA-proofness and SCP.} }
{In a recent work \citep{chung2024collusion}, the authors study the relations between the two notions. As argued by \cite{chung2024collusion}, the OCA-proofness notion conceptually depicts the property that the colluding parties cannot ``steal'' from the protocol to gain more utility, and SCP means they can neither ``steal'' from the protocol nor from other users. It has also been shown in \citep{chung2024collusion} that OCA-proofness is implied by $c$-SCP for any $c$, i.e., $\infty$-SCP is a stronger notion than OCA-proofness. Nevertheless, \cite{chung2024collusion} also show that 1-SCP and OCA-proofness are incomparable notions.}

{From an economic perspective, it is not necessarily detrimental if colluding parties ``steal'' from the protocol to gain more utility. If offchain payments ensure that every agent experiences a weak increase in utility while maintaining protocol functionality, this collusion can actually represent a \emph{Pareto improvement} \citep{pareto1919manuale}. While a Pareto improvement appears to be a desirable improvement of the ecosystem that does not need to be prevented, in a non-OCA-proof mechanism, the existence of Pareto improvements indicates that the mechanism is not economically efficient. Therefore, while the concept of SCP aligns with the principle of strategy-proofness, OCA-proofness is more relevant to the idea of \emph{Pareto optimality}. Considering that ``stealing from others'' is clearly a dishonest behavior we aim to prevent, the study of SCP remains crucial.}

\subsection{EIP-1559 and Posted-Price TFM Design in Literature}
\label{sec:posted:price}

Due to the anonymity of blockchain systems, blockchain TFMs are subject to a wider scope of dishonest behavior than traditional auctions, e.g., collusions and fake identities. While auction-like mechanisms can balance the supplies and demands as the prices are decided by the users' bids, the parties indeed have access to more strategies to manipulate the prices and gain advantages via dishonest bidding. In response to this challenge, there are another line of studies that replaces auctions in TFM design with a widely studied toolbox of \emph{(dynamic) optimal pricing}, in which the ``posted'' prices are not decided by users of the current block, but from the statistics of previous blocks \citep{roughgarden2020transaction,ferreira2021dynamic,wu2023maximizing}. With the purpose of dynamically adjusting price based on supplies and demands, auction mechanisms and dynamic posted-price mechanisms are indeed solutions with different paradigms that both have the potential to be utilized in TFM design, as discussed in the studies of \citet{hammond2010comparing,bubeck2017online}, and so on.

Particularly, the EIP-1559 TFM, which is currently adopted in Ethereum, is essentially designed to be a posted-price mechanism that effectively prevents dishonest behavior, with a backup component of an auction-like mechanism in case the posted price is (unexpectedly) too low to prevent congestion.
The EIP-1559 TFM works as follows:

\begin{enumerate}
\item The blockchain system adaptively decides on a \emph{base fee} for the current block.
\item Each user proposes a transaction, paying the base fee and a voluntary \emph{tip} if the transaction gets confirmed.
\item The (winning) miner confirms the transactions proposed by the users. It is expected that the number of transactions usually does not exceed the block size.
\item The miner gets a pre-defined \emph{block reward} as well as all the tips. The base fees are \emph{burned} and removed from the blockchain system. 
\end{enumerate}

In the ecosystem of EIP-1559, when there is no congestion, the miner does not need to consider the block size constraint and can confirm all the proposed transactions, so the users would pay very small tips and it is enough to incentivize the miner to confirm their transaction. At ``exception'' scenarios when there is congestion, the miner would be incentivized to confirm transactions with the highest tips, and all the tips go to the miner. Hence, the EIP-1559 mechanism can be modeled as follows:

\begin{itemize}
    \item If there is no congestion, EIP-1559 is essentially a posted-price mechanism with the posted price equal to the base fee; all payments are burnt and the miner gets no revenue from the transaction fees.
    \item If there is congestion, EIP-1559 shifts to a multi-item first-price (aka. pay-as-bid) auction. A fix amount of tokens (\emph{base fee}~$\cdot$~\emph{block size}) are burnt and the remaining tokens go to the miner.
\end{itemize}

{Due to the zero-revenue disadvantage of EIP-1559 and the existing impossibility results, researchers also attempt to avoid this negative aspect with alternative modeling and reasonable relaxations that apply to the blockchain environment. In consideration of the cryptographic nature of blockchain systems, \citet{cryptoeprint:2022/1294} introduce a multi-party-computation (MPC) model that achieves $\epsilon$-approximate incentive properties with a positive miner revenue scaling with $\Theta(\sqrt{\epsilon})$. Following the MPC model, \citet{wu2023maximizing} developed an LP-based posted-price mechanism that achieves U-DSIC, approximate Bayesian MIC and approximate Bayesian $1$-SCP with positive miner revenue. The collusion-proof properties in the MPC-based models differ from ours as follows: in our work, we assume the miner to have access to all users' valuations, and may pick any $c$ user(s) to collude with; in the MPC-based models described in \citep{cryptoeprint:2022/1294,wu2023maximizing}, the miner only has access to the valuations of $c$ \emph{colluding} users, so their $c$-SCP notion is weaker than our paper. Particularly, as their study has shown the impossibility even for $c=2$, their method is only applicable for the case that the miner may only collude with a \emph{fixed} user, which is highly restrictive for real-world blockchain systems.}

{
\textbf{Comparison between our mechanism and \citep{wu2023maximizing}.} The paper of \citet{wu2023maximizing} shows that even in the MPC-based model $2$-SCP is impossible. From the above explanation, their approximate Bayesian $1$-SCP property is weaker than our work, but their U-DSIC property is stronger than ours. For the sybil-proofness properties, their mechanism also upper bounds the number of fake bids to $h$, comparable to our approximate-SP notions. For the methodologies to secure incentive guarantees, similar to EIP-1559, the mechanism in \citep{wu2023maximizing} uses a variant of posted-price mechanism in which confirmed users' payments are fixed, and the miner's utility is calculated by a \emph{linear program} and only depends on the \emph{number} of \emph{candidate} users with valuations above the posted price. Their mechanism thus prevents the \UB~strategy by essentially avoiding the bidding process, while our mechanism is in an auction-like form with the \emph{auxiliary mechanism method} to ensure that honest bids achieve optimal expected utilities. In general, the paper of \citep{wu2023maximizing} adopts different modeling and methodology in mechanism design, while achieving a comparable level of incentive and revenue guarantees to our work. The diversity in models and methodologies renders the topic of TFM design a novel and valuable area for future exploration in the OR community.}

\color{black}

\subsection{Choice Modeling and the Multinomial Logit (MNL) Choice Model}

Choice modeling, which models how consumers would make choices among provided goods, plays an important role in revenue management \citep{fiamohe2015assessing, bitran2003overview}. Indeed, assortment optimization under a wide range of choice models has been extensively studied in the operations literature. The most popular choice model is the multinomial logit model (MNL) (see, e.g., \cite{Ryzin1999,Mahajan2001,Liu2008,rusmevichientong2010dynamic}). Other choice models, such as nested logit models \citep{davis2014assortment, Li2014}, non-parametric choice model \citep{Farias2013}, Markov chain choice model \citep{Blanchet2016,markchoice3}, have been studied under the problem of assortment optimization. 


Instead of using the standard auction (e.g., first-price auction) to assign the winning bidder deterministically, the MNL model provides a randomized way to select the winning bidder. In particular, for a given set of alternatives, assuming each choice $j$ has an expected utility $u_i(j)$ for agent $i$. In other words,  the agent $i$ perceives a value $\hat{u}_i(j)=u_i(j) + e_i(j)$ for the choice $j$, where  $e_i(j)$ is a random variable of perception error with $\mathbb{E}[e_i(j)] = 0$. A standard MNL model \citep{Train2009} assumes that all $e_i(j)$'s are \emph{i.i.d.} Gumbel distribution, and then the choice probability takes the following form:
$ \Pr[i \textup{ chooses } j] = \Pr[j \in \arg \max_{k} \hat{u}_i(k)] 
    = \frac{e^{m \cdot u_i(j)}}{\sum_{k} e^{m \cdot u_i(k)}}$.

As argued by \cite{shi}, randomness in choosing the winning user is necessary for a TFM to guarantee the collusion-proofness property, which we also prove in Appendix~\ref{sec:imp} even for the Bayesian setting. Intuitively, it is more profitable for a miner to collude with a user who deterministically gets her transaction confirmed than a user who only has a certain chance. While the burning second-price mechanism in \cite{shi} gives each user who bids high enough a pre-set probability to get confirmed, we consider a more natural idea of randomization by leveraging the logit choice model into the allocation rule (see Section~\ref{sec:size1:M}). In this paradigm, we can prioritize high-bidding users in a more natural and smooth way. The MNL-based allocation rule also fits well into our auxiliary mechanism method (Section~\ref{sec:roadmap}) and yields a constant-factor expected revenue compared to the optimal revenue.

In addition, in different fields of mechanism design, the MNL choice model is also adopted for other purposes. For example, \cite{huang2012exponential} utilize a similar mechanism to achieve differential privacy requirements, and the mechanism is also called \emph{exponential mechanism} in their work. Not surprisingly, they also replace the allocation rule of ``the highest-bidder gets the item'' with a soft-max relaxation while preserving the near-optimal property.

\subsection{Bayesian Mechanism Design}\label{sec:bmd}

From the famous \emph{revelation principle} \citep{myerson1981optimal, myerson1979incentive}, it is desirable to design incentive-compatible mechanisms. While the strongest notion of DSIC guarantees agents to report true types even if they have the complete information of other agents, this requirement could be a bit too restrictive in blockchain systems as users might not have such sufficient information about others. 

In the Bayesian game setting, we assume that the distribution of agents' types is known by the public. At the beginning of the game, each agent's type is assigned by nature following the corresponding distribution. Then, each agent only knows her own type and the distribution of others' types conditioned on her type, and seeks to maximize her expected utility. In this scope, a mechanism is BNIC if everyone truthfully reporting their true types forms a Bayesian Nash equilibrium.

While the design of BNIC mechanisms is less restrictive than DSIC mechanisms, the \emph{revenue equivalence} theorem \citep{myerson1981optimal} shows that the Bayesian game setting cannot gain extra revenue in conventional auctions when users have \emph{i.i.d.} valuation distributions, which is also the basis of a series of ``equivalence'' results between conventional BNIC and DSIC auctions, e.g. \citep{eso1999auction}. Furthermore, \cite{manelli2010bayesian,gershkov2013equivalence} show that the BNIC and DSIC conditions are equivalent in the conventional auction setting without the involvement of burning and collusion-proofness requirements. In the scope of blockchain transaction fee mechanism design, however, the existence of \emph{burning} allows partial ``revenue'' (total payment from users) to be rewarded to the miner while still keeping ex-post budget feasibility. Hence, while the revenue equivalence theorem \citep{myerson1981optimal} dictates that the Bayesian game setting cannot increase the total user payment, our research shows that a TFM can indeed increase the miner revenue with a decreased level of burning.

Two recent works \citep{yotam, cryptoeprint:2022/1294} on blockchain transaction fee mechanism design also consider the Bayesian setting. In particular, \cite{yotam} argue that a simple variation of the first-price auction can simultaneously satisfy U-BNIC and another collusion-proofness named \emph{OCA-proofness}, and \cite{cryptoeprint:2022/1294} show that if we relax all the incentive conditions (to almost U-BNIC, almost interim MIC and almost interim $c$-SCP), it is also possible to achieve positive revenue.  However, our result is different from \citep{cryptoeprint:2022/1294}. We have \emph{ex-post} almost MIC and SCP guarantees, which are crucial as discussed in Remark~\ref{remark:expost} in Appendix~\ref{sec:crypto}; their work in turn considers a different MPC-assisted setting{, but a weaker notion of SCP (as discussed in Section~\ref{sec:posted:price})}. Besides, the paper by \cite{yotam} is incomparable to ours as the OCA-proofness is fundamentally a different model from our work.

Similar to the currently used EIP-1559 TFM of Ethereum \citep{roughgarden2020transaction}, and papers of \cite{yotam} and \cite{cryptoeprint:2022/1294}, our work also uses a \emph{prior-dependent} mechanism that requires a parameter $c_\rho$ (as defined in Section~\ref{sec:alg:one_size}) that depends on the prior distribution. While blockchain mechanisms are usually hard-coded into the system, the distributional parameter can still be implemented to update adaptively based on historical data, similar to the base fee in EIP-1559. As shown by \cite{maheshwari2022inducing}, adaptiveness is indeed crucial in the development of social optimality in large-scale network mechanisms in the presence of selfish agents.

\section{Preliminaries}
\label{sec:prelim}

\subsection{Overview and Classification of Dishonest Behavior} \label{subsec:classif:dishonest}

In the blockchain system, either a user or miner may \emph{deviate} from the supposed behavior and behave dishonestly. First of all, as the blockchain system is anonymous, either type of agent may conduct \emph{sybil attack} \citep{zhang2014sybil} by creating multiple fake identities to influence the performance of the system. However, in the blockchain system, the PoW or PoS mechanism makes it costly to create a fake identity as a (winning) miner, so we mainly consider creating fake user identities, i.e. injecting fake bids, as also mentioned by \cite{roughgarden2021transaction} and \cite{shi}. On the other hand, even if every agent takes their true identities, they may do their jobs dishonestly: the users may bid differently from their true valuation, and the miner may purposely ignore some bids. Furthermore, the miner and users may also collude, i.e. conduct such dishonest behavior as a party in seek of increasing their total utility.  

In summary, dishonest behavior (deviations) can include untruthful bidding, fake identities and dishonest confirmation. Therefore, the deviations can be classified into the following three types: 
\begin{itemize}
    \item Untruthful Bids: proposing a bid different from the true valuation.
    \item Fake Transactions (Sybil attack): injecting fake transactions.
    \item Transaction Deletion: ignoring certain transactions proposed by users.
\end{itemize}

The dishonest behavior can be conducted by the miner, a user, or a colluding party of them. A user, or a colluding party of multiple users, can make untruthful bids and inject fake transactions, and the miner can inject fake transactions and delete existing transactions. A colluding party that consists of the miner and users, can do all these three deviations. Therefore, there can be 9 types of deviations in the system, as shown in Table~\ref{table:dishonst}.


We say a dishonest action is \emph{profitable} when it strictly increases the total utility of all agents participating in it. Precisely,

\begin{itemize}
    \item If it is an individual deviation, the agent strictly increases her utility via that deviation.
    \item If it is a collusion, the colluding party strictly increases its total utility via that collusion.
\end{itemize}

In this sense, for a strategy space $\mathcal{S}$, we say a TFM is $\mathcal{S}$-proof if all deviations in $\mathcal{S}$ are not profitable in this TFM.

\textbf{Information sets of agents.}
Before specifying the strategy-proof conditions, we need to discuss the \emph{information sets} of agents. In traditional auctions, the notion of \emph{incentive compatibility} means that the mechanism would optimize any bidder's utility when they bid their true valuations. On the other hand, the subtle meaning of \emph{incentive compatibility}  also depends on \emph{{what the bidders know}}. 
In this scope, the strongest notion is {Dominant Strategy Incentive Compatibility} (DSIC), which means that it optimizes any individual bidder's utility even if they know all others' bids, i.e. they have the complete information. Nevertheless, in sealed-bid auctions, bidders would not actually know what other bids, so the \emph{complete-information  (deterministic) setting} may be too strong. Nevertheless, we may still assume that they can perceive the \emph{distributions} of others' bids. In this so-called \emph{Bayesian-game setting}, if the mechanism can guarantee that any bidder would maximize their \emph{expected} utility via bidding their true valuation, we call the property Bayesian-Nash Incentive Compatibility (BNIC). 

In our paper, the bidding is conducted by users, so we name the DSIC and BNIC properties for users as U-DSIC and U-BNIC, respectively. The formal definitions are in Section~\ref{subsec:strategy:proof}.

\subsection{The Basic Model}\label{sec:basic-model}

There are $n$ users numbered by $1,2, \dots, n$ and each user proposes a transaction to compete for a block. There is also a winning miner owning the block. The block has size $k$, the maximum number of transactions it can confirm. For user $i$, w.l.o.g.\ we assume her valuation $v_i$ is in $[0,1]$ and drawn from an \emph{i.i.d.}\ distribution $V_i=V_0$ with pdf $\rho_i(\cdot)=\rho(\cdot)$. We let $V = V_1 \times V_2 \times \dots \times V_n$ be the distribution of the valuation vector $\v = (v_1, v_2, \dots, v_n)$.

By the revelation principle (\citep{myerson1981optimal, myerson1979incentive}, see Appendix~\ref{sec:crypto}), we only need to consider direct mechanisms in which users propose bids, the miner collects the bids and the system decides which transactions to confirm and processes the payments. Formally, we can model any Transaction Fee Mechanism w.r.t.\ its allocation, payment and miner revenue rules, as follows.

\begin{definition}[Transaction Fee Mechanism] \label{def:TFM}
    For a fixed number $n$ of users, a Transaction Fee Mechanism is defined by $M(\a,\p,r)$, where 
    \begin{itemize}
        \item the \emph{allocation rule} $\a: [0,1]^n \to [0,1]^n$ maps the bid vector to the allocation vector indicating the probability each user's transaction to be confirmed;
        \item the \emph{payment rule} $\p: [0,1]^n \to \R^n$ maps the bid vector to the payment vector indicating the payment of a user \emph{if her transaction is confirmed};\footnote{This definition is different from some literature (where the ``payment rule'' indicates the expected payment of a user whether she gets confirmed or not, which can be transformed to $a_i(\cdot)p_i(\cdot)$ in our notation). Besides, our definition naturally guarantees that unconfirmed users do not pay transaction fees.}
        \item the \emph{miner revenue rule} $r: [0,1]^n \to \R$ maps the bid vector to the miner's revenue.
    \end{itemize}
\end{definition}

In the na\"ive implementation of transaction collection in blockchains, users propose transactions (including bids) publicly and sequentially in a mempool and miners pack them into blocks, so it acts as an auction format between open bidding and sealed bidding --- users can see bids submitted before them but not after them, which may lead to several issues (see Appendix~\ref{sec:crypto}). In our mechanism, we implement sealed bids via a \emph{commitment scheme} as described in Appendix~\ref{sec:crypto} so that no bid can be viewed by the miner or other users until all transactions that compete for the block are finalized. In the execution of the mechanism, the system essentially elicits users for their (sealed) bids $\{b_i\}$, draw
the confirmed transactions
according to probabilities from $\{a_i(b_i,\b_{-i})\}$ (for $k>1$, one follows the sampling method discussed in Section~\ref{sec:drawing:rule} to ensure exactly $k$ transactions are confirmed), and then charge transaction fees from confirmed bidders according to $\{p_i(b_i,\b_{-i})\}$ and give the miner revenue $r(\b)$ to the miner. Due to the size constraint, we need to guarantee $\sum_{i=1}^n a_i(b_i,\b_{-i}) \le k$. Since the transactions are naturally anonymous and unsorted, in this paper we only consider the mechanisms satisfying the following symmetric condition:

\begin{definition}[Symmetry]
A TFM is symmetric if the allocation and payment rules do not depend on the order of users, i.e. when we swap any pair of users, each should still have the same allocation probability and payment as in their original positions.
\end{definition}

In Definition~\ref{def:TFM}, we assume that (in the usual case) each of the $n$ users makes one bid, and therefore the allocation, payment, and miner revenue rules are functions of exactly $n$ bids. While this definition is enough for the discussion of the main strategy-proof properties (e.g., U-BNIC, $1$-SCP) concerned in this paper, we also expect the proposed TFMs to be strategy-proof against deviations such as the Sybil Attack and the deletion of user bids by the miner (namely, {\tt FT} and {\tt TD} in Table~\ref{table:dishonst}). These deviations could change the number of bids presented to the TFM, and we need to define the following \emph{variable-bid-size TFM} to deal with this technical issue.

\textsc{Definition~\ref{def:TFM}' (Variable-bid-size TFM).} A variable-bid-size TFM $\mathcal{M}(\a, \p, r)$ is similar to the regular TFM defined in Definition~\ref{def:TFM} where the only difference is that the allocation rule $\a : [0, 1]^* \to [0, 1]^*$,\footnote{For any set $A$, we use $A^* = \cup_{\ell=0}^{\infty} A^{\times \ell}$ to denote the set of all finite sequences where the elements are drawn from $A$.} the payment rule $\p : [0, 1]^* \to \R^*$ and the miner revenue rule $r : [0, 1]^* \to \R$ may take sequences of bids of any size, while the size of the output of $\a$ ($\p$) should be the same as the input, where each entry in the output is the allocation (the bid respectively) of the corresponding bid.

Throughout the paper when we refer to a TFM $M$, unless specially noted, we assume that $M$ is a regular TFM (Definition~\ref{def:TFM}). We will mostly focus on the design of a TFM for any fixed bid size. To construct a variable-bid-size TFM $\mathcal{M}$, we will first choose an $M_\eta = (\a_\eta, \p_\eta, r_\eta)$ for $\eta$ bids for each $\eta$ according to the regular TFM design, and then let $\mathcal{M} = (\a, \p, r) = \cup_\eta \{M_\eta\}$, or more concretely, for any bidding vector $\b$, let $|\b|$ denote the size of $\b$ and we set
\begin{align}\label{eq:regular-to-variable-size-TFM}
\a(\b) = \a_{|\b|}(\b), 
\p(\b) = \p_{|\b|}(\b), r(\b) = r_{|\b|}(\b).
\end{align}

Given a variable-bid-size TFM $\mathcal{M}$, for any fixed number of bids, namely $n$, there is a \emph{natural restriction} of $\mathcal{M}$ to a regular TFM, namely $M$. To derive $M$, we simply let its allocation, payment, and miner rules be the corresponding functions of $\mathcal{M}$ when restricted to inputs of size $n$.

\subsection{Incentive and Collusion-Proof Conditions}\label{subsec:strategy:proof}

We now discuss the desired properties we would like the mechanism to enjoy, i.e. the properties that agents would not gain additional utility via dishonest behavior (\UB, \FT, \TD).

As a basis, we note that the users have \emph{quasi-linear} utility: when user $i$ has valuation $v_i$ and the bidding vector is $(b_i, \b_{-i})$, user $i$'s (expected) utility is
\begin{align}\label{eq:def-u}
u_i(b_i,\b_{-i};v_i) = a_i(b_i,\b_{-i})\cdot(v_i - p_i(b_i, \b_{-i})) .    
\end{align} 

{For the miner's utility, in real-world blockchains, the miner's reward comes from the combination of the \emph{block reward} and the \emph{miner revenue} from the transaction fees, and the miner also pays a {mining cost} due to the computational consumption / token staking in PoW or PoS protocols, respectively. Since the block reward and the mining cost are fixed due to the blockchain protocol and is not affected by the transactions, for simplicity of expression, we just denote the miner's utility as the miner revenue provided by the TFM, i.e.,
\begin{equation}
    u^{(miner)}(\b) = r(\b).
\end{equation}
}

We now formally define U-DSIC and U-BNIC as follows.
\begin{definition}[User {Dominant-Strategy-Incentive-Compatibility} (U-DSIC)] \label{def:udsic}
For any user $i$, assuming the miner follows the inclusion rule truthfully, a TFM is U-DSIC if and only if it is a {dominant} strategy for any user 
to bid their valuations, i.e.
$    \forall \b_{-i}, v_i \in \arg\max_{b_i} [u_i(b_i,\b_{-i};v_i)]$. 
\end{definition}

\begin{definition}[User Bayesian-Nash-Incentive-Compatibility (U-BNIC)]  \label{def:ubnic}
Assume $\Omega$ is the type space of nature, and there is a public mapping $B:\Omega \to \R_{\ge 0} ^n$ that determines the valuation of all users. Thus, the valuation vector $\mathbf{v} \sim V = V_1 \times V_2 \times \dots \times V_n$ where each $V_i = V_0$ due to our model assumption.

For each user $i$, she only knows her own valuation and the distribution of other users' valuations conditioned on $v_i$, denoted as $V_{-i} = V|v_i$. A TFM is (interim) U-BNIC if and only if, when other users all bid their valuations, it maximizes user $i$'s expected utility if she bids her valuation too, i.e.
    $v_i \in \arg\max_{b_i}\mathbb{E}_{\b_{-i} \sim V_{-i}}[u_i(b_i,\b_{-i};v_i)]$. 
\end{definition}

We also characterize the user Sybil-proofness property, which guarantees that the user cannot increase via injecting fake bids, her expected utility over the distribution of other users. Formally,

{
\begin{definition}[User Sybil-Proofness (U-SP)]  \label{def:usp}
Assuming each user $i$ only knows her own valuation and the distribution of other users' valuations conditioned on $v_i$, denoted as $V_{-i}$, and assuming all users bid their true valuations. Then, we call a variable-bid-size TFM (interim) $(C,N)$-U-SP {for a fixed $(C,N)$} when $n>N$, user $i$ cannot increase her expected utility (over the distributions of other users' bids) via injecting $l \le Cn$ fake bids $\b^\# = \{b_{n+1}, \cdots, b_{n+l}\}$. Here we still denote $\b = (b_1,\cdots,b_{n})$, and define $\b^{+} = (b_1,\cdots,b_{n+l})$ containing all real and fake bids.

Notice that fake bids {generally} have a valuation of 0 because getting the fake transaction confirmed does not have any value for the user. 
{A possible exception is repeating the same transaction when the block size $k=1$, as at most one of them can be confirmed. However, in the general case where $k\ge 2$, the adverse consequence of having both transactions confirmed (e.g., paying twice for the same transaction) far outweighs the benefit of transaction confirmation. On the other hand, preventing this type of strategy when $k=1$ is impossible for the following reason: if all transactions have the same bid, the anonymity of the blockchain system ensures that any mechanism can only randomly select a transaction to confirm. Therefore, duplicating a transaction will always increase the likelihood of it being confirmed.
Since our study is mainly motivated by real-world blockchains with $k\ge 2$ block sizes, we exclude these types of strategies from consideration in our model.}

The utility of user $i$ is the total utility of $i$ herself and all fake bids. Therefore, a variable-bid-size TFM is $(C,N)$-U-SP if and only if for any valid $\b^+$, {if $b_i=v_i$, then
\begin{align*}
    &\mathbb{E}_{\b_{-i} \sim V_{-i}}[u_i(b_i,\b_{-i};v_i)] \\
    \ge &\mathbb{E}_{\b_{-i} \sim V_{-i}}\left[u_i(b_i,\b^+_{-i};v_i) + \sum_{j=n+1}^{n+l} u_j(b_{j},\b^+_{-j};0)\right]. 
\end{align*}}
As a shorter notion, we call a variable-bid-size TFM U-SP when there exist constants $C>0, N>0$ such that the TFM is $(C, N)$-U-SP.

\end{definition}
}

To describe the collusion-proofness, we use the notation of $c$-SCP in \cite{shi}, defined as:
\begin{definition}[$c$-Side-Contract-Proofness ($c$-SCP)] \label{def:cscp}
We call a TFM (ex-post) $c$-SCP when is impossible for the miner to collude with at most $c$ users to strictly increase their total utility when other users bid according to the Bayesian Nash equilibrium, even if the miner knows all users' valuations and bids.
\end{definition}

Note that successful collusion only needs the party to have increased total utility rather than individual utilities, because the members can make payments among themselves to make everyone get increased utility.

\textbf{Collusion-proofness notions for non-truthful TFMs.} It might be tricky to define the collusion-proofness for a TFM that does not satisfy U-BNIC. For example, \cite{shi} claimed that the first-price auction mechanism is $c$-SCP in the sense that $c$ users and the miner could not increase their joint utility via any deviation \emph{when all other users bid truthfully}. However, in the first-price auction, the users may not report their real valuation, leading to a different Bayesian Nash equilibrium. 

In our Bayesian setting, we assume that users who do not conduct the collusion would bid according to the Bayesian Nash equilibrium to maximize their expected utility. We show in {Appendix~\ref{sss:aux:fpa}} that \emph{at a Bayesian Nash equilibrium}, a user and the miner could increase their joint utility via deviation, even if burning is allowed. Therefore, we state that the first-price auction is not even 1-SCP in our notion.

Nevertheless, for TFMs that satisfy U-BNIC, assuming non-colluding users to bid truthfully or to bid as the Bayesian Nash equilibrium does not make a difference.

\textbf{Miner-Incentive-Compatibility.}
As is mentioned in the papers of \cite{shi} and \cite{roughgarden2021transaction}, the Miner-Incentive-Compatibility (MIC) is the property that assuming the users bid truthfully, the miner could not increase her utility via deviations from the inclusion rule (i.e. \MFT~and \MTD). It will be rigorously defined in Section~\ref{subsec:almost:mic} (Definition~\ref{def:mic}).

\subsection{Rationality and Feasibility Requirements}\label{sec:irbf}

Besides truthfulness and collusion-proofness, the mechanism also needs to satisfy more general properties, e.g. the balance must be feasible, the users should not pay more than their bid, etc. Formally, the following properties should also be satisfied.

\noindent {\textbf{(Ex-post) \textbf{User Individually Rationality (UIR).}} Each user gets non-negative utility when truthful bidding, no matter how others bid, i.e. ${\forall \b_{-i}},$
$    u_i(v_i,\b_{-i};v_i) \ge 0$.
Equivalently,  $a_i(v_i,\b_{-i})>0 \implies p_i(v_i,\b_{-i}) \le v_i$.

\noindent{\textbf{(Ex-post) Budget Feasibility (BF).}} {For all bidding vector $\b$, the miner's revenue $r(b_i,\b_{-i})$ should not be greater than the total user payment:} 
\begin{align}\label{eqn:total_payment}
P(\b) = \sum_{i=1}^n a_i(b_i,\b_{-i})\cdot p_i(b_i,\b_{-i}).
\end{align}
In other words, we should have {$\forall \b$,} $P(\b) \ge r(\b)$.
    
Here we allow the miner revenue to be \emph{less} than the total fee paid by users, in which the difference will be \emph{burnt}. The burning can decouple payments on miners' and users' sides, which is an effective way to broaden the design space and allow additional strategy-proof properties to be satisfied \citep{zishuo2022dynamic}. Actually, the burning has been used in the EIP-1559 TFM of Ethereum.

Additionally, while we expect the transaction \emph{fee} a user pays should be non-negative, it is okay as long as the users have a non-negative \emph{expected} payment to prevent users from submitting transactions to gain money out of nothing, which is guaranteed in our mechanisms for both $k=1$ and general $k$ (as (\ref{eqn:bnic:suff:cond}) is satisfied). However, as UIR requires the payment of the zero-valuation user to be no greater than zero,  the payment of the zero-valuation user is always zero (rather than negative). Therefore, we need the following NFL condition.

\noindent\textbf{No-Free-Lunch (NFL).} We call a TFM $(\a,\tilde{\p},\tilde{r})$ NFL when the payment of a zero-bidding user is always zero {no matter how other users bid}, i.e.,
\begin{align}
    a_i(0,\b_{-i})\tilde{p}_i(0,\b_{-i}) = 0,\quad {\forall \b_{-i}}. \label{eqn:nfl}
\end{align}
\subsection{Deterministic and Randomized Mechanisms}
\label{subsec:determ:discussion}
We assume generic positions of bids, which means that all bids are distinct. For simplicity, we would want the mechanism to be deterministic, i.e., the same input bidding vector leads to the same allocation outcome, equivalently $a_i(b_i,\b_{-i})\in \{0,1\}$. However, we can prove that even if we relax U-DSIC to U-BNIC, no \emph{deterministic} U-BNIC and 1-SCP TFM that satisfy mild conditions can achieve positive miner revenue, indicating that the randomness in our main mechanism is necessary. The formal discussion is in Appendix~\ref{sec:imp}.

{An intuitive explanation about how this works is as follows: to construct a U-BNIC TFM we essentially ``adjust'' the users' payments from a U-DSIC mechanism in a way that increases the miner revenue while preserving U-BNIC and 1-SCP (via the auxiliary mechanism method discussed in Section~\ref{sec:roadmap}); however in a deterministic mechanism, the payments of users with $a_i(b_i,\b_{-i})=0$ are fixed at $0$ and cannot be adjusted, rendering the auxiliary mechanism method inapplicable.}


\section{The Auxiliary Mechanism Method}\label{sec:roadmap}


In this section, we introduce our main technique to construct the desired U-BNIC mechanism, named as \emph{auxiliary mechanism method}. We will make the connection between BNIC and DSIC mechanisms by developing a decomposition of many U-BNIC TFMs into an \emph{auxiliary} U-DSIC TFM and a so-called \emph{variation term}. In light of this, we develop Theorem~\ref{thm:decomposition}, the key theorem of this section, to provide sufficient conditions for any combination of an auxiliary U-DSIC TFM and a variation term to form a desired TFM that is simultaneously U-BNIC and $1$-SCP (Section~\ref{sec:auxiliary-dominant-auxiliary} and Section~\ref{sec:auxiliary-auxiliary-variation-decomposition}). Our Theorem~\ref{thm:decomposition} will provide a general framework to facilitate the construction of our desired U-BNIC TFMs, and we will discuss more about this framework in Section~\ref{sec:auxiliary-construction-framework}. Finally, in Section~\ref{sss:aux:1scp} and Appendix~\ref{app:add_persp}, we will provide more explanations and concrete examples to help the readers understand our auxiliary mechanism method, which, however, is not a pre-requisite of the constructions in the later sections. 


\subsection{The Dominant Auxiliary of a BNIC TFM and Their Relations} \label{sec:auxiliary-dominant-auxiliary}

For simplicity, let us first consider the mechanism on the users' side, and recall the famous Myerson's Lemma \citep{myerson1981optimal} that characterizes the sufficient and necessary condition for a TFM to be U-DSIC, stated as follows. (We note that the original Myerson's Lemma is stated for auctions. However, if we only focus on the users' incentive compatibility constraints, a TFM reduces to an ordinary auction. Indeed, in our statement of Lemma~\ref{lem:myerson}, the miner's revenue function $r$ is irrelevant.)
\begin{lemma}[Myerson's Lemma   \citep{myerson1981optimal}]\label{lem:myerson}
Any TFM $M = (\a, \p, r)$ is U-DSIC if and only if the following conditions are satisfied.
  \begin{itemize}
      \item Monotone allocation: $a_i(\cdot,\b_{-i})$ is monotonic non-decreasing,
      \item Constrained payment function: 
      \begin{align}\label{eq-lem-myerson-constrained-payment-function}
      \begin{aligned}
      &a_i(b_i,\b_{-i})p_i(b_i,\b_{-i}) \\
      =& \int_{0}^{b_i} t \pd{a_i(t,\b_{-i})}{t} dt + a_i(0,\b_{-i})p_i(0,\b_{-i}).
      \end{aligned}
      \end{align}
  \end{itemize}
\end{lemma}

Motivated by Lemma~\ref{lem:myerson}, for any monotonic non-decreasing allocation rule $\a$, we define a payment rule $\p$ to be its \emph{dominant association}. In particular, we set
\begin{align}\label{eqn:aux:payment}
p_i(b_i,\b_{-i}) &= \begin{cases}
    \frac{\int_{0}^{b_i} t \pd{a_i(t,\b_{-i})}{t}dt}{a_i(b_i,\b_{-i})},&\:a_i(b_i,\b_{-i})>0\\
	0,&\:a_i(b_i,\b_{-i})=0\\
 \end{cases}  .
\end{align}
Since our definition Eq.~\eqref{eqn:aux:payment} satisfies the condition in Eq.~\eqref{eq-lem-myerson-constrained-payment-function}, for any monotonic non-decreasing $\a$, together with its dominant association $\p$, we get a U-DSIC mechanism $(\a, \p, 0)$. (The miner reward function here is set to be constantly $0$ only for illustration purpose. It could be a different $r$.) We also note that there seems to be little freedom for the payment rule $\p$ in order to form a U-DSIC TFM with $\a$. Indeed, $p_i(b_i, \b_{-i})$ is relevant only when $a_i(b_i, \b_{-i}) > 0$ (i.e., when the $i$-th bidder has a chance to be confirmed). In this case, if we additionally add the natural boundary condition $p_i(0, \b_{-i}) = 0$ (similar to the NFL assumption), then $p_i(\cdot, \b_{-i})$ defined in Eq.~\eqref{eqn:aux:payment} is the unique solution to Eq.~\eqref{eq-lem-myerson-constrained-payment-function}.

For any U-BNIC mechanism $\tilde{M} = (\a, \tilde{\p}, \tilde{r})$ with monotonic non-decreasing $\a$, we let $M = (\a, \p, 0)$ to be its \emph{dominant auxiliary} mechanism where $\p$ is the dominant association of $\a$. Let us compare the payment rules of the two mechanisms and define a ``payment difference'' function that denotes the over-payment of each user according to $\tilde{M}$ compared to $M$, as 
\begin{align}\label{eq:def-theta}
    \theta_i(b_i,\b_{-i}) = a_i(b_i,\b_{-i})(\tilde{p}_i(b_i,\b_{-i})-p_i(b_i,\b_{-i})).
\end{align}
Note that the Revenue Equivalence Theorem 
indicates that for the same \emph{i.i.d.}\ path-connected distribution of valuations and given boundary conditions, all BNIC mechanisms with the same allocation rule should have the same \emph{expected} payment for a bidder with a fixed valuation (where the expectation is taken over the valuation of the other bidders). Since $M = (\a, \p, 0)$ is a DSIC mechanism and therefore also BNIC, we have that $\tilde{M} = (\a, \tilde{\p}, \tilde{r})$ and $(\a, \p, 0)$ make the same expected payment for any bidder the fixed valuation. Formally for any bidder $i$, we have that 
\begin{align}\label{eq:exp-theta-equals-zero}
\mathbb{E}_{\b_{-i} \sim V_{-i}}[\theta_i(b_i,\b_{-i})] &= 0,
\end{align}
which marks the close relation between the BNIC mechanism $\tilde{M}$ and its dominant auxiliary $M$.

\subsection{The Auxiliary-Variation Decomposition} \label{sec:auxiliary-auxiliary-variation-decomposition}



We have just shown that a U-BNIC mechanism $\tilde{M} = (\a, \tilde{\p}, \tilde{r})$ has a corresponding dominant auxiliary U-DSIC mechanism $M = (\a, \p, 0)$ (when $\a$ is monotone) and defined their payment difference to be $\bmtheta$ defined in Eq.~\eqref{eq:def-theta}. Formally, we summarize this decomposition and make the following definition.

\begin{definition}[Auxiliary-Variation Decomposition] \label{def:auxiliary-variation-decomposition}
Given a mechanism $\tilde{M} = (\a, \tilde{p}, \tilde{r})$ where $\a$ is monotonic non-decreasing, {we set up the \emph{auxiliary mechanism} $M=(\a,\p,0)$ and the \emph{variation term}  $T=(\bmtheta, \tilde{r})$.} Suppose we have the following conditions met, we call $(M, T)$ an \emph{auxiliary-variation decomposition} of $\tilde{M}$ and also write $\tilde{M} = M + T$ for short.
\begin{enumerate}
    \item In the auxiliary mechanism $M=(\a,\p,0)$, $\p$ is the dominant association of $\a$.
    \item The variation term  $T=(\bmtheta, \tilde{r})$ satisfies that
      $  \tilde{p}_i(b_i,\b_{-i}) = p_i(b_i,\b_{-i}) + \frac{\theta_i(b_i,\b_{-i})}{a_i(b_i,\b_{-i})}$,
    where we require that $\theta_i(b_i,\b_{-i}) = 0$ whenever $a_i(b_i,\b_{-i}) = 0$ and treat $0/0$ as $0$.
\end{enumerate}
\end{definition}

We will use Definition~\ref{def:auxiliary-variation-decomposition} in a reverse way: given a TFM $M = (\a, \p, 0)$ such that $\p$ is the dominant association of $\a$, if we could design a variation term $T = (\bmtheta, \tilde{r})$ that satisfies the additional \emph{admissibility} conditions, then we would expect that $\tilde{M} = M + T$ is not only U-BNIC but also $1$-SCP.
Formally, we define the admissibility conditions as follows. 
\begin{definition}\label{def:admissibility}
    We call the variation term $T = (\bmtheta,\tilde{r})$ \emph{admissible} if it satisfies the following conditions for every $\b$.
    \begin{align}
        \tilde{r}(b_i,\b_{-i}) - \tilde{r}(0,\b_{-i}) &= \theta_i(b_i,\b_{-i}), \label{eq:def-admissibility-1}\\
        \mathbb{E}_{\b_{-i} \sim V_{-i}}[\theta_i(b_i,\b_{-i})] &= 0. \label{eq:def-admissibility-2}
    \end{align}
\end{definition}

We note that the second admissibility condition (Eq.~\eqref{eq:def-admissibility-2}) derives from Eq.~\eqref{eq:exp-theta-equals-zero} which is necessary for the composed TFM $\tilde{M}$ to be U-BNIC. The first admissibility condition (Eq.~\eqref{eq:def-admissibility-1}) is key to guarantee the $1$-SCP property, which will be further explained in Section~\ref{sss:aux:1scp}.

The following theorem states that the admissibility conditions for the variation term are enough to guarantee the composed TFM is U-BNIC and $1$-SCP, and is the basis of our TFM constructions later.
\begin{theorem} \label{thm:decomposition}
Suppose the variation term $T$ is admissible. For any auxiliary mechanism $M$ that may form an auxiliary-variation decomposition with $T$, the composed TFM $\tilde{M} = M + T$ is U-BNIC and $1$-SCP.
\end{theorem}

The formal proof of Theorem~\ref{thm:decomposition} is deferred to Appendix~\ref{app:proof:decomposition}. 

\subsection{Using the Auxiliary-Variation Decomposition to Construct TFMs} \label{sec:auxiliary-construction-framework}
Given a U-DSIC TFM $M = (\a, \p, 0)$ where $\a$ is monotonic non-decreasing and $\p$ is the dominant association of $\a$, it is trivial to see that $T_{\bot} = (0, 0)$ is a trivial admissible variation term and $M + T_{\bot} = (\a, \p, 0)$ is both U-BNIC and $1$-SCP. However, in this trivial construction, the miner revenue is always $0$, which is not desirable. 

In order to achieve a larger miner revenue, we would like to explore more choices of $\tilde{M}$. Now, Theorem~\ref{thm:decomposition} provides an approach to create a class of U-BNIC and $1$-SCP TFMs based on $M$, so that we could hope to find a desired TFM from this class (e.g., with large miner revenue, and other properties such as UIR, BF, etc.). In particular, for any fixed $M$, we will be able to able to perturb the payment function (via $\bmtheta$) and create more design choices for $\tilde{r}$, which jointly enrich the space of admissible variation terms $\{T\}$, as well as the corresponding class $\{\tilde{M}\}$ of the composed U-BNIC and $1$-SCP TFMs, thanks to Theorem~\ref{thm:decomposition}. 

One nice thing about the above approach is that it is \emph{almost modular} in the design of the auxiliary mechanism $M$ and the variation term $T$. We note that the constraints for $M$ are almost independent of that for $T$, while the only correlating constraint is that $a_i(b_i,\b_{-i}) = 0 \Rightarrow \theta_i(b_i,\b_{-i}) = 0$, which is very mild and holds for most natural choices of $M$ and $T$. This modular framework decouples our design tasks for $M$ and $T$, greatly reduces the design complexity and renders our final mechanism more interpretable.

To describe more details about our construction framework, note that we would like to construct a TFM that also satisfies the UIR and BF conditions. Note that the auxiliary $M (\a, \p, 0)$ always satisfies UIR, and the zero miner revenue indicates it also satisfies BF. Now, intuitively, if we make the variation term $T$ ``small enough'' so that $\tilde{M} = M + T$ is close to $M$, $\tilde{M}$ is also likely to satisfy UIR and BF conditions, although having lower miner revenue as well. We further notice that the admissibility condition of $T$ is \emph{scale-free}, i.e., if $T = (\bmtheta, \tilde{r})$ is admissible and for any $h > 0$, $hT = (h \cdot \bmtheta, h \cdot \tilde{r})$ is also admissible. Therefore, we will construct our desired TFM along the following steps.
\begin{enumerate}
\item Construct an allocation rule $\a$, derive the corresponding dominant association $\p$ and the auxiliary TFM $M=(\a,\p,0)$.
\item Find a ``good'' admissible variation term $T = (\bmtheta, \tilde{r})$.
\item Compute the (approximately) optimal $h$ so that $M + hT$ maximizes miner revenue while still obeying UIR and BF conditions.
\end{enumerate}

In Sections~\ref{sec:size1}--\ref{section:size:k}, we will present our  constructions of $M$ and $T$, which, together with the optimal
choice of $h$, can achieve a constant-fraction approximation of the optimal miner revenue (with the desired properties: U-BNIC, $1$-SCP, UIR, BF, etc.).






\subsection{Further Explanation of the Condition Eq.~\eqref{eq:def-admissibility-1}.} \label{sss:aux:1scp}

In this part, we explain the relationship between the payment difference function $\bmtheta$ and the miner revenue function $\tilde{r}$ in an admissible variation term, as well as the condition Eq.~\eqref{eq:def-admissibility-1} for admissibility. 

To characterize the properties of the TFM $\tilde{M}$, we first look into the relations between $\tilde{\p}$ and $\tilde{r}$. Recall that in our model, users and the miner have different information on the bids: a user only knows the distribution of other users' valuations (and bids, if the mechanism is U-BNIC), but the miner knows all bids accurately. Therefore, in mechanism $\tilde{M}$, for fixed $\b_{-i}$, truthfully bidding $v_{i}$ does not guarantee to maximize user $i$'s utility $\tilde{u}(b_i,\b_{-i};v_i)$ (which is defined to be $a_i(b_i,\b_{-i})\cdot(v_i - \tilde{p}_i(b_i, \b_{-i}))$ following the definition in Eq.~\eqref{eq:def-u}), but it must maximize the total utility of her and the miner, as $\tilde{u}(b_i,\b_{-i};v_i) + \tilde{r}(b_i,\b_{-i})$, so that the miner is not incentivized to ask her to deviate. Therefore, if we further assume smoothness for  simplicity (a formal proof without the smoothness assumption is provided in Appendix~\ref{app:proof:decomposition}), 
we have
\begin{align}
    \left.\pd{}{b_i}\left(\tilde{u}(b_i,\b_{-i};v_i) + \tilde{r}(b_i,\b_{-i})\right)\right|_{b_i=v_i} = 0.
\end{align}

However, in the auxiliary TFM $M$, which is U-DSIC, bidding $b_i=v_i$ maximizes user $i$'s utility, hence
\begin{align}
    \left.\pd{}{b_i}u_i(b_i,\b_{-i};v_i) \right|_{b_i=v_i} = 0.
\end{align}

Since TFMs $\tilde{M}$ and $M$ have the same allocation rule, users' utilities only differ in payments, we have $\tilde{u}(b_i,\b_{-i};v_i) = u(b_i,\b_{-i};v_i) - \theta_i(b_i,\b_{-i})$. Therefore, we get the relation between $\theta_i(b_i,\b_{-i})$ and $\tilde{r}(b_i,\b_{-i})$:
\begin{align}\label{eqn:am:differential}
    \pd{}{b_i}\theta_i(b_i,\b_{-i}) = \pd{}{b_i}\tilde{r}(b_i,\b_{-i}).
\end{align}
That is, if user $i$ would benefit from an infinitesimal deviation from truthful bidding, the miner would lose the same amount in turn, so that the miner has no incentive to let user $i$ deviate, even though she has additional information about other users' bids. With the boundary condition $\tilde{p}_i(0,\b_{-i})=0$ (thus $\theta_i(0,\b_{-i})=0$), we get
\begin{equation*}
    \theta_i(b_i,\b_{-i}) = \tilde{r}(b_i,\b_{-i}) - \tilde{r}(0,\b_{-i}), \quad \forall i. 
\end{equation*}
This characterizes the relation between user payments and miner revenue in $1$-SCP TFMs, and also shows the need of condition Eq.~\eqref{eq:def-admissibility-1} in an admissible variation term.

Following the discussion of this part, we can actually find a critical challenge in constructing an admissible variation term, and develop an alternative field-theoretic perspective of the admissibility condition. The detailed discussion is in Appendix~\ref{app:add_persp}.

\section{The Proposed Mechanism for Block Size 1} \label{sec:size1}

In this section, we consider the case with block size $k=1$, where exactly one transaction is confirmed, to give a simple and intuitive understanding of our mechanism. We follow the pipeline of auxiliary mechanism method in construction. In Section~\ref{sec:size1:M} we construct the auxiliary mechanism named soft second-price mechanism, and in Section~\ref{sec:alg:one_size} we compute the variation term, thus finishing the construction of the proposed mechanism and compute its miner revenue.



\subsection{Auxiliary: the Soft Second-Price Mechanism}\label{sec:size1:M}

The second-price auction mechanism has been widely used in traditional auctions, in which the highest bidder gets confirmed but pays the second-highest bid. However, as we prove that any deterministic TFM which is U-BNIC and 1-SCP satisfying mild assumptions has non-positive miner revenue (Appendix~\ref{sec:imp}), we try to introduce randomness into the allocation rule.

Here we consider the widely used multinomial logit choice model 
in which the choice probability of an item is proportional to the exponential of a parameter $m$ times its value. If we set the $m$ to infinity, then the item with the highest value is deterministically chosen, coinciding with the allocation rule of second-price auction; if we set $m$ to zero, then all items are randomly chosen with uniform chances. For $m\in (0,+\infty)$, the choice is random, but higher-valued items are more likely to be chosen.

As a basis of our main mechanism, we first develop a U-DSIC and 1-SCP mechanism named soft second-price mechanism, which is the auxiliary mechanism of our proposed TFM. It adopts the multinomial logit choice model as the allocation rule. After fixing the allocation rule $\a$, we derive the corresponding dominant association $\p$ (according to Eq.~\eqref{eqn:aux:payment}), and form the auxiliary mechanism $M = (\a, \p, 0)$, which is explicitly presented in Mechanism~\ref{mec:1:aux}.

\begin{Mechanism}[!t]
\begin{framed}
    \begin{align}
        a_i(b_i,\b_{-i}) &= \frac{e^{mb_i}}{\sum_{j=1}^{n} e^{mb_j}}  \label{ssp_begin}\\
        p_i(b_i,\b_{-i}) &= b_i - \frac{\sum_{j=1}^{n}e^{mb_j}}{me^{mb_i}} \cdot  \ln\frac{\sum_{j=1}^{n}e^{mb_j}}{1+\sum_{j\neq i}e^{mb_j}}  \label{ssp_payment}\\
        r(\b) &= 0. \label{ssp_end}
    \end{align}
\end{framed}
\caption{Auxiliary Mechanism for block size 1}
\label{mec:1:aux}
\end{Mechanism}

One good thing about our auxiliary mechanism is that every entry of the allocation function $\a$ is always positive for all $m < \infty$. Therefore, it automatically satisfies the requirement in the second item of Definition~\ref{def:auxiliary-variation-decomposition} and can be combined with any variation term (to be designed soon) in our auxiliary-variation decomposition.

Although the soft second-price mechanism has zero miner revenue, we can modify it via the auxiliary mechanism method that preserves U-BNIC and 1-SCP properties and yields positive expected miner revenue, as in Section~\ref{sec:alg:one_size}.

\ifdefined\ConnectionToSPA

\subsubsection{Connection to the Ordinary Second-Price Auction}

We notice that when $m\to +\infty$, the highest bidder will be confirmed. Actually, we can observe that when $m\to \infty$, and bids are \emph{generic}(distinct in $(0,1)$),  the user payment rule converges to the second price auction, i.e.

\begin{observation} \label{obs:1}
  Without loss of generality, we assume that $b_{(1)}>b_{(2)}>\cdots >b_{(n)}$ is a permutation of $\b$. We then have that
  \begin{align}
      \lim_{m\to +\infty} a_i(b_i,\b_{-i}) &= \begin{cases}
    1,\:b_i = b_{(1)}\\
	0,\:b_i\ne b_{(1)}
  \end{cases}, \label{ssp:lim:alloc}\\
  \lim_{m\to +\infty} p_i(b_i,\b_{-i}) &= \min\{b_i,b_{(2)}\}. \label{ssp:lim:payment}
  \end{align}
\end{observation}

Then we consider another extreme case when $m= 0$. However, the denominator of $p_i(b_i,\b_{-i})$ is zero, so we consider its limit of $m\to 0$. In this case, the mechanism converges to a random-free-allocation mechanism to confirm a transaction uniformly at random and charge no fee, i.e.
\begin{observation}\label{obs:2}
  \begin{align}
      \lim_{m\to 0} a_i(b_i,\b_{-i}) &= \frac{1}{n}, \label{ssp:lim0:alloc}\\
  \lim_{m\to 0} p_i(b_i,\b_{-i}) &= 0. \label{ssp:lim0:payment}
  \end{align}
\end{observation}

The proofs of observations above are deferred to Appendix \ref{appendix:proof:obs}. In fact, we can observe that when $m$ increases from $0$ to $+\infty$, the mechanism continuously shifts from random free allocation to all-burn second-price auction, while preserving 1-SCP and U-DSIC properties. When we consider fixed $m \in (0,+\infty)$, we can modify the mechanism via the auxiliary mechanism method to make it still U-BNIC but with positive expected miner revenue.
\fi
\subsection{The Variation Term and Our Proposed Mechanism for Block Size 1}
 \label{sec:alg:one_size}
 
Following the auxiliary mechanism method in Section~\ref{sec:roadmap}, we can construct a mechanism $\tilde{M}=(\mathbf{a},\tilde{\mathbf{p}},\tilde{r})$ via the composition of its auxiliary mechanism $M$ and the variation term $T = (\bmtheta,\tilde{r})$. In this section, we construct the variation term of our proposed mechanism.

When the distributions of all users' valuations are \emph{i.i.d.}, i.e. $V = V_1 \times V_2 \times \cdots \times V_n$ and $\forall V_i = V_0$ has identical pdf $\rho:[0,1]\to [0,+\infty)$, we denote
\begin{align} \label{eqn:crho}
    c_\rho = \int_0^1 \rho^2(t) dt.
\end{align}

Now, for any scaling parameter $h\in [0,+\infty)$, we construct $T=(\bmtheta,\tilde{r})$ as follows. {The intuition in the construction is elaborated in Appendix~\ref{app:am:intuition}.}
\begin{align}\label{eqn:theta}
    \theta_i(b_i,\b_{-i}) = -\frac{1}{2}hb_i^2\left(\frac{\sum_{j\ne i}b_j^2}{c_\rho(n-1)}-1\right) ,
\end{align}
\begin{align}\label{eqn:r}
    \tilde{r}(\b) = \frac{1}{2} h \left(\sum_{i=1}^{n}b_i^2 - \frac{\sum_{1\le i < j \le n}{b_i^2b_j^2}}{c_\rho(n-1)}\right).
\end{align}

As a sanity check, we note that $\tilde{r}$ is the potential of $\D_\bmtheta$ (as in Appendix~\ref{app:add_persp}). Formally, the following lemma verifies that the above variation term $T$ is admissible. The proof of Lemma~\ref{lem:admis} is deferred to Appendix~\ref{app:proof:admis}.

\begin{lemma}\label{lem:admis}
    The variation term $T=(\bmtheta, \tilde{r})$ defined in Eqs.~(\ref{eqn:theta}-\ref{eqn:r}) is admissible.
\end{lemma}

Now we combine the auxiliary mechanism $M$ defined in Mechanism~\ref{mec:1:aux} and our variation term $T$ to form the mechanism $\tilde{M}=(\a,\tilde{\p},\tilde{r}) = M + T$, which is explicitly presented in Mechanism~\ref{mec:size1}. (We assume $n\ge 2$. If $n=0$ then nobody can be confirmed and there is no payment, and if $n=1$ then we set $(a, \tilde{p}, \tilde{r}) = (1, 0, 0)$.)

\begin{Mechanism}[!t]
\begin{framed}
\begin{align*}
        a_i(b_i,\b_{-i}) &= \frac{e^{mb_i}}{\sum_{j=1}^{n} e^{mb_j}}  \\
        \tilde{p}_i(b_i,\b_{-i}) &= b_i - \frac{\sum_{j=1}^{n}e^{mb_j}}{me^{mb_i}}\left(  \ln\frac{\sum_{j=1}^{n}e^{mb_j}}{1+\sum_{j\neq i}e^{mb_j}}\right. \\
        &\qquad\qquad \left.+ \frac{1}{2}hmb_i^2 \left(\frac{\sum_{j\ne i}b_j^2}{c_\rho(n-1)}-1\right)\right)  \\
        \tilde{r}(\b) &= \frac{1}{2} h \left(\sum_{i=1}^{n}b_i^2 - \frac{\sum_{1\le i < j \le n}{b_i^2b_j^2}}{c_\rho(n-1)}\right)
\end{align*}
\end{framed}
\caption{Transaction Fee Mechanism for block size 1}
\label{mec:size1}
\end{Mechanism}

By Lemma~\ref{lem:admis} and Theorem~\ref{thm:decomposition}, we have that the TFM $\tilde{M}$ is U-BNIC and 1-SCP for all $h\in [0,+\infty)$. We can also compute the expected miner revenue as follows.
\begin{align} \label{eq:block-size-1-miner-revenue}
    \E_{\b\sim V} [\tilde{r}(\b)] = \frac{1}{4} hnc_{\rho}>0.
\end{align}
From Eq.~\eqref{eq:block-size-1-miner-revenue}, we see that the expected miner revenue is always positive and it grows linearly with our scaling parameter $h$.

However, we have to be careful about choosing the value of $h$. Intuitively, the value of $h$ describes the extent of perturbation from the original U-DSIC mechanism, and when the perturbation is too large, the \emph{individual rationality} ($p_i(b_i,\b_{-i}) \le b_i$) and \emph{budget feasibility} properties may not hold. Actually, since the block size is $1$, the payment cannot exceed the valuation of the accepted bid. Then we have $\E_{\b\sim V} [\tilde{r}(\b)] \le 1$. Therefore, we have the following natural upper bound for $h$:
\begin{align}
    h \le O(1/(c_\rho n)). \label{eqn:h:upper:bound}
\end{align}

For the best miner revenue, we want to make $h$ as large as possible while keeping the mechanism feasible. Fortunately, for fixed $c_\rho$, we have an estimation of optimally feasible $h$ that enables a constant approximation ratio of the optimal revenue while preserving UIR and BF constraints.  {Hence, in the setting of block size $1$, we can design a TFM that satisfies desirable incentive properties and has a constant fraction of optimal revenue. The formal result is:}

\begin{theorem}\label{thm:h:value:1}
    For $n\ge 2$, we consider Mechanism~\ref{mec:size1} {with parameter $m=1$. Then for any $h \in[0, \frac{2c_\rho (n-1)^2}{en^3}]$}, the corresponding mechanism $\tilde{M}=(\mathbf{a},\tilde{\mathbf{p}},\tilde{r})$ is U-BNIC, 1-SCP, UIR and BF. 
    
    {
    Furthermore, for any $C\in [0,1)$, the mechanism is $(C,\frac{6C+5}{1-C^2})$-U-SP.
    }
\end{theorem}

The proof of Theorem~\ref{thm:h:value:1} is deferred to Appendix~\ref{app:proof:param:1}. {Note that when we set $h=\frac{2c_\rho (n-1)^2}{en^3}$, the expected miner revenue is $\frac{1}{4}hnc_{\rho} = \frac{c_\rho^2(n-1)^2}{2en^2} = \Theta(c_\rho^2)$}. As the optimal miner revenue is at most $\max\{v_i\}\le 1$, for fixed distribution (fixed $c_\rho$), our mechanism yields a constant-ratio approximation of the optimal miner revenue for $n\to \infty$. {Particularly, for the case of uniform distribution with $c_\rho=\frac{1}{3}$ and $n\to\infty$, our approximation ratio is $\frac{1}{18e}$. Furthermore, the U-SP results show that for large $n$, a user cannot gain more utility unless she injects as many fake transactions as the total amount of honest transactions competing for the block, which is unrealistic for the real-world blockchain ecosystem. This result also matches with a basic concept of blockchains: \emph{the security of a blockchain system is fundamentally based on the assumption that (at least) 50\% of participants are honest.}}

\section{Mechanism for General Block Size $k$}\label{section:size:k}

In most blockchains, a block usually contains multiple transactions. Therefore, it is desirable to extend our Mechanism~\ref{mec:size1} to general block size $k$. Recall that when the block size is $1$, we adopted a simple soft second-price mechanism as the auxiliary. However, it seems trickier to extend this auxiliary to a general block size $k$, as the softmax function does not have a straightforward extension for soft-top-$k$. In Section~\ref{sec:drawing:rule}, we will work on the details of the auxiliary mechanism for general block size $k$. The high-level idea of this step is natural -- we adopt a $k$-step weighted sampling without replacement approach \citep{ben2018weighted} to confirm the $k$ bids in a block, where in each step, we still apply the logit choice rule.

Once we figure out the details of the allocation rule $\a$ for general block size $k$, the rest construction will follow our auxiliary mechanism -- We first straightforwardly compute its dominant auxiliary mechanism $M = (\a, \p, 0)$ where the corresponding dominant association $\p$ is defined by Eq.~\eqref{eqn:aux:payment}. Then, we will still use the variation term $T$ defined previously in Section~\ref{sec:alg:one_size}, and combine the auxiliary and the variation term to derive the final mechanism.  

In Section~\ref{sec:general-block-size-miner-revenue}, we will show that the generalized mechanism enjoys the same incentive compatibility properties as the basic version (for block size $1$). We will also analyze the expected miner revenue of the generalized mechanism.




\subsection{Allocation Rule: Weighted Sampling without Replacement} \label{sec:drawing:rule}

In this section, we assume the bidding vector $\b$ and block size $k$ are fixed. For bidder $i\in B=[n]$, we set her weight $w_i=e^{mb_i}$. Now we compute $a_i$, the probability user $i$ has her transaction confirmed.

\def\J{\mathcal{J}}

Denote $\delta_t(i)$ as the probability that user $i$ in the $t$-th round and $W=\sum_{i=1}^n w_i$, then $\delta_1(i)=\frac{w_i}{W}$. For fixed $t\ge 2$, we consider $j=(j_1,\cdots,j_t)$ as the \emph{sampling vector} describing the outcome of the weighted sampling without replacement in the first $t$ rounds, in which $j_s$ is the user confirmed in the $s$-th round, and denote $\J$ as the distribution of $j$. Therefore, we get
$    \delta_t(i;b_i,\b_{-i}) = \Pr_{j\sim \J} [j_t=i]$.

We use the notation $\delta_t(i)= \delta_t(i;b_i,\b_{-i})$ when $(b_i,\b_{-i})$ is fixed, and denote $J_t(i)=\{j \textup{ is a sampling vector}: j_{t}=i\}$, then $\forall j\in J_t(i)$, denote $\delta_t(i;j) = \Pr_{u\sim \J}[u=j]$, then we have
\begin{align}
    \delta_t(i) = \sum_{j\in J_t(i)} \delta_t(i;j).\label{eqn:delta_decomposition}
\end{align}

Note that $\delta_t(i;j)$ denotes the probability that the sampling outcome is $(j_1,j_2,\cdots,j_{t-1},i)$, thus the probability is
\begin{small}
\begin{align}
    \delta_t(i;j) &= \frac{w_{j_1}}{W}\cdot \frac{w_{j_2}}{W-w_{j_1}} \cdot \cdots  \cdot \frac{w_i}{W-w_{j_1}-\cdots-w_{j_{t-1}}}. \label{eqn:delta_component}
\end{align}
\end{small}
Since we know that
\begin{align}
    a_i = \sum_{t=1}^{k} \delta_t(i), \label{eqn:delta_sum}
\end{align}
the allocation rule $\a$ can be computed from Eqs.~(\ref{eqn:delta_decomposition}-\ref{eqn:delta_sum}). According to Eq.~\eqref{eqn:aux:payment}, we compute the corresponding dominant association payment rule $\tilde{\mathbf{p}}$. We use the same variation term $T$ as in Section~\ref{sec:alg:one_size}. Thus, the final TFM can be described as in Mechanism~\ref{mec:sizek} (and we note that Mechanism~\ref{mec:size1} exactly the same as Mechanism~\ref{mec:sizek} when $k=1$). 
{Also} note that every entry of our constructed allocation rule $\a$ is always positive, and therefore the mechanism is well defined.

\begin{Mechanism}
\begin{framed}
\begin{small}
\begin{align*}
a_i(b_i,\b_{-i}) &= \sum_{t=1}^{k}  \delta_t(i;b_i,\b_{-i}) \\
\tilde{p}_i(b_i,\b_{-i}) &= \frac{1}{a_i(b_i,\b_{-i})} \Big[a_i(b_i, \b_{-i}) b_i - \int_{0}^{b_i} a_i(t, \b_{-i}) dt\\
&\qquad \qquad\qquad  -\frac{1}{2}hb_i^2\left(\frac{\sum_{j\ne i}b_j^2}{c_\rho(n-1)}-1\right) \Big]  \\
\tilde{r}(\b) &= \frac{1}{2} h \left(\sum_{i=1}^{n}b_i^2 - \frac{\sum_{1\le i < j \le n}{b_i^2b_j^2}}{c_\rho(n-1)}\right)
\end{align*}
\end{small}
\end{framed}

\caption{Transaction Fee Mechanism for general block size $k$}
\label{mec:sizek}
\end{Mechanism}

\subsection{Estimation of $h$: How Much Revenue Can Miner Get?} \label{sec:general-block-size-miner-revenue}


For the case of general block size $k$, we also have a similar result on the value of $h$, which additionally requires the number of users $n$ to be at least $\left(\frac{e}{e-1} + \Theta(1)\right) k \approx 1.582 k$. {Hence, for general block size $k$, as long as $n\ge 1.582k$, we can still design a TFM that satisfies desirable incentive properties with a constant fraction of optimal revenue.} The formal result is shown below:

\begin{theorem}[Main Theorem] \label{thm:h:value:k}
    For any block size $k$ and any parameters $m, h \in [0, +\infty)$, the mechanism $\tilde{M}=(\mathbf{a},\tilde{\mathbf{p}},\tilde{r})$ defined in Mechanism~\ref{mec:sizek} is U-BNIC and 1-SCP. For any fixed $\lambda_0 > \frac{e}{e-1}\approx 1.582$, and any 
    
    $n\ge \max\{\lambda_0 k,30\},$
    if we set
    $m=\min\left\{\frac{1}{2}\ln \frac{1}{\ln \frac{\lambda_0}{\lambda_0-1}},1 \right\}$, then for every
    $$
    h  \in \left[0,g(\lambda_0)\frac{2kc_\rho(n-1)}{en^2}\right],
    $$
    where
    $$g(\lambda_0)=\frac{m\cdot\max\left\{1-\sqrt{\ln\frac{\lambda_0}{\lambda_0-1}}, 1-e\ln\frac{\lambda_0}{\lambda_0-1}\right\} }{e^{m-1+\frac{e}{0.9\lambda_0-1}}},$$
    then our Mechanism~\ref{mec:sizek} $\tilde{M}$ is UIR and BF.

    Furthermore, for any $C\in [0,1)$, the mechanism is $\big(C,O(\frac{1}{1-C})\big)$-U-SP.
    \color{black}
\end{theorem}

\begin{figure}[tb]
    \centering
    \includegraphics[width = 0.9\textwidth]{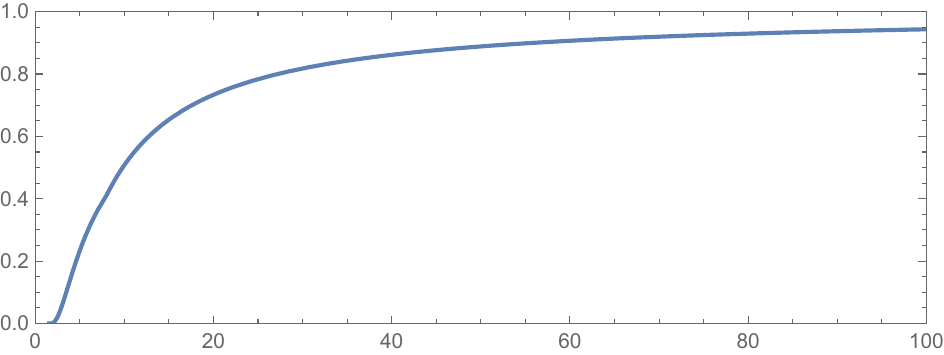}
    \caption{The plot of $g(\lambda_0)$ in Theorem~\ref{thm:h:value:k}.}
    \label{fig:thm4_plot2}
\end{figure}

The proof of Theorem~\ref{thm:h:value:k} is technically complicated and deferred to Appendix~\ref{app:proof:param:k}. {We plot $g(\lambda_0)$ in Figure~\ref{fig:thm4_plot2} and it holds that $\lim_{\lambda_0\to+\infty}g(\lambda_0)=1$. Combined with $\E[\tilde{r(\b)}]=\frac{1}{4}hnc_\rho$, the expected miner revenue is $\Theta(g(\lambda_0)c_\rho^2 k)$.} Because the optimal revenue for block size $k$ is at most $k \max\{v_i\}\le k$, for fixed $c_\rho>0$ and $\lambda_0 > \frac{e}{e-1}$, our mechanism yields a constant-factor approximation of the optimal revenue as long as {$n>\max\{\lambda_0 k,30\}$. While Mechanism~\ref{mec:size1} is essentially a special case of Mechanism~\ref{mec:sizek} with $k=1$, we can also notice that when $k=1$ and $n\to \infty$, the range of $h$ in Theorem~\ref{thm:h:value:k} is $[0,(1-o(1))\frac{2c_\rho}{en}]$,  matching with the result of Theorem~\ref{thm:h:value:1}.}

\section{Additional Properties of Our Mechanism}\label{sec:almost:mic}

In this section, we discuss the incentive and revenue properties for the miner in our proposed TFM. In Section~\ref{subsec:almost:mic}, we show that our proposed TFM is almost miner incentive compatible (MIC) and it is impossible to achieve the strict MIC property for any TFM. In Section~\ref{subsec:revenue:stability} we show that although the miner may get negative revenue in our TFM, it is only of a negligible probability and the miner will actually get a stable revenue close to the expectation. Therefore, our TFM can satisfy the miner's expectation on stable mining rewards in practice.

\subsection{Almost Miner Incentive Compatibility}\label{subsec:almost:mic}

In the previous parts of our paper, we mainly focused on the prevention of an individual user's deviations and the miner-and-single-user collusion (i.e., the deviation set \{\UUB, \UFT, \MUUB{1}\} in Table \ref{table:dishonst}). We now show that our TFM is also able to greatly reduce the additional miner utility derived from injecting and deleting (a limited number of) bids, and therefore achieving an \emph{almost MIC} property.

In particular, let us first fix the block size $k$ and the number of users $n$.  Since the injection-and-deletion deviation may change the $n$ parameter presented to the mechanism, we have to consider a variable-bid-size TFM $\tilde{\mathcal{M}}$ (as in Definition~\ref{def:TFM}'). To construct $\tilde{\mathcal{M}}$, we follow the method described above Eq.~\eqref{eq:regular-to-variable-size-TFM} -- we first choose a parameter $m > 0$; for each integer $\eta \geq 0$, we also choose a parameter $h_\eta > 0$ so that $\tilde{M}_\eta$ is fully determined following the description of Mechanism~\ref{mec:sizek}. In order to establish our almost MIC property, we need to choose $\{h_\eta\}$ in a way so that there exists $L$ satisfying
\begin{align}\label{eq:assumption-almost-MIC-L}
L = \frac{h_\eta \eta}{c_\rho k} \qquad \forall \eta.
\end{align}
Note that it is possible to appropriately set $\{h_\eta\}$ to meet the above condition while each $\tilde{M}_\eta$  also satisfies the conditions in Theorem~\ref{thm:h:value:k} (for $\eta > \lambda_0 k$, where $\lambda_0$ is defined in the theorem statement).  We finally let  $\mathcal{M}  = (\a, \tilde{\p}, \tilde{r}) =\cup_\eta \{\tilde{M}_\eta\}$ as formally defined in Eq.~\eqref{eq:regular-to-variable-size-TFM}.

We also need to characterize the degree of injection-and-deletion deviation by the miner. For any positive integer $\Delta$, we denote $B_\Delta(\b)$ as the set of all bidding vectors generated via injecting and deleting a total of at most $\Delta$ transactions to/from $\b$. Given the original bidding vector $\b$, for any $\b'\in B_{\Delta}(\b)$ that could result from the miner's injection-and-deletion deviation, we note that change to the miner's utility consists of the following two parts:
\begin{enumerate}
\item the change of the miner's reward: $\tilde{r}(\b') - \tilde{r}(\b)$;
\item the cost for the miner to inject fake bids: $
\sum_{b'_j \in \b' \backslash \b} \tilde{u}_{j}(b'_{j}, \b'_{-j}; 0) = - \sum_{b'_j \in \b' \backslash \b}   a_i(b_j',\b_{-j}')\cdot \tilde{p}_j(b_j', \b_{-j}')$,
where we extend the definition in Eq.~\eqref{eq:def-u} to our TFM sequence $\tilde{\mathcal{M}} = (\a,\tilde{\p},\tilde{r})$.
\end{enumerate}  
Note that the miner does not need to directly pay any cost for deleting a bid, but her reward $\tilde{r}(\cdot)$ may be changed due to this deletion.

We first define the (strict) Miner-Incentive-Compatibility (MIC) notion as follows:
\begin{definition}[Miner-Incentive-Compatibility (MIC)]\label{def:mic}
    Suppose there are $n$ real users and the block size is $k$. A variable-bid-size TFM $\tilde{\mathcal{M}} = \cup_{\eta} \{\tilde{M}_\eta\} = (\a,\tilde{\p},\tilde{r})$ satisfies the MIC property if and only if for any bidding vector $\mathbf{b}$ and $\Delta \ge 1$, we have
        $\sup_{\b' \in B_{\Delta}(\b)}\left(\tilde{r}(\b') - \tilde{r}(\b) + \sum_{b'_j \in \b' \backslash \b} \tilde{u}_{j}(b'_{j}, \b'_{-j}; 0)\right) \le 0$.  
\end{definition}


Relaxing from the strict MIC notion, are now able to introduce the almost MIC property for our TFM $\tilde{\mathcal{M}}$. We show that for if $\Delta = o(n)$ (i.e., the number of injected and deleted bids is a tiny fraction of the total number of users), then, as $n \to \infty$, with overwhelming probability, the additional utility for the miner is also $o(1)$.  {Particularly, by injecting or deleting $\Delta$ transactions when the total number of honest users is $n$, with high probability among the distribution of bids, the miner cannot gain an increase of revenue above $O(\frac{k\Delta^{4/3}}{n^{4/3}}).$ Hence, for example, if the miner can inject a constant number of fake transactions (without being caught, as discussed in Appendix~\ref{sec:crypto}), the additional revenue she may get is $O(\frac{k}{n^{4/3}})$ which is negligible when $n$ is large.} Formally, we have the following theorem.

\begin{theorem}[Our TFM is Almost MIC]
\label{thm:almost:mic}
Suppose there are $n$ real users and the block size is $k$ so that $n > \lambda_0 k + \Delta$. 
{Let} the variable-bid-size TFM $\tilde{\mathcal{M}} = \cup_{\eta} \{\tilde{M}_\eta\} = (\a,\tilde{\p},\tilde{r})$ be defined above with the $L$ parameter satisfying Eq.~\eqref{eq:assumption-almost-MIC-L}.  There exist universal constants $C_{M0}, C_{M1}, C_{M2} > 0$ such that for any $\epsilon \in (0,\frac{1}{2})$, if 
$n \ge \max \left\{C_{M1} \Delta, C_{M2} \log^3 \frac{1}{\epsilon}\right\}$,
we have that
\begin{align}
\begin{small}
\begin{aligned}
&\Pr_{\b \sim V}\left[\sup_{\b' \in B_{\Delta}(\b)}\left(\vphantom{ \sum_{b'_j \in \b' \backslash \b}}\tilde{r}(\b') - \tilde{r}(\b) \right.\right.\\
&\qquad\left.\left.+ \sum_{b'_j \in \b' \backslash \b} \tilde{u}_{j}(b'_{j}, \b'_{-j}; 0)\right) > C_{M0} L \cdot  \frac{k\Delta^{4/3}}{ n^{4/3}}\right] < \epsilon.
\end{aligned}
\end{small}
\end{align}

\end{theorem}

In other words, Theorem~\ref{thm:almost:mic} states that given a moderately large $n$, with probability at least $(1 - \epsilon)$ (over the realization of the $n$ real user valuations), the additional miner utility that could be gained from injecting and deleting at most $\Delta$ bids is at most $C_{M0} L \cdot  \frac{k\Delta^{4/3}}{ n^{4/3}}$, which is $o(1)$ for $\Delta = o(n)$ and fixed $k$, $L$. This result is quite non-trivial as one would naturally expect the relative revenue advantage should be $\Theta(\Delta/n)$ for usual {(incentive compatible)} mechanisms. {For example, in the $k$-item second-price auction with valuation distribution $\mathrm{Unif}[0,1]$, the expected $(k+1)$-th price is $\frac{n-k}{n+1}$, but if the miner injects a fake bid to be infinitesimally lower than the $k$-th bid, the expected price increases to $\frac{n-k+1}{n+1}$, gaining an $\Theta(1/n)$ relative advantage via injecting one bid.} It is also more useful since it gives additional incentive restrictions to the miner when we narrow the range of ``acceptable deviations'' for the miner. Here, the notion of the acceptable deviation is the range of small $\Delta$ (compared to $n$), as a large amount of injected or missing transactions (greater than the acceptable threshold) can be detected by the blockchain system via cryptographic schemes and the miner would be penalized for injection/deletion deviation. Please refer to Section~\ref{sec:crypto} for more details. The proof of Theorem~\ref{thm:almost:mic} is deferred to Appendix~\ref{app:proof:almost:mic}.


On the other hand, we can show that achieving strict MIC together with the main strategy-proof properties studied in this paper (U-BNIC and 1-SCP) is impossible if we additionally assume the natural NFL condition defined in Section~\ref{sec:irbf} (\cite{cryptoeprint:2022/1294} have independently proven a similar result). In particular, we prove that any TFM satisfying the above-mentioned properties has non-positive expected miner revenue, even if we only allow the miner to inject one zero-bidding fake transaction in the MIC property. Formally, we have the following theorem. 

\begin{theorem} [Impossibility of MIC] \label{thm:mic:impossibility}
Suppose there are $n$ real users and the block size is $k$. Consider any variable-bid-size TFM $\tilde{\mathcal{M}}=(\a,\tilde{\p},\tilde{r})$. Assume that for any $\eta \in \{1, 2, \dots, n+1\}$,  the natural restriction (formally defined in Section~\ref{sec:basic-model}) of $\mathcal{M}$ to a regular TFM for $\eta$ bids is U-BNIC and $1$-SCP, and NFL. If
$\tilde{r}(\b,0) - \tilde{r}(\b) \le 0$ holds for all $\b \in [0, 1]^0 \cup [0, 1]^1 \cup \cdots \cup [0, 1]^n$,
then we have that 
$\mathbb{E}_{\b \sim V} [\tilde{r}(\b)] \le 0$.
\end{theorem}
    
The proof of Theorem~\ref{thm:mic:impossibility} is deferred to Appendix~\ref{app:proof:no:mic}. Nevertheless, if we do not allow the miner to inject fake transactions, but allow her to delete existing transactions, whether there exists a TFM that is U-BNIC, 1-SCP and \{\MTD\}-proof remains open.

\subsection{Almost Miner Individual Rationality and Stability of Miner Revenue}\label{subsec:revenue:stability}

In previous sections, we have discussed the \emph{expected} miner revenue, which is in general a constant-fraction approximation of optimum, which naturally implies interim Miner Individual Rationality (MIR). In this section, we will be concerned about the guarantee on the \emph{worst-case} miner revenue. We define the (ex-post) MIR as a worst-case specification for the miner, which requires that the miner revenue is always non-negative no matter how the users bid. Formally,

\begin{definition}[(Ex-Post) Miner Individual Rationality]
Suppose there are $n$ users and the block size is $k$. A TFM $\tilde{M} = (\a,\tilde{\p},\tilde{r})$ satisfies the Miner Individual Rationality (MIR) property if and only if 
$    \tilde{r}(\b) \ge 0$ holds for all $\b \in [0, 1]^n$.
\end{definition}

The MIR condition requires that the miner revenue is non-negative for \emph{every} realization of the bidding vector $\b$. Unfortunately, our Mechanism~\ref{mec:sizek} do not satisfy such a strong condition. Recall that, in  Mechanism~\ref{mec:sizek}, we have $\tilde{r}(\b) = \frac{1}{2} h \left(\sum_{i=1}^{n}b_i^2 - \frac{\sum_{1\le i < j \le n}{b_i^2b_j^2}}{c_\rho(n-1)}\right)$ where $h > 0$.  For the particular bidding vector $\b = (1, 1, \dots, 1) \in [0, 1]^n$, we have that
$\tilde{r}(\b) = \frac{h}{2}  \cdot \left(n - \frac{{n(n-1)/2}}{c_\rho(n-1)}\right) =  \frac{hn}{2} \cdot\left(1-\frac{1}{2c_\rho}\right)$.

We now observe that, for the user valuation distribution $V_0$ with $c_\rho \in (0, 1/2)$ (for example, the uniform distribution over $[0, 1]$ leads to $c_\rho = 1/3$), our Mechanism~\ref{mec:sizek} is not MIR. In fact, we have the following sufficient and necessary condition of when our mechanism is MIR, the proof of which is deferred to Appendix~\ref{app:proof:mir:condition}.
\begin{theorem} \label{thm:mir:condition}
Our Mechanism~\ref{mec:sizek} is MIR if and only if  $c_\rho \ge \frac{1}{2}$.
\end{theorem}

Nevertheless, even when $c_\rho < 1/2$, so long as it is not too small, we are able to show that the miner gets a non-negative revenue with overwhelming probability, and the probability that the miner is ``not lucky'' to receive a negative revenue diminishes exponentially as the growth of $n$. Formally, we have the following theorem.
\begin{theorem}[Concentration of Miner Revenue]\label{thm:rev:conc}
Fix $n$ to be the number of users. For any block size $k$, let the TFM $\tilde{M} = (\a, \tilde{\p}, \tilde{r})$ be defined according to Mechanism~\ref{mec:sizek}. Assume that each user's bid follows the \emph{i.i.d.}\ distribution $V_0$ and let $c_\rho$ be correspondingly defined according to Eq.~\eqref{eqn:crho}. For any $\lambda > 0$, we have that
$\Pr\left[ \frac{\tilde{r}(\b)}{\E [\tilde{r}(\b)]} \le 1-\frac{\lambda}{c_\rho^2 n}\right] \le 2\exp(- \lambda)$.
\end{theorem}

The proof of Theorem~\ref{thm:rev:conc} is deferred to Appendix~\ref{app:thm:rev:conc}. Let $\lambda = c_\rho^2 n$, we immediate get the \emph{almost-MIR} property of our mechanism:

\begin{corollary}[Almost Ex-Post Miner Individual Rationality]\label{cor:alm:mir}
Following the setup in Theorem~\ref{thm:rev:conc}, we have that
$\Pr\left[\tilde{r}(\b) < 0\right] \le 2\exp(-c_\rho^2 n)$.
\end{corollary}
\begin{figure}[htb]
    \centering
    \includegraphics[width = 0.8\textwidth]{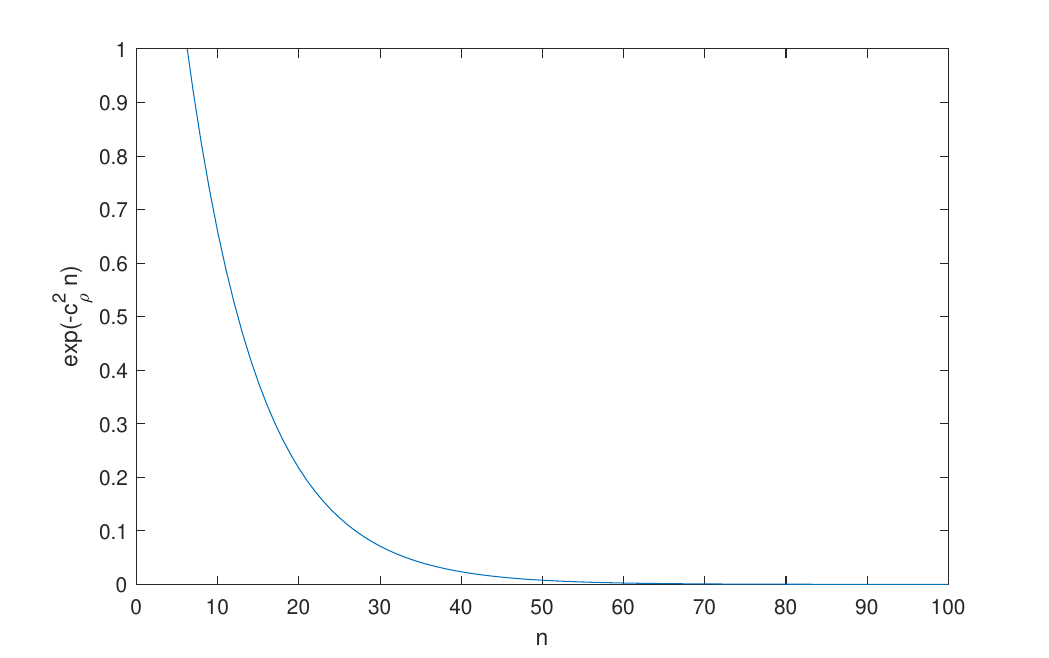}
    \caption{The diminishing probability of MIR being violated for the uniform distribution $b_i\sim \mathrm{Unif}[0,1]$.}
    \label{fig:mir}
\end{figure}

{From Corollary~\ref{cor:alm:mir}, we can see that the probability of $\tilde{r}(\b)<0$ diminishes exponentially as $n$ increases. As a demonstration, we show the plot of $\exp(-c_\rho^2 n)$ for the uniform distribution $b_i\sim \mathrm{Unif}[0,1]$ in which $c_\rho = \frac{1}{3}$. Considering the practical scenario that the Bitcoin and Ethereum blockchains typically have thousands of transactions in each block, in practice $\tilde{r}(\b)\ge 0$ would indeed hold with overwhelming probability.}

\textbf{Consideration of Mining Costs.}
In Section~\ref{subsec:strategy:proof} we assume that the miner's utility is equal to the miner revenue and ignore the mining costs and block rewards. In actual cases, we denote $r_0$ as the block reward minus mining cost, and as long as $r_0\ge 0$, the Corollary~\ref{cor:alm:mir} is still valid. In case that $r_0<0$, from Theorem~\ref{thm:rev:conc} we see that as long as $\E[\tilde{r}(\b)] \ge (1+\Theta(1))|r_0|$, the probability that the miner gets a negative utility still diminishes exponentially with $n$.

\color{black}

\section{Discussion}

In this paper, we model each transaction as constant-sized. However, in modern blockchain systems, transactions -- especially with smart contracts in modern applications, have variant sizes. The knapsack auction problem \citep{aggarwal2006knapsack} for the setting of blockchains is still open to future study.

{Our mechanism mainly considers the valuation distributions of bounded supports. While our methodology may extend to distributions of unbounded supports as long as $c_\rho = \E[b_i^2]$ is finite, the bounded support assumption is necessary for our estimation of $h$ and the revenue approximation guarantees. Whether it is possible to extend our approach to more general valuation distributions is also open to future study.}

While we assume symmetry of the joint distribution of users' valuations, the actual distribution can still be correlated, which is not discussed in our work. Furthermore, in the real-world scenario, as blockchain users may value the blockchain space depending on the expectation of the market, they may even have interdependent valuations, as modeled by the work by \cite{eden2022private}. Consideration of correlated or interdependent valuations in the design of blockchain TFMs can be a challenging but interesting future direction.

\theendnotes

\ACKNOWLEDGMENT{Zishuo Zhao would like to appreciate Shuran Zheng for delightful discussions that inspire his interest in a wide range of relevant areas. Xi Chen would like to thank the support from NSF via the Grant IIS-1845444. The authors would like to thank the area editor, associate editor and referees for useful comments that help improve this paper.}

\bibliographystyle{apalike}
\bibliography{ref}

\subsection*{Biography}

Xi Chen is the Andre Meyer Professor in the Department of Technology, Operations, and Statistics at
the Stern School of Business, New York University. His research spans the intersection of data science
and operations, with particular focus on statistical machine learning, stochastic optimization, data-driven
operations management, and blockchain technology.

David Simchi-Levi is a professor of engineering systems at Massachusetts Institute of Technology, where he leads the MIT Data Science Laboratory. He is an INFORMS Fellow, MSOM Fellow, and a member of the National Academy of Engineering. Algorithms developed by the Data Science Laboratory have been implemented by a large number of companies in a variety of industries such as airline, insurance, manufacturing, and retail.

Zishuo Zhao is a Ph.D. candidate in the Industrial \& Enterprise Systems Engineering Department at
University of Illinois Urbana-Champaign. He received his B.Eng. degree from the Yao Class of Institute for Interdisciplinary Information Sciences, Tsinghua University. His research interest lies in the intersection of mechanism design, blockchain technologies, and trustworthy AI, with an emphasis on the design and theoretical boundaries of incentive-compatible decentralized platforms in the perspective of social responsibilities, e.g., social fairness, economical efficiency, and environmental sustainability.

Yuan Zhou received his B.Eng. degree from Tsinghua University and his Ph.D. in Computer Science from Carnegie Mellon University. He has held academic positions as an Assistant Professor at the University of Illinois at Urbana–Champaign and Indiana University Bloomington, as well as serving as an Instructor in Applied Mathematics at the Massachusetts Institute of Technology. Currently, he is an Associate Professor at the Yau Mathematical Sciences Center and the Department of Mathematical Sciences, Tsinghua University. His research interests lie in stochastic and combinatorial optimization, with applications in online learning \& decision-making theory, and operations research \& management.

\ECSwitch





\begin{center}
    \Large \textbf{Appendix}
\end{center}

\section{Cryptographic Protocols for On-Chain Implementation} \label{sec:crypto}

Ideally, to design a credible blockchain TFM, we seek to discourage all kinds of dishonest behavior by either \emph{systematically} preventing them from being conducted, or \emph{economically} discouraging them by making them non-profitable.

Fortunately, the transparency property of a blockchain \citep{bertino2019data} and its implementation of many cryptographic protocols \citep{luntovskyy2018cryptographic} have already helped prevent several types of dishonest behaviors. For example, since the blockchain is public, it is not possible for the miner to behave in a Byzantine manner via commuting different bidding vectors to different users (see the discussion in \citep{credibility}), and the slashing rule in the Ethereum blockchain also discourage the miner from conducting certain classes of dishonest behavior via monetary penalties \citep{cassez2022formal}.

Also, \cite{credibility} propose to adopt a secure commitment scheme, which uses cryptographic protocols to guarantee that a bid cannot be modified after proposal. This scheme has the following advantages:

\begin{enumerate}
    \item It restricts the strategy space of the miner to merely adding fake transactions and concealing transactions, ruling out strategies for the miner to collude with users and change existing bids.
    \item It implements a sealed-bidding auction format that not only makes the Bayesian game modeling valid but also guarantees fairness among users' information sets, restricting users' strategy space and preventing the MEV issue in which the miners strategically manipulate transaction orders to increase their utility.
\end{enumerate}

\begin{remark}
\label{remark:expost}
while we only need to prevent individual user deviations in the interim setting, for $c$-SCP and MIC properties we want a stronger ex-post version.
\end{remark}

Particularly, we can implement the \emph{commitment scheme} in the way as follows:
\begin{enumerate}
    \item Users submit the (salted) \emph{hash values} of their transactions.
    \item The miner packs and broadcasts all the hash values of the transactions that compete for the block, following by a hash value of the all packed hash values. 
    \item The users reveal their transactions and the miner uploads them. If the uploaded transactions deviates from the hash values too much ($\Delta\ge \epsilon_3 n$ for a pre-set $\epsilon_3\in (0, 1)$, with $\Delta$ defined in Section~\ref{subsec:almost:mic}), the miner is penalized.
    \item The system processes the TFM.
\end{enumerate}

For the miner-only deviation, the miner may behave dishonestly in Steps 1-3, and the number of deviations can be restricted in the way as follows:

\begin{enumerate}
    \item The miner may submit fake transactions in Step 1, without seeing the honest transactions (interim \MFT). The system can restrict the number of transactions proposed by an identity in any block, and require any identity to have a deposit before proposing any transaction, so that the miner cannot create a large number of identities to submit too many fake transactions. We assume that the miner would not afford to inject more than $\epsilon_1 n$ transactions.
    \item The miner may ignore some hashes in Step 2, without seeing their bids (interim \MTD). In this way, the system effectively runs with a smaller $n$. But if we set the parameter $h$ in the way described in Section~\ref{subsec:almost:mic}, reducing $n$ cannot benefit the miner's revenue. Besides, the users who have their hashes ignored can also report this behavior and get the miner penalized. We assume that the miner will be caught if she ignores more than $\epsilon_2 n$ hashes.
    \item The miner may insert or ignore transactions after she sees the bids in Step 3 (ex-post \MFT~and~\MTD), but this type of behavior will be detected. If the number of deviations goes beyond an acceptable level, the miner will be penalized. On the other hand, an acceptable level $\epsilon_3 > 0$ is necessary because a missing transaction might also be simply due to the unstable connection from the user.
\end{enumerate}

Hence, our protocol can restrict the miner individual deviation into a low level compared to $n$, and from the argument in Section~\ref{subsec:almost:mic}, the relative advantage in miner revenue from \{\MFT, \MTD\} is bounded below $O\left(\left(\frac{\epsilon_1+\epsilon_3}{1-\epsilon_2}\right)^{4/3}\right)$. However, the miner-user collusion cannot be effectively prevented in this way, as they may conduct the collusion off-chain before Step 1.

Therefore, we can remark that:

\begin{remark}
    \label{remark:no:mic}
    Existing cryptographic protocols can effectively prevent miner individual deviations, but can only prevent part of miner-user collusions.
\end{remark}

On the other hand, one may feel that the individual user's deviation is a ``least destructive'' honest behavior, because it happens in users' minds and does not seemingly disrupt the blockchain system. Hence, it also cannot be detected or prevented on the system level at all. However, we still argue that a desirable TFM should satisfy \emph{truthfulness}, i.e., no individual user's deviation should be profitable. One key reason to design truthful mechanisms is the Revelation Principle \citep{myerson1981optimal, myerson1979incentive}: informally, for any non-truthful mechanism, we can construct an ``equivalent'' direct truthful mechanism that incorporates agents' optimal strategies into the mechanism itself, so that agents would maximize their utilities by reporting their true types (bidding their valuations). It renders untruthfulness unable to gain more advantage revenue.\footnote{As long as there exists a mechanism whose outcome can achieve certain desired properties, we can indeed construct the equivalent truthful mechanism that both prevents agents from strategic behavior, and simplify the analysis as we can assume rational agents who seek to maximize their individual utilities will indeed follow the mechanism as we expect.}  Additionally, by the argument of the Revelation Principle, we also only need to consider single-round mechanisms.
Hence, we remark that:

\begin{remark}\label{remark:revelation}
    The optimal revenue for any single-round truthful TFM  is optimal even considering the class of non-truthful and multi-round mechanisms.
\end{remark}

Furthermore, due to the anonymity of the blockchains \citep{khalilov2018survey}, it is difficult for users to collude with each other, as argued by \cite{shi}. Thus, user-user collusion is not a critical issue in the design of blockchain transaction fee mechanisms. Therefore, the remaining challenge to resolve is the prevention of user individual deviation and miner-user collusion, but as we have discussed, such dishonest behavior cannot be effectively prevented at the systematic level, so we have to discourage them in an economic way. In conclusion, we can remark that:

\begin{remark} \label{remark:objective}
    To design a desirable blockchain transaction fee mechanism, the most critical challenge is to \textbf{discourage \underline{individual user's deviation} and \underline{miner-user collusion} via economic methods}.
\end{remark}

\section{Impossibility Result on Deterministic TFM} \label{sec:imp}

In this section, we propose an impossibility result that under certain conditions, any deterministic TFM which is U-BNIC and 1-SCP cannot have positive miner revenue. Here we additionally introduce several notions. Although this impossibility does not fully rule out deterministic mechanisms, it does motivate us to introduce randomness into our main mechanism.

\textbf{Deterministic.} When bids are distinct, the outcome of the auction is deterministic, i.e., $a_i \in \{0,1\}$.

\textbf{Symmetric.} When we swap the bids of two users, their allocations and payments are exactly swapped.

\textbf{Continuous.} $\p$ and $r$ are continuous functions of $\b$, and $V$ has bounded, strictly positive PDF on a simply connected support $\textup{dom} (V)$.

\textbf{Strongly Monotone.} If we raise the bid of bidder $i$ while leave other bids unchanged, $a_i, p_i$ do not decrease and $a_j(\forall j\ne i)$ does not increase.

\begin{theorem} \label{thm:impossible}
  For all deterministic, symmetric, continuous, strongly monotone, user-individually-rational and budget-feasible TFMs, if $\mathbf{0} \in V$, then U-BNIC and 1-SCP implies non-positive miner revenue.
\end{theorem}

\subsection{Proof of Theorem~\ref{thm:impossible}}

\textbf{Proof sketch.} To prove the non-positive-miner-revenue property of all satisfying mechanisms, we first show that all satisfying mechanisms must obey certain restrictive conditions, as the payment (Sec. \ref{sec:imp:payment}) and revenue (Sec. \ref{sec:imp:rev}) rules both must follow corresponding closed-form formulas; then we show that this type of mechanisms have non-positive miner revenue.

In this section, we introduce the $\delta$-function with

\begin{align}
    \int_{-\epsilon}^{\epsilon} \delta(t)dt=1,\qquad \forall \epsilon>0.
\end{align}

We assume there exists a transaction fee mechanism $M_0(\a,\p,r)$ that satisfies all conditions.

\subsubsection{Pinning down the payment rule} \label{sec:imp:payment}

From definition we know that if $M_0$ is BNIC, then

\begin{align}
    \left.\pd{E_{\v_{-i}\sim V_{-i}}\left[u_i(b_i,v_i,\v_{-i})\right]}{b_i} \right|_{b_i=v_i} = 0, \qquad \forall v_i
\end{align}

i.e.,

\begin{align}
    &\int_{\v_{-i}} \left( \left(v_i-p_i(v_i,\v_{-i})\right)\pd{a_i(v_i,\v_{-i})}{v_i} \right.\nonumber\\
    &\qquad\quad\left. - a_i(v_i,\v_{-i})\pd{p_i(v_i,\v_{-i})}{v_i} \right)\rho_{-i}(\v_{-i}) d\v_{-i} = 0,
\end{align}

in which $\rho_{-i}(\cdot)$ is the pdf of $V_{-i}$.

For fixed $\v_{-i}$, since the mechanism is deterministic, we have that $a_i(\cdot,\v_{-i})\in \{0,1\}$ almost everywhere. Additionally because $a_i(\cdot, \v_{-i})$ is monotonic increasing, we have

\begin{align}
    a_i(v_i,\v_{-i}) =   \begin{cases}
    0,\:v_i<\theta(\v_{-i})\\
	1,\:v_i>\theta(\v_{-i}),
  \end{cases}
\end{align}

in which $\theta(\v_{-i})$ is a constant for fixed $\v_{-i}$. Therefore,

\begin{align} \label{eqn:imp:alloc}
    \pd{a_i(v_i,\v_{-i})}{v_i} = \delta(v_i-\theta(\v_{-i})).
\end{align}

Now we have a lemma:

\begin{lemma} \label{lem:det:thres}
    For $\forall \v_{-i}$,
    \begin{align}
        p_i(\theta(\v_{-i}),\v_{-i}) = \theta(\v_{-i}).
    \end{align} 
\end{lemma}

\proof{Proof.}
    If $p_i(\theta(\v_{-i}),\v_{-i}) > \theta(\v_{-i})$, let $t = p_i(\theta(\v_{-i}),\v_{-i}) - \theta(\v_{-i})$. Then by continuity, there exists a small $\epsilon>0$ s.t. $p_i(\theta(\v_{-i})+\epsilon ,\v_{-i}) > \theta(\v_{-i}) + \frac{t}{2}$ and $a_i(\theta(\v_{-i})+\epsilon ,\v_{-i}) = 1$, and the user $i$ would have negative utility. In this scenario, the miner would want to collude with user $i$ and ask him to change his bid to $\theta(\v_{-i}) - \epsilon$, so that user $i$ would now have $0$ utility.
    
    But by continuity, the change of the miner's revenue is arbitrarily small, increasing their total utility. So the 1-SCP property is violated.
    
    If $p_i(\theta(\v_{-i}),\v_{-i}) < \theta(\v_{-i})$, similarly there exists a scenario where user $i$ has valuation $\theta(\v_{-i})-\epsilon$ but the miner would want to let her bid $\theta(\v_{-i})+\epsilon$ instead, also violating 1-SCP.
    
    Therefore, it must hold that $p_i(\theta(\v_{-i}),\v_{-i}) = \theta(\v_{-i})$.

\hfill $\square$
\endproof
From Lemma \ref{lem:det:thres} we have
\begin{small}
\begin{align}
    \int_{\v_{-i}} \left( \left(v_i-p_i(v_i,\v_{-i})\right)\pd{a_i(v_i,\v_{-i})}{v_i} \right)\rho_{-i}(\v_{-i}) d\v_{-i} = 0,
\end{align}
\end{small}
so
\begin{align}
    \int_{\v_{-i}} \left( a_i(v_i,\v_{-i})\pd{p_i(v_i,\v_{-i})}{v_i} \right)\rho_{-i}(\v_{-i}) d\v_{-i} = 0.
\end{align}
Since monotonicity implies $\pd{p_i(v_i,\v_{-i})}{v_i} \ge 0$, we know that $\forall v_i> \theta(\v_{-i}), \pd{p_i(v_i,\v_{-i})}{v_i}=0$. Therefore, 
\begin{small}
\begin{align}
    \forall b_i>\theta(\v_{-i}), \epsilon > 0,\quad p_i(b_i,\v_{-i}) = p_i\left(\theta(\v_{-i}) + \epsilon, \v_{-i}\right).
\end{align}
\end{small}
Combined with Lemma \ref{lem:det:thres}, from continuity we get

\begin{align} \label{eqn:imp:payment}
    \forall b_i\ge\theta(\v_{-i}), \quad p_i(b_i,\v_{-i}) = \theta(\v_{-i}).
\end{align}

\subsubsection{Pinning down the miner revenue rule} \label{sec:imp:rev}

In this part, we mainly use the 1-SCP property to prove that the miner revenue is a constant with regard to any user. To show this, we prove a lemma:

\begin{lemma} \label{lem:imp:const_rev}
    If $v_i \ne \theta(\v_{-i})$, then $\pd{r(v_i,\v_{-i})}{v_i} = 0$.
\end{lemma}

\proof{Proof.}
We recall that the total utility of the miner and user $i$ is
\begin{small}
\begin{align}
        C_i(b_i,v_i,\v_{-i}) = a_i(b_i,\v_{-i}) (v_i - p_i(b_i,\v_{-i})) + r(b_i,\v_{-i}). 
    \end{align}
\end{small}
From 1-SCP we know that

\begin{align}
    0 &= \left.\pd{C_i(b_i,v_i,\v_{-i})}{b_i}\right|_{b_i=v_i} \\
    &= \left(\left(v_i - p_i(v_i,\v_{-i})\right)\pd{a_i(v_i,\v_{-i})}{v_i} \right. \nonumber \\
    &\qquad \qquad\left.- a_i(v_i,\v_{-i})\pd{p(v_i,\v_{-i})}{v_i}\right) + \pd{r(v_i,\v_{-i})}{v_i}.    
\end{align}
 
 From Eq.~(\ref{eqn:imp:payment}) we know $a_i(v_i,\v_{-i})\pd{p(v_i,\v_{-i})}{v_i} \equiv 0$, and from Eq.~(\ref{eqn:imp:alloc}) we know $v_i \ne \theta(\v_{-i}) \implies \pd{a_i(v_i,\v_{-i})}{v_i} = 0$. So we deduce
 
 \begin{align}
     v_i \ne \theta(\v_{-i}) \implies \pd{r(v_i,\v_{-i})}{v_i} = 0.
 \end{align}

\hfill $\square$
\endproof

Because the continuity condition guarantees $r(\b)$ is a continuous function of $\b$, from Lemma~\ref{lem:imp:const_rev} we know that for fixed $\v_{-i}$, $r(\cdot,\v_{-i})$ is a constant, hence
\begin{align} \label{eqn:imp:const_rev}
    r(v_i,\v_{-i}) = r(0,\v_{-i}).
\end{align}
By iteratively apply Eq.~(\ref{eqn:imp:const_rev}) to all components of $\v$, we get
\begin{align}
    r(\v) = r(\mathbf{0}).
\end{align}
We notice that from UIR,
\begin{align}
    r(\mathbf{0}) &\le \sum_{i=1}^n a_i(\mathbf{0}) p_i(\mathbf{0}) \\
    &\le \sum_{i=1}^n a_i(\mathbf{0}) \cdot 0 \\
    &= 0.
\end{align}
Therefore, we have
\begin{align}
    r(\v) \le 0, \qquad \forall \v.
\end{align}
Here we prove Theorem~\ref{thm:impossible}.

\section{Additional Perspectives of Auxiliary Mechanism Method}\label{app:add_persp}

\subsection{A Failed Example: the First-Price Auction} \label{sss:aux:fpa}

In this part, we use a simple example to help readers understand the constraints for an admissible variation term. In particular, we will demonstrate a $\bmtheta$ function that cannot be coupled with any $\tilde{r}$ to form an admissible variation term. The $\bmtheta$ function is constructed based on the natural first-price auction. As an interesting by-product, this example also shows that, although the first-price auction mechanism can be adapted to satisfy U-BNIC, it cannot be combined with a miner payment rule $\tilde{r}$ to further enjoy the $1$-SCP property.

We now define $\bmtheta$ based on the first-price auction. For simplicity, we consider only $n = 2$ users and the block size $k=1$. The first-price auction for the single block entry defines the following allocation rule $\a$ (both first-price and second-price auctions confirm the highest-bid user, also note that $b_{-i}$ is a scalar since there are only $2$ users):
\begin{align}
    a_i(b_i,b_{-i}) &= \begin{cases}
    1, \:b_i > b_{-i}\\
    \frac{1}{2}, \: b_i = b_{-i} \\
    0,\:b_i < b_{-i}
  \end{cases} .
\end{align}

We then consider the payment rules that will help us to finally define $\bmtheta$. The first payment rule $\p$ is the dominant association of $\a$. We calculate $\p$ via Eq.~\eqref{eqn:aux:payment} as follows.
\begin{align}
p_i(b_i,b_{-i}) &= \begin{cases}
b_{-i},&\:b_i \ge b_{-i}\\
0,&\:b_i < b_{-i}
\end{cases}. 
\end{align}
Indeed, $\a$ and $\p$ form the second-price auction which is DSIC. 

We now turn to the second payment rule $\tilde{\p}$ which is adapted from the payment rule of the first-price auction. It is well-known that the first-price auction is not truthful (DSIC) \citep{roughgarden2021transaction}: users would prefer to bid lower than their valuations, which is necessary for them to get any surplus even if they get the item. Nevertheless, there exist Bayesian Nash equilibria for specific settings when distributions of valuations are known. For example, when there are $n$ users with \emph{i.i.d.}\ uniformly random valuations over $[0,1]$, it is a Bayesian Nash equilibrium for each bidder to bid $\frac{n-1}{n}v_i$. By the Revelation Principle \citep{myerson1981optimal, myerson1979incentive}, we can derive a payment rule $\tilde{p}$ to make the confirmed user pay $\frac{n-1}{n}$ times her bid. For $n = 2$, we derive $\tilde{p}$ as follows.
\begin{align}
    \tilde{p}_i(b_i,b_{-i}) &= \begin{cases}
\frac12 b_i &\:b_i \ge b_{-i}\\
0 &\:b_i < b_{-i}
\end{cases}.
\end{align}

Finally, we define $\bmtheta$ according to Eq.~\eqref{eq:def-theta} and get that
\begin{align}
    \theta_i(b_i,b_{-i}) &= \begin{cases}
\frac{1}{2} b_i - b_{-i} &\:b_i > b_{-i}\\
-\frac{1}{4} b_i &\:b_i = b_{-i}   \\
0 &\:b_i < b_{-i}
\end{cases}. 
\end{align}

When the user valuation is uniformly random over $[0, 1]$, we have that $\mathbb{E}_{b_{-i} \sim U[0,1]}[\theta_i(0,b_{-i})] = 0$ for $i \in \{1, 2\}$, indicating that $\bmtheta$ satisfies the second condition (Eq.~\eqref{eq:def-admissibility-2}) of the admissibility property. Suppose that we could find a miner revenue function $\tilde{r}$ such that $T = (\bmtheta, \tilde{r})$ is admissible. Let $M = (\a, \p, 0)$. According to Theorem~\ref{thm:decomposition} and by the definition of $\bmtheta$, we have that the composed TFM 
\[
\tilde{M} = M + T = (\a, \p, 0) + (\bmtheta, \tilde{r}) = (\a, \tilde{\p}, \tilde{r})
\]
is U-BNIC and $1$-SCP. Then we could get the TFM $\tilde{M}$ which is a natural adaptation of the first-price auction (since its payment rule $\tilde{\p}$ is adapted from the first-price payment rule).

On the other hand, however, we show that this is impossible -- there exists no $\tilde{r}$ such that $(\bmtheta, \tilde{r})$ is admissible. We prove this by contradiction. Suppose there exists such an $\tilde{r}$, we compute $\tilde{r}(1,1)$ in two different ways. By the first condition of admissibility (Eq.~\eqref{eq:def-admissibility-1}), we have that
\begin{small}
\begin{align*}
    \tilde{r}(1,1) &= \tilde{r}(0,0) + (\tilde{r}(1,0) - \tilde{r}(0,0)) + (\tilde{r}(1,1) - \tilde{r}(1,0)) \\
    &= \tilde{r}(0,0) + \theta_1(1,0) + \theta_2(1,1) \\
    &= \tilde{r}(0,0) + 0.5-0.25 \\
    &= \tilde{r}(0,0) + 0.25.
\end{align*}
\end{small}
We can also invoke Eq.~\eqref{eq:def-admissibility-1} and compute $\tilde{r}(1,1)$ via a different path:
\begin{small}
\begin{align*}
    &\tilde{r}(1,1)\\
    &= \tilde{r}(0,0) + (\tilde{r}(0.5,0) - \tilde{r}(0,0)) + (\tilde{r}(0.5,1) - \tilde{r}(0.5,0)) \\
    &\qquad\quad+ (\tilde{r}(1,1) - \tilde{r}(0,1)) - (\tilde{r}(0.5,1) - \tilde{r}(0,1)) \\
    &= \tilde{r}(0,0) +\theta_1(0.5,0) + \theta_2(0.5,1) + \theta_1(1,1) - \theta_1(0.5,1) \\
    &= \tilde{r}(0,0) +0.25 + 0 - 0.25 - 0 \\
    &= \tilde{r}(0,0) + 0.
\end{align*}
\end{small}
Now we reach the contradiction. This example shows that using our auxiliary mechanism method, we are not able to extend the natural first-price auction to a U-BNIC and $1$-SCP TFM.\footnote{It is possible to prove a stronger statement: there exists no $\tilde{r}$ such that the TFM $(\a, \tilde{\p}, \tilde{r})$ is $1$-SCP (where $\a$ and $\tilde{\p}$ are defined based on the first-price auction as in Section~\ref{sss:aux:fpa}). Therefore, the does not exist a U-BNIC and $1$-SCP TFM extension based on the first-price auction. We omit the detailed proof of this statement since it is not directly related to the construction and the analysis of our TFM.} We will need to carefully design a different $\bmtheta$ to satisfy the admissibility conditions.

\subsection{A Conservative-field Perspective of the Payment Difference Function $\{\theta_i\}$} \label{sss:aux:ubnic}


In this part, we distill our experience in the trial in Appendix~\ref{sss:aux:fpa} and provide an additional perspective for the design of $\bmtheta$.
From the example, we see that if we sum up the differences of $\bmtheta$ along any path that consists of axis-aligned arcs, the summation should only depend on the two terminals of the path. This suggests the path-independence property of the $\bmtheta$ function. In particular, for any $\bmtheta$ in an admissible variation term $(\bmtheta, \tilde{r})$, if we define the vector field

\begin{align}
    \D_\bmtheta(\b) = \left( \pd{}{b_1}\theta_1(b_1,\b_{-1}), \cdots, \pd{}{b_n}\theta_n(b_n,\b_{-n})\right),
\end{align}
then $\D_\bmtheta$ should be a conservative field \citep{connell1983conservative}. In other words, for any closed curve $C$ (with parametrization $\mathbf{z}$), we have the following equality for the integration
\begin{align}
    \oint_{C} \D_\bmtheta \cdot d \mathbf{z} = 0.
\end{align}
According to Eq.~\eqref{eq:def-admissibility-1}, $\tilde{r}$ is actually the potential of $\D_\bmtheta$. From this conservative-field perspective, we see that in order to successfully construct an admissible variation term, we may consider first constructing a $\tilde{r}$ (as the \emph{potential} that determines the field), while guaranteeing the $\bm{\theta}$ functions satisfies Eq.~\eqref{eq:def-admissibility-2}. This intuition helps our design of the admissible variation term. Nevertheless, it is still quite challenging to construct a good variation term. Thanks to the almost-modular property of the auxiliary mechanism and the variation term, we can re-use an admissible variation term in different settings, as we do in Sections~\ref{sec:size1}-\ref{section:size:k}.

\subsection{Intuition of Variation Term Construction in Section~\ref {sec:alg:one_size}} 
\label{app:am:intuition}

From the admissibility 
condition Eq.~\eqref{eq:def-admissibility-1}, i.e., $\theta_i(b_i,\b_{-i}) = \tilde{r}(b_i,\b_{-i}) - \tilde{r}(0,\b_{-i})$, we get
\begin{equation}
    \tilde{r}(b_i,\b_{-i}) = 
\tilde{r}(b_i,\mathbf{0}) + \theta(b_i,\b_{-i})
\end{equation}

From another admissibility condition of Eq.~\eqref{eq:def-admissibility-2}, i.e., $\E_{\b_{-i}}[\theta_i(b_i,\b_{-i})]=0$, for convenience we decouple $b_i$ and $\b_{-i}$ and construct $\theta_i$ in the following form
\begin{align}\label{eqn:theta:decomp}
    \theta_i(b_i,\b_{-i}) = h\cdot \alpha(b_i)\cdot\beta(\b_{-i}),
\end{align}
in which $\beta(\b_{-i})$ is a symmetric expression on $\b_{-i}$ and
\begin{align}\label{eq:am:intuit:beta}
    \E_{\b_{-i}} [\beta(\b_{-i})] = 0.
\end{align}

Now we consider the case of $\b_{-i}=\mathbf{0}$ and $m$ is large, i.e., the situation is close to a second-price auction in which all other users bid zero, and the user $i$'s payment in the auxiliary mechanism is close to zero.

However, as long as $b_i>0$, by intuition user $i$ is capable of paying more. From the allocation rule, for any fixed $m$ we can actually find a $K>0$ in which $a_i(b_i,\mathbf{0})\ge K b_i$, and hence user $i$ is able to pay at least $a_i(b_i,\mathbf{0}) \cdot b_i \ge K b_i^2$. On the other hand, from Myerson's Lemma (Lemma~\ref{lem:myerson}), in the auxiliary mechanism we also have $a_i(b_i,\mathbf{0})p_i(b_i,\mathbf{0}) = \Theta(b_i^2)$ when $b_i\to 0$, but quickly ``saturating'' when $b_i > \Theta(\frac{1}{m})$ and $a_i(b_i,\mathbf{0})$ become close to $1$. Hence, to uniformly exploit payment from user $i$ for different values of $b_i$, we would like to construct\endnote{We introduced a coefficient $\frac{1}{2}$ because we initially constructed the variation term via partial derivatives.}
\begin{align}\label{eq:alpha}
\alpha(b_i)=\frac{1}{2}b_i^2.
\end{align}

On the other hand, since the expression of $\theta_i(\cdot)$ will appear in the expression of $\tilde{r}(\cdot)$, and $\tilde{r}(\cdot)$ is a symmetric expression. In order to ensure symmetry, we construct
\begin{align}\label{eq:beta}
    \beta(\b_{-i}) = 1 - \mu \sum_{j:j\ne i} b_j^2.
\end{align}

Even if it indicated less payment when $\b_{-i}$ are large on the users' side, the negative fourth order terms in the expression of  $\tilde{r}(\b)$ are ``halved'' compared to the sum of $\{\theta_i(b_i,\b_{-i})\}$, yielding a positive expected miner revenue.

From Eq.\eqref{eq:am:intuit:beta}, we have 
\begin{align}
    \mu &= \frac{1}{\E_{b_{-i}}[\sum_{j:j\ne i}b_j^2]}\\
    &= \frac{1}{c_\rho (n-1)}.\label{eq:mu}
\end{align}

From Eqs.~(\ref{eqn:theta:decomp},\ref{eq:alpha},\ref{eq:beta},\ref{eq:mu}) we get the construction of the variation term as Eqs.~(\ref{eqn:theta},\ref{eqn:r}).

\color{black}

\section{Omitted Proofs}
\ifdefined\ConnectionToSPA
\subsection{Proofs of Observations of Soft Second-Price Mechanism}\label{appendix:proof:obs}

\subsubsection{Proof of Observation~\ref{obs:1}}

    We can obviously see Eq.~(\ref{ssp:lim:alloc}). Now we prove Eq.~(\ref{ssp:lim:payment}).
    
    We first consider the case when $b_i = b_{(1)}$. Notice that when $m\to \infty$, $\frac{\sum_{j=1}^{n}e^{mb_j}}{me^{mb_i}} \sim \frac{e^{mb_i}}{me^{mb_i}} = \frac{1}{m}$, and $\ln \frac{\sum_{j=1}^{n}e^{mb_j}}{1+\sum_{j\neq i}e^{mb_j}} \sim \ln \frac {e^{mb_{(1)}}}{e^{mb_{(2)}}} = \ln e^{m(b_{(1)}-b_{(2)})} = m(b_{(1)}-b_{(2)})$. Therefore, 
    
    \begin{align}
        \lim_{m\to +\infty} p_i(b_i,\b_{-i}) &= \lim_{m\to +\infty} \left(b_{(1)} - \frac{1}{m}\cdot m(b_{(1)}-b_{(2)}) \right)\\
        &= b_{(2)}.
    \end{align}
    
    Then we consider the case of $0 < b_i < b_{(1)}$. In this case, we have $\frac{\sum_{j=1}^{n}e^{mb_j}}{me^{mb_i}} \sim \frac{e^{mb_{(1)}}}{me^{mb_i}} = \frac{1}{m} e^{b_{(1)}-b_i}$ and 
    \begin{align}
        \ln \frac{\sum_{j=1}^{n}e^{mb_j}}{1+\sum_{j\neq i}e^{mb_j}} &= \ln\left(1 + \frac{e^{mb_i}-1}{1+\sum_{j\ne i}e^{mb_j}}\right) \\
        &\sim \ln \left(1+\frac{e^{mb_i}}{e^{mb_{(1)}}}\right) \\
        &\sim \frac{1}{e^{b_{(1)}-b_i}},
    \end{align}
    
    Therefore, $p_i(b_i,\b_{-i}) \sim b_i - \frac{1}{m}$, thus
    
    \begin{align}
        \lim_{m\to +\infty} p_i(b_i,\b_{-i}) &= \lim_{m\to +\infty} \left(b_i - \frac{1}{m}\right) \\
        &= b_i.
    \end{align}
    
    In conclusion, we have Eq.~(\ref{ssp:lim:payment}).

\subsubsection{Proof of Observation~\ref{obs:2}}

    We also only prove Eq.~(\ref{ssp:lim0:payment}). For fixed $\b$ and $m\to 0$, we have
    
    \begin{align}
        p_i(b_i,\b_{-i}) &= b_i - \frac{\sum_{j=1}^{n}e^{mb_j}}{me^{mb_i}} \cdot  \ln\frac{\sum_{j=1}^{n}e^{mb_j}}{1+\sum_{j\neq i}e^{mb_j}} \\
        &= b_i - \frac{\sum_{j=1}^n (1+mb_j + O(m^2))}{m(1+mb_i + O(m^2))} \cdot \ln \frac{\sum_{j=1}^n(1+mb_j+O(m^2))}{1+\sum_{j\ne i}(1+mb_j+O(m^2))}\\
        &= b_i - \frac{n+m\sum_{j=1}^n b_j + O(m^2)}{m(1+mb_i + O(m^2))} \cdot \ln\left(1+\frac{mb_i+O(m^2)}{n+m\sum_{j\ne i}b_j + O(m^2)}\right) \\
        &= b_i - \frac{n}{m} (1+ o(1))\cdot \frac{mb_i}{n}(1+o(1)) \\
        &= o(1) b_i.
    \end{align}
    
    Therefore, 
    
    \begin{align}
        \lim_{m\to 0} p_i(b_i,\b_{-i}) = 0.
    \end{align}

\fi
\subsection{Proof of Theorem~\ref{thm:decomposition}} \label{app:proof:decomposition}

First, we observe a sufficient condition for a TFM to be U-BNIC.

\begin{observation}\label{obs:3}
    $\tilde{M}$ is U-BNIC if 
    \begin{align}
        \mathbb{E}_{\b_{-i}\sim V_{-i}}[\theta_i(b_i,\b_{-i})] =0. \label{eqn:bnic:suff:cond}
    \end{align}
\end{observation}

\proof{Proof.}
    Because user $i$'s expected utility $\tilde{u}(b_i,\b_{-i};v_i) = u(b_i,\b_{-i};v_i) - \theta(b_i,\b_{-i})$, if $\mathbb{E}_{\b_{-i}\sim V_{-i}}[\theta_i(b_i,\b_{-i})]=0$, then for any bidding vector $\b$ and $i$'s valuation $v_i$, mechanisms $M$ and $\tilde{M}$ have the same expected utility 
    \begin{align}
        \mathbb{E}_{b_{-i} \sim V_{-i}}[u(b_i,\b_{-i};v_i)] = \mathbb{E}_{b_{-i} \sim V_{-i}}[\tilde{u}(b_i,\b_{-i};v_i)]. \label{eqn:condition_ubnic}
    \end{align}
    As mechanism $M$ is U-BNIC, it holds that $\tilde{M}$ is also U-BNIC.
    
    \hfill$\square$
\endproof

As we have characterized a sufficient condition for U-BNIC, now we consider the condition for 1-SCP.  We first introduce a lemma as a sufficient and necessary condition for a TFM to be 1-SCP:

\begin{lemma}\label{lem:1scp}
  The mechanism $M=\apr$ is 1-SCP if and only if the following conditions are satisfied:
  \begin{itemize}
      \item Monotone allocation: $a_i(\cdot,\b_{-i})$ is monotonic non-decreasing,
      \item Constrained payment function: 
      \begin{small}
      \begin{align}
      &a_i(b_i,\b_{-i})p_i(b_i,\b_{-i}) - r(b_i,\b_{-i}) \nonumber\\
      &= \int_{0}^{b_i} t \pd{a_i(t,\b_{-i})}{t} dt + a_i(0,\b_{-i})p_i(0,\b_{-i}) - r(0,\b_{-i}).
      \end{align}
      \end{small}
  \end{itemize}
\end{lemma}

\proof{Proof.}
    Consider another mechanism $M'=(\mathbf{a},\mathbf{p}-\frac{r}{\mathbf{a}},0)$. Since $M'$ has zero miner revenue, it is 1-SCP if and only if it is U-DSIC.
    
    From Lemma~\ref{lem:myerson}, $M'$ is U-DSIC if and only if the given conditions hold. So $M'$ is 1-SCP if and only if the conditions hold.
    
    Notice that for the same bidding vector $\b$, the miner and user $i$ have the same total utilities in mechanisms $M$ and $M'$. So $M$ is 1-SCP if and only if the conditions hold.
    
    \hfill$\square$
\endproof

From Lemma~\ref{lem:1scp} we know that for an 1-SCP mechanism $(\mathbf{a},\tilde{\mathbf{p}},\tilde{\mathbf{r}})$, if we fix $\b_{-i}$, the difference of $a_i(\cdot,\b_{-i})\tilde{p}_i(\cdot,\b_{-i})$ and $\tilde{r}_i(\cdot,\b_{-i})$ is a constant. Furthermore, since $a(\cdot,\b_{-i})$ is monotonic increasing, if we want $\tilde{M}$ to be 1-SCP, from Lemma~\ref{lem:1scp} we need and only need:
\begin{small}
\begin{align}
      &a_i(b_i,\b_{-i})\tilde{p}_i(b_i,\b_{-i}) - \tilde{r}(b_i,\b_{-i}) \nonumber\\
      &= \int_{0}^{b_i} t \pd{a_i(t,\b_{-i})}{t} dt + a_i(0,\b_{-i})\tilde{p}_i(0,\b_{-i}) - \tilde{r}(0,\b_{-i}).
      \end{align}
\end{small}
From the construction of $\p$ we have

\begin{align}
      a_i(b_i,\b_{-i})p_i(b_i,\b_{-i})  = \int_{0}^{b_i} t \pd{a_i(t,\b_{-i})}{t} dt.
\end{align}

Since we set the boundary condition $\tilde{p}_i(0,\b_{-i})=0$, and the definition of $\{\theta_i\}$ as $\theta_i(b_i,\b_{-i}) = a_i(b_i,\b_{-i})(\tilde{p}_i(b_i,\b_{-i})-p_i(b_i,\b_{-i}))$, we get a sufficient condition of 1-SCP as:
\begin{align}
    \theta_i(b_i,\b_{-i}) = \tilde{r}(b_i,\b_{-i}) - \tilde{r}(0,\b_{-i}), \quad \forall i. 
\end{align}
So $\tilde{M}$ is indeed U-BNIC and 1-SCP if $M$ is U-DSIC and 1-SCP and $T$ is admissible.

\subsection{Proof of Lemma~\ref{lem:admis}} \label{app:proof:admis}
We have
\begin{align}
\tilde{r}(b_i, \b_{-i}) - \tilde{r}(0, \b_{-i}) &= \frac{1}{2} h \left( b_i^2 - \frac{\sum_{j\ne i} b_i^2 b_j^2}{c_\rho(n-1)}\right) \\
&= \frac{1}{2} h b_i^2 \left( 1 - \frac{\sum_{j\ne i}  b_j^2}{c_\rho(n-1)} \right) \\
&= \theta_i(b_i, \b_{-i})
\end{align}
and
\begin{align}
&\mathbb{E}_{\b_{-i} \sim V_{-i}} \theta_i(b_i, \b_{-i}) \nonumber\\
&= \mathbb{E}_{\b_{-i} \sim V_{-i}}\left[-\frac{1}{2}hb_i^2\left(\frac{\sum_{j\ne i}b_j^2}{c_\rho(n-1)}-1\right) \right] \\
&= -\frac{1}{2}hb_i^2\left(\frac{\sum_{j\ne i}\mathbb{E}_{\b_{-i} \sim V_{-i}}[b_j^2]}{c_\rho(n-1)}-1\right) \\
&= -\frac{1}{2}hb_i^2\left(\frac{\sum_{j\ne i}c_\rho}{c_\rho(n-1)}-1\right) \\
&= 0.
\end{align}
Therefore, the variation term $T$ is admissible.

\subsection{Proof of Theorem~\ref{thm:h:value:1}} \label{app:proof:param:1}

From the auxiliary mechanism method, the mechanism $\tilde{M}=(\a,\tilde{\p},r)$ is U-BNIC and 1-SCP from Theorem~\ref{thm:decomposition}. Now we prove the UIR, BF and U-SP properties.

\subsubsection{Proof of UIR and BF}

From Eq.~(\ref{ssp_payment}) we know $p_i(0,\b_{-i})=0$. Then for $n\to \infty$, from Lemma~\ref{lem:myerson} {and $b_i\in [0,1]$} we get:

    \begin{align}
              a_i(b_i,\b_{-i})p_i(b_i,\b_{-i}) &= \int_{0}^{b_i} t \pd{a_i(t,\b_{-i})}{t} dt \\
              &= \int_{0}^{b_i} t\cdot \frac{e^t \sum_{j\ne i} e^{b_j}}{\left(e^t+\sum_{j\ne i} e^{b_j}\right) ^2} dt.
    \end{align}

    Since $t\in [0,1]$, it holds that

    \begin{align}
        \frac{e^t \sum_{j\ne i} e^{b_j}}{e^t+\sum_{j\ne i} e^{b_j} } &= \left( \frac{1}{e^t} + \frac{1}{\sum_{j\ne i} e^{b_j}} \right)^{-1} \\
        &\ge \left( \frac{1}{1} + \frac{1}{n-1} \right)^{-1} \\
        &= \frac{n-1}{n}.
    \end{align}

    Combined with $e^t+\sum_{j\ne i} e^{b_j} \le en$, we have
    \begin{align}
        \frac{e^t \sum_{j\ne i} e^{b_j}}{\left(e^t+\sum_{j\ne i} e^{b_j}\right) ^2} \ge \frac{n-1}{en^2}.
        \label{eqn:thm3:estim1}
    \end{align}
    Hence,
        \begin{align}
              a_i(b_i,\b_{-i})p_i(b_i,\b_{-i}) &\ge \int_{0}^{b_i} t \frac{n-1}{en^2 } dt \\
              &= \frac{n-1}{2en^2}\cdot b_i^2.
              \label{eqn:auxiliary:payment}
    \end{align}
    
    Therefore, the difference of the total collected fee and miner revenue in $\tilde{M}$ is
    \begin{align}
        &~~~~\sum_{i=1}^n a_i(b_i,\b_{-i})\tilde{p}_i(b_i,\b_{-i}) - \tilde{r}(\b) \nonumber \\
        &= \sum_{i=1}^n a_i(b_i,\b_{-i})p_i(b_i,\b_{-i}) + \sum_{i=1}^n \theta_i(b_i,\b_{-i}) -\tilde{r}(\b) \label{eqn:similar_begin} \\
        &\ge \frac{n-1}{2en^2}\cdot \sum_{i=1}^{n} b_i^2 -\left(h \frac{\sum_{1\le i < j \le n}b_i^2 b_j^2}{c_\rho (n-1)} - \frac{h}{2}\sum_{i=1}^n b_i^2 \right) \nonumber  \\
        &\quad~~~ - \frac{1}{2} h \left(\sum_{i=1}^{n}b_i^2 - \frac{\sum_{1\le i < j \le n}{b_i^2b_j^2}}{c_\rho(n-1)}\right)  \\
        &= \frac{n-1}{2en^2}\cdot \sum_{i=1}^{n} b_i^2 - \frac{h}{2c_\rho (n-1)} \sum_{1\le i < j \le n}b_i^2 b_j^2 \\
        & \ge \frac{n-1}{2en^2}\cdot \sum_{i=1}^{n} b_i^2 - \sum_{i=1}^n b_i^2 \left(\frac{h}{4c_\rho (n-1)}\sum_{i=1}^n b_i^2\right) \\
        & \ge \frac{n-1}{2en^2}\cdot \sum_{i=1}^{n} b_i^2 - \sum_{i=1}^n b_i^2 \left(\frac{h}{4c_\rho (n-1)}\cdot n\right) \\
        &= \sum_{i=1}^{n} b_i^2 \cdot \left( \frac{n-1}{2en^2} - \frac{hn}{4c_\rho(n-1)}\right).
    \end{align}
So $\tilde{M}$ is budget feasible as long as $h \le  \frac{2c_\rho (n-1)^2}{en^3}=\Theta(c_\rho/n)$.

    \color{black}

For user individual rationality, 
\vspace{-4em}
\begin{small}
\begin{align}
    &~~~~~b_i-\tilde{p}_i(b_i,\b_{-i}) \nonumber \\
    &= b_i - p_i(b_i,\b_{-i}) - \frac{\theta_i(b_i,\b_{-i})}{a_i(b_i,\b_{-i})} \\
    &= \frac{1}{a_i(b_i,\b_{-i})} \Big[b_i\left(a_i(0,\b_{-i}) + \int_{0}^{b_i} \pd{a_i(t,\b_{-i})}{t} dt\right)\nonumber\\
    &\qquad\qquad\qquad\qquad - \int_{0}^{b_i} t\pd{a_i(t,\b_{-i})}{t}\Big] \\
    &= \frac{1}{a_i(b_i,\b_{-i})} \Big[b_i a_i(0,\b_{-i}) + \int_{0}^{b_i} (b_i - t)\pd{a_i(t,\b_{-i})}{t} dt \nonumber\\
    &\qquad\qquad\qquad\qquad + \frac{1}{2} h b_i^2 \left(\frac{\sum_{j\ne i} b_j^2}{c_\rho(n-1)}-1\right)\Big].
 \label{eqn:similar_end}
\end{align}
\end{small}

From Eq.\eqref{eqn:thm3:estim1}, we also have
\begin{align}
&~~~~~\int_{0}^{b_i} (b_i-t) \pd{a_i(t,\b_{-i})}{t} dt \nonumber\\
&= \int_{0}^{b_i} (b_i-t)\cdot \frac{e^t \sum_{j\ne i} e^{b_j}}{\left(e^t+\sum_{j\ne i} e^{b_j}\right) ^2} dt \\
&\ge \int_{0}^{b_i} (b_i-t)\cdot \frac{n-1}{en^2} dt \\
&=  \frac{n-1}{2en^2}\cdot b_i^2.
\end{align}

Therefore, when $h = \frac{2c_\rho (n-1)^2}{en^3}$, since $c_\rho\le 1$, we have
\begin{align}
&~~~~~\tilde{u}_i(b_i,\b_{-i};b_i)\nonumber\\
&= b_i a_i(0,\b_{-i})+ \int_{0}^{b_i} (b_i - t)\pd{a_i(t,\b_{-i})}{t} dt \nonumber\\
&\qquad + \frac{1}{2} h b_i^2 \left(\frac{\sum_{j\ne i} b_j^2}{c_\rho(n-1)}-1\right) \\
&\ge \frac{b_i}{en} + \int_{0}^{b_i} (b_i - t)\pd{a_i(t,\b_{-i})}{t} dt \nonumber\\
&\qquad+ \frac{1}{2} h b_i^2 \left(\frac{\sum_{j\ne i} b_j^2}{c_\rho(n-1)}-1\right) \\
&\ge b_i^2 \left(\frac{1}{en}+\frac{n-1}{2en^2}-\frac{h}{2}\right) \\
&= b_i^2 \left(\frac{1}{en}+\frac{n-1}{2en^2}-\frac{c_\rho (n-1)^2}{en^3}\right)\\
&\ge b_i^2 \left(\frac{1}{en}+\frac{n-1}{2en^2}-\frac{1}{en}\right) \\
&\ge 0.
\end{align}

So the UIR also holds for $h = \frac{2c_\rho (n-1)^2}{en^3}$.

Therefore, we have shown $h_*(n,c_\rho) \ge \frac{2c_\rho (n-1)^2}{en^3} = \Omega (c_\rho / n)$.

\color{black}

\subsubsection{Proof of U-SP}\label{app:proof:size1:usp}
~

We firstly consider the auxiliary mechanism $(\a,\p,r)$. Denote $w_{-i} = \sum_{j\ne i} e^{m b_j}$, then we have
\begin{align}
    a_i(b_i,\b_{-i}) &= \frac{e^{m b_i}}{e^{m b_i} + w_{-i}} \\
    p_i(b_i,\b_{-i}) &= b_i - \frac{e^{m b_i} + w_{-i}}{m e^{m b_i}} \ln \frac{e^{m b_i} + w_{-i}}{1+w_{-i}}.
\end{align}

The utility of identity $i$ is $u_i(b_i,\b_{-i}; v_i) = a_i(b_i,\b_{-i}) (v_i - p_i(b_i,\b_{-i}))$. We can also regard as it as a function of $(b_i, w_{-i}, v_i)$, then we have
\begin{align}
    \left.\pd{u_i}{w_{-i}}\right|_{b_i = v_i} = \frac{-\frac{1}{1+w_{-i}} + \frac{1}{e^{m}+w_{-i}}}{m} \le 0.
\end{align}
As injecting fake bids is equivalent to increasing $w_{-i}$ for identity $i$ in the auxiliary mechanism, it cannot increase identity $i$'s utility in the auxiliary mechanism.

However, the injected fake bids can influence user $i$'s utility in two more aspects, as:

\begin{itemize}
    \item The variation term.
    \item The utilities of fake identities.
\end{itemize}

We denote $h$ as the scaling parameter for total user number $n+l$, hence, we have
\begin{align}
    h\le \frac{2c_\rho(n+l-1)^2}{e(n+l)^3}.
\end{align}

Without fake identities, the expectation of $\theta_i(b_i,\b_{-i})$ is zero. Therefore, denote $\Omega = \{i\} \cup \{n+1,\cdots, n+l\}$, then $\Omega$ is the set of all identities that the user has access to, and we only need to show that
\begin{align}
    \E_{\b_{-i}\sim V_{-i}}\Bigg[&\sum_{j=n+1}^{n+l} a_j(b_j,\b_{-j}^+)p_j(b_j,\b_{-j}^+)\nonumber \\
    &+ \sum_{j\in \Omega}  \theta(b_j,\b_{-j}^+)\bigg] \ge 0.
    \label{eqn:usp:cond}
\end{align}

For a refined analysis of constants, we denote the Sybil attacker has real identity $i$, and submits fake bids with identities $n+1,\cdots, n+l$. We denote that: 
\begin{align*}
    \sigma = \sum_{j\le n, j\ne i} b_j^2, \\
    \sigma_\# = \sum_{j=n+1}^{n+l} b_j^2,
\end{align*}

Then $\sigma$ is a random variable independent to any $b_j$ for $j\in \Omega$, and it holds that 
\begin{equation*}
\E_{\b_{-i}\sim V_{-i}} [\sigma] = c_\rho (n-1).
\end{equation*}

From Eq.~(\ref{eqn:auxiliary:payment}) , we have
\begin{align}
\sum_{j=n+1}^{n+l} a_j(b_j,\b_{-j}^+)p_j(b_j,\b_{-j}^+) &\ge \frac{n+l-1}{2e(n+l)^2} \sum_{j=n+1}^{n+l} b_j^2 \\
& \ge \frac{1}{2e(n+l+2)}\cdot \sigma_\#.\label{eqn:usp:1}
\end{align}

For $j\in \Omega$,
we have
\begin{align}
    \theta(b_j,\b_{-j}^+) &= -\frac{1}{2}hb_j^2 \left(\frac{\sum_{t\le n+l, t\ne j} b_t^2}{c_\rho(n+l-1)}-1\right) \\
    &= -\frac{1}{2}hb_j^2 \left(\frac{\sigma + \sigma_\# + b_i^2 - b_j^2}{c_\rho(n+l-1)}-1\right).
\end{align}

Here, $\sigma$ is the only random variable in the expression, and
\begin{align}
&~~~~~\E_{\b_{-i}\sim V_{-i}}\left[\sum_{j\in \Omega}  \theta(b_j,\b_{-j}^+)\right]\nonumber\\
&=\E_{\b_{-i}\sim V_{-i}}\left[-\sum_{j\in\Omega} \frac{1}{2}hb_j^2 \left(\frac{\sigma + \sigma_\# + b_i^2 - b_j^2}{c_\rho(n+l-1)}-1\right)\right]\\
&=-\frac{h}{2}\sum_{j\in\Omega}b_j^2 \cdot\left(\frac{\E_{\b_{-i}\sim V_{-i}}\left[\sigma\right] + \sigma_\# + b_i^2}{c_\rho(n+l-1)}-1\right)
\nonumber\\
&\qquad + \frac{h}{2}\sum_{j\in \Omega}\frac{b_j^4}{c_\rho(n+l-1)} \\
&= -\frac{h}{2}(\sigma_\# + b_i^2)\cdot\frac{\sigma_\# + b_i^2-c_\rho l}{c_\rho(n+l-1)}+ \frac{h}{2}\cdot\frac{\sum_{j\in \Omega}b_j^4}{c_\rho(n+l-1)}
\\
&= \frac{h}{2c_\rho(n+l-1)}\Bigg(-\sigma_\#^{2} - 2b_i^2 \sigma_\# -b_i^4 \nonumber\\
&\qquad\qquad\qquad\quad~~~~+c_\rho l (\sigma_\# + b_i^2) + b_i^4 + \sum_{j=n+1}^{n+l} b_j^4\Bigg)\\
&\ge \frac{h}{2c_\rho(n+l-1)}\left(-\sigma_\#^{2} - 2b_i^2 \sigma_\#\right)\\
&\ge \frac{n+l-1}{e(n+l)^3}\left(-\sigma_\#^{2} - 2 \sigma_\#\right).
\label{eqn:usp:2}
\end{align}

From Eqs.~(\ref{eqn:usp:1},\ref{eqn:usp:2}), Eq.~\eqref{eqn:usp:cond} is implied by
\begin{align}
    \frac{1}{2e(n+l+2)}\sigma_\# \ge \frac{n+l-1}{e(n+l)^3}(\sigma_\#^2 +2 \sigma_\#).
\end{align}

Noticing that $\forall b_i\le 1$, so $\sigma_\# \le l$. We only need
\begin{align}
    (n+l)^3 \ge 2(n+l+2)(n+l-1)(l+2).
\end{align}

Now for any $C\in [0,1)$, we assume $n\ge \frac{6C+5}{1-C^2}$, then denote $\varphi = \frac{l}{n} \le C$, and we have

\begin{align}
    &(1+\varphi)^3 n^3 - 2((1+\varphi)n+2)((1+\varphi)n-1)(\varphi n+2) \\
    &=(1+\varphi)((1-\varphi^2)n-(6\varphi+4))n - 4) n + 8 .
\end{align}

Since $n\ge \frac{6C+5}{1-C^2}$, we see that $n\ge 5$, and $(1-\varphi^2)n - (6\varphi + 4) \ge (1-C^2)n - (6C + 4) \ge 1$. Hence,
\begin{align}
&(1+\varphi)((1-\varphi^2)n-(6\varphi+4))n - 4) n + 8 \\
&\ge 1\cdot (1\cdot n - 4) n +8\\
&\ge 13 \\
&> 0.
\end{align}

Now we prove that the mechanism is $(C,\frac{6C+5}{1-C^2})$-U-SP for any $C\in [0,1)$.

\color{black}

\subsection{Proof of Theorem~\ref{thm:h:value:k}} \label{app:proof:param:k}
From the auxiliary mechanism method, we have the U-BNIC and 1-SCP properties as long as the allocation rule is monotone. Hence, our proof for Theorem consists of 3 parts:
\begin{itemize}
    \item Proof of monotonicity of allocation rule.
    \item Proof of UIR and BF.
    \item Proof of U-SP.
\end{itemize}

\subsubsection{Proof of Monotonicity of Allocation Rule.}

For monontonicity, we just need to show that for any $\b_{-i}$, $a_i(b_i,\b_{-i}) \ge a_i(b'_i,\b_{-i})$ if $b_i \ge b'_i$.

If $1\le n \le k$, we have $a_i(b_i, \b_{-i}) = a_i(b'_i, \b_{-i}) = 1$, so the monotonicity holds. Now we consider $n>k$.

For convenience denote $w_i = e^{m b_i}$ and without loss of generality we assume $i=n$. Now For any map $X: \mathbb{N}_+ \to [0,1)$, vector $\mathbf{t}$ s.t. $0=t_0<t_1<t_2<\cdots<t_{n-1}<t_n=1$, $B_0 \subseteq [0,1)$ and $k \le n-1$, define an algorithm as Algorithm~\ref{Algo:draw}:

\setlength{\textfloatsep}{8pt}
\def\la{\leftarrow}
\begin{algorithm}[htb] 
\caption{$Draw(X, \mathbf{t}, B_0, k)$}
\label{Algo:draw}
\begin{algorithmic}[1]
\State Input $X, \mathbf{t}, B_0, k$;
\State $B \la B_0$;  $S \la \emptyset$; 
\State $u \la 1$;  $v \la 1$;
\While{$v\le k$}
    \State $x \la X(u)$;
    \If{$x \notin B$}
        \State Find $i$ s.t. $x \in [t_{i-1} , t_i)$;
        \State $S \la S \cup i$;     
        \State $v\la v+1$;
    \EndIf
    \State $B \la B \cup [t_{x-1}, t_x)$;
    \State $u\la u+1$;
\EndWhile
\State Output $S$;
\end{algorithmic}
\end{algorithm}
\def\t{\mathbf{t}}
Now we denote $W_i = \frac{w_i}{\sum_{i=1}^n w_i}$ for $1\le i \le n$, and 

\[
    W'_i = \begin{cases}
    W_i,\quad & i\le n-1 \\
    \frac{e^{m b'_{n}}}{\sum_{i=1}^n w_i}, \quad & i=n. 
    \end{cases}
\]

Then, we define $\t, \t'$ as
\begin{align*}
    t_i &= \sum_{j=1}^i W_j,\quad 0\le i \le n \\
    t'_i &= \begin{cases}
        \sum_{j=1}^i W'_j,\quad & 0\le i \le n \\
        1,\quad & i = n+1.
    \end{cases}
\end{align*}

Then when $X$ is a \emph{i.i.d.}\ uniform random sequence in $[0,1)$, we can see that

\begin{itemize}
    \item $Draw(X, \mathbf{t}, \emptyset, k)$ randomly samples $k$ items among $\{1, \cdots, n\}$ with weights $\{W_i\}$ without replacement.
    \item $Draw(X, \mathbf{t}', [t'_{n}, 1), k)$ randomly samples $k$ items among $\{1, \cdots, n\}$ with (relative) weights $\{W'_i\}_{i\in [n]}$ without replacement.
\end{itemize}

In fact, Algorithm~\ref{Algo:draw} performs random drawing without replacement in the following way. Every round an item in $\{1, \cdots, n\}$ is drawn, and in the second scenario the total weights is less than 1 so that a ``placeholder'' item $n+1$ with weight $1-t'_n$ is added. If the item is already drawn or is the ``placeholder'', we draw again; other wise, we finalize it and add it to $S$.

In the rest of the proof, we prove that
\begin{align*}
&\Pr[n \in Draw(X, \mathbf{t}, \emptyset, k)]\\
&\qquad \ge  \Pr[n \in Draw(X, \mathbf{t}', [t'_{n}, 1), k)]
\end{align*}
by actually showing
\[
n \in Draw(X, \mathbf{t}', [t'_{n}, 1), k)  \implies  n \in Draw(X, \mathbf{t}, \emptyset, k).
\]

In fact, assume $n \in Draw(X, \mathbf{t}', [t'_{n}, 1), k)$. By the time the drawing process $Draw(X, \mathbf{t}', [t'_{n}, 1), k)$ stops, if no value $X(u) \in [t'_{n}, 1)$ is obtained, then $Draw(X, \mathbf{t}, \emptyset, k)$ has exactly the same outcome, so it also contains $n$.

If in some round $X(u) \in [t'_{n}, 1)$ is obtained in $Draw(X, \mathbf{t}', [t'_{n}, 1), k$, we consider the first round that happens. 

Before that round, $Draw(X, \mathbf{t}, \emptyset, k)$ have the same outcome, so it is not stopped either. In that round, $Draw(X, \mathbf{t}, \emptyset, k)$ adds $n$ to $S$, so $n \in Draw(X, \mathbf{t}, \emptyset, k)$.

So we have shown that $n \in Draw(X, \mathbf{t}', [t'_{n}, 1), k)  \implies  n \in Draw(X, \mathbf{t}, \emptyset, k)$, implying the monotonicity of the allocation rule.

\subsubsection{Proof of UIR and BF.}

From Lemma~\ref{lem:myerson}, similar to the case of block size 1, we essentially need to derive a lower bound on $\pd{a_i(t,\b_{-i})}{t}$, in order to lower bound the total payment. Therefore, we only need to analyze the partial derivative of $\delta_t(i;j)$ on $w_i$.

When we fix $j$ and $\b_{-i}$ (i.e., $\mathbf{w}_{-i}$), we can regard $\delta_t(i;j)$ as a function of $w_i$. Here we make a notation of $X_s$ for $0\le s \le k-1$ as

\begin{align}
    X_s = W-w_i - \sum_{z=1}^{s} w_{j_z},
\end{align}

then $X_s$ is a constant.

From Eq.~(\ref{eqn:delta_component}) we get (note that $W-\sum_{z=1}^{s}w_{j_s} = X_s + w_i$)

\begin{align} \label{eqn:delta_comp}
    \pd{\delta_t(i;j)}{w_i} = \left(\prod_{s=1}^{t-1} w_{j_s}\right) \cdot \pd{}{w_i}\frac{w_i}{\prod_{s=0}^{t-1}(X_s+w_i)},
\end{align}

and 
\begin{small}
\begin{align}
    &~~~~~\pd{}{w_i}\frac{w_i}{\prod_{s=0}^{t-1}(X_s+w_i)} \nonumber\\
    &= \pd{}{w_i}\left(w_i \cdot \prod_{s=0}^{t-1}\frac{1}{X_s+w_i}\right) \\
    &= \prod_{s=0}^{t-1}\frac{1}{X_s+w_i} + w_i\cdot\pd{}{w_i}\prod_{s=0}^{t-1}\frac{1}{X_s+w_i} \\
    &= \prod_{s=0}^{t-1}\frac{1}{X_s+w_i} - w_i \cdot \left( \sum_{s=0}^{t-1} \frac{1}{X_s+w_i}\right) \cdot \left( \prod_{s=0}^{t-1} \frac{1}{X_s+w_i}\right) \\ 
    &= \left(\prod_{s=0}^{t-1}\frac{1}{X_s+w_i}\right) \cdot\left(1- w_i \sum_{s=0}^{t-1} \frac{1}{X_s+w_i}\right)
\end{align}
\end{small}
Notice that $X_s+w_i$ is a sum of $(n-s)$ weights, each one no less than $1$, so $\frac{1}{X_s+w_i} \le \frac{1}{n-s} \le \ln\frac{n-s}{n-s-1}$, and $w_i=e^{mb_i}\le e^m$. Therefore,

\begin{align}
    1- w_i \sum_{s=0}^{t-1} \frac{1}{X_s+w_i} &\ge 1 - e^m \sum_{s=0}^{t-1}\ln\frac{n-s}{n-s-1} \\
    &= 1-e^m \ln \frac{n}{n-t}.
\end{align}

Denote
\begin{align}
    D(m, \lambda) = 1- e^m \ln \frac{\lambda}{\lambda-1},
\end{align}

then $\forall \frac{n}{k}<\frac{e}{e-1}$, $\exists m>0$ s.t. $D\left(m,\frac{n}{k}\right)>0$.

Therefore, from Eq.~(\ref{eqn:delta_comp}) we have 

\begin{align}
    &\pd{\delta_t(i;j)}{w_i} \ge  D\left(m,\frac{n}{k}\right) \left(\prod_{s=1}^{t-1} w_{j_s}\right)  \left(\prod_{s=0}^{t-1}\frac{1}{X_s+w_i}\right) \\
    =& D\left(m,\frac{n}{k}\right) \cdot \frac{w_{j_1}}{X_0+w_i} \cdot \frac{w_{j_2}}{X_1+w_i} \cdot \cdots \cdot \frac{1}{X_{t-1}+w_i}. 
\end{align}

We notice that $\frac{w_{j_1}}{X_0+w_i} \cdot \frac{w_{j_2}}{X_1+w_i} \cdot \cdots \cdot \frac{w_{j_{t-1}}}{X_{t-2}+w_i}$ is just the probability that the sampling outcome of the first $t-1$ rounds are $(j_1,j_2,\cdots,j_{t-1})$, denoted as $P(j_{[t-1]})$. Furthermore, from $X_{t-1}+w_i \le e^m\cdot n$, we have

\begin{align}
    \pd{\delta_t(i;j)}{w_i} \ge \frac{D\left(m,\frac{n}{k}\right)}{e^m n} P(j_{[t-1]}).
\end{align}

Therefore from Eq.~(\ref{eqn:delta_decomposition}):

\begin{align}
    \pd{\delta_t(i)}{w_i} &= \sum_{j\in J_t(i)} \pd{\delta_t(i;j)}{w_i} \\
    &\ge \frac{D\left(m,\frac{n}{k}\right)}{e^m n} \sum_{j\in J_t(i)} P(j_{[t-1]}).
\end{align}

For $j\in J_t(i)$, we observe that $j_{[t-1]}$ iterates through all $(t-1)$-permutations of $[n]$ that does not contain element $i$. Therefore, $\sum_{j\in J_t(i)} P(j_{[t-1]})$ is the probability that $i$ is not chosen in the first $(t-1)$ rounds.

To compute the probability that $i$ is not chosen in the first $(t-1)$ rounds, we consider each round. In each round, there are at least $(n-k)$ users each with weight at least $1$, and user $i$ has weight at most $e^m$, so $i$ is chosen with probability at most $\frac{e^m}{n-k}$. Therefore for $t$ rounds, the probability that $i$ is not ever chosen is at most $\left(1-\frac{e^m}{n-k}\right)^t \ge \left(1-\frac{e^m}{n-k}\right)^k = (1-o(1)) e^{-\frac{e^m k}{n-k}}$. That implies:

\begin{align}
    \pd{\delta_t(i)}{w_i} \ge (1-o(1)) \frac{D\left(m,\frac{n}{k}\right)}{e^m n} e^{-\frac{e^m k}{n-k}},
\end{align}

so
\begin{small}
\begin{align}
    \pd{a_i(b_i,\b_{-i})}{b_i} &= \pd{w_i}{b_i} \cdot \pd{a_i(b_i,\b_{-i})}{w_i} \\
    &= me^{mb_i} \cdot \sum_{t=1}^k\pd{\delta_t(i)}{w_i} \\
    &\ge  me^{mb_i} \cdot \sum_{t=1}^k \left((1-o(1)) \frac{D\left(m,\frac{n}{k}\right)}{e^m n} e^{-\frac{e^m k}{n-k}}\right) \\ 
    &= \frac{k}{n} \left((1-o(1)) me^{mb_i}  \frac{D\left(m,\frac{n}{k}\right)}{e^m } e^{-\frac{e^m k}{n-k}}\right)
\end{align}
\end{small}
For any fixed $\lambda_0 > \frac{e}{e-1}$, let $\lambda = \frac{n}{k}$. If $\lambda\ge \lambda_0$, let 

\begin{align}
    m = m_\#(\lambda_0) = \min\left\{\frac{1}{2} \ln \frac{1}{\ln\frac{\lambda_0}{\lambda_0-1}},1\right\}
\end{align}

be a constant. 
{Then we have:
\begin{align}
    D(m,\lambda) = \max\left\{1-\sqrt{\ln\frac{\lambda}{\lambda-1}},1-e\ln\frac{\lambda}{\lambda-1}\right\}.
\end{align}}

Because $m_\#(\cdot)$ and $D(m,\cdot)$ are non-decreasing, we have
\begin{small}
\begin{align}
\pd{a_i(b_i,\b_{-i})}{t} &\ge \frac{k}{n} \left((1-o(1)) me^{mb_i}  \frac{D\left(m,\frac{n}{k}\right)}{e^m } e^{-\frac{e^m k}{n-k}}\right) \\
&\ge \frac{k}{n} \left((1-o(1)) m   \frac{D\left(m,\lambda_0\right)}{e } e^{-\frac{e^m}{\lambda_0 - 1}}\right).
\end{align}
\end{small}
Because $m,\lambda_0,D(m,\lambda_0)$ are all positive constants, we get
\begin{align}
    \pd{a_i(b_i,\b_{-i})}{t} \ge \frac{k}{n} f(\lambda_0)(1-o(1)).
\end{align}

Therefore, from Lemma~\ref{lem:myerson} and $p_i(0,\b_{-i})=0$, we get
\begin{align}
    a_i(b_i,\b_{-i})p_i(b_i,\b_{-i}) &= \int_{0}^{b_i} t \pd{a_i(t,\b_{-i})}{t} dt \\
      &\ge \int_{0}^{b_i} t \frac{k}{n} f(\lambda_0)(1-o(1)) dt \\
      & = f(\lambda_0) \Theta\left(\frac{k}{n}b_i^2 \right). \label{eqn:k:payment}
\end{align}

\begin{figure}[tb]
    \centering
    \includegraphics[width = 0.6\textwidth]{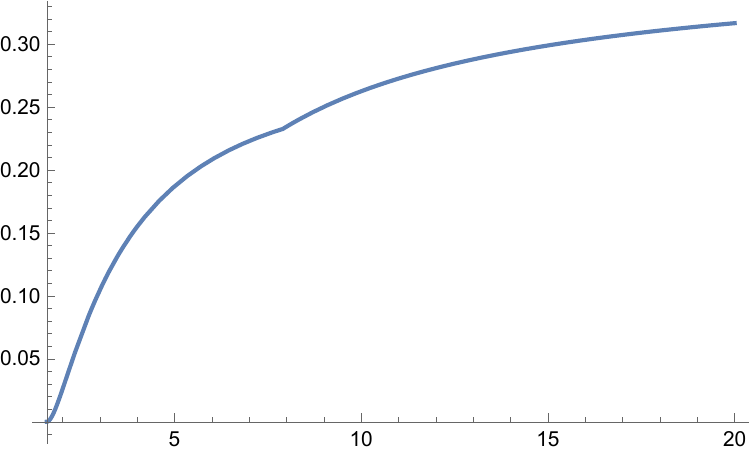}
    \caption{The plot of $f(\cdot)$.}
    \label{fig:thm4_plot}
\end{figure}

Here, the expression of $f(\cdot)$ is given by
\begin{align}
    f(\lambda)=\frac{m_\#(\lambda) D(m_\#(\lambda),\lambda)}{e^{m_\#(\lambda)}}
\end{align}
and can be plotted as in Figure~\ref{fig:thm4_plot}. It can be noticed that $f(\cdot)$ is monotonic increasing and
\begin{align}
\lim_{\lambda\to +\infty} f(\lambda) = \frac{1}{e}.
\end{align}
\color{black}

Then, when we let $\tilde{p}_i(b_i,\b_{-i}) = p_i(b_i,\b_{-i})+ \frac{\theta_i(b_i,\b_{-i})}{a_i(b_i,\b_{-i})}$ while using the variation term of Eqs.~(\ref{eqn:theta}-\ref{eqn:r}), similar to the argument of Eqs.~(\ref{eqn:similar_begin}-\ref{eqn:similar_end}), we can get the UIR and BF properties.

\textbf{Detailed constant analysis.}

From the assumption that $n\ge 30$ and $n>\frac{e}{e-1}k$, we have $n-k\ge 3>e\ge e^m$. Since
\begin{align*}
    (1-\alpha)^k = (1+\frac{\alpha}{1-\alpha})^{-k} \ge e^{-\frac{k\alpha}{1-\alpha}}, ~~\alpha\in[0,1)
\end{align*}

we have
\begin{align}
    \left(1-\frac{e^m}{n-k}\right)^k \ge e^{-\frac{e^m k}{n-k-e^m}}.
\end{align}

Then we get that
\begin{align}
    \pd{\delta_t(i)}{w_i} \ge  \frac{D\left(m,\frac{n}{k}\right)}{e^m n} e^{-\frac{e^m k}{n-k-e^m}},
\end{align}

Since $n\ge 30$, we have
\begin{align}
\pd{a_i(b_i,\b_{-i})}{t} &\ge \frac{k}{n} \left(me^{mb_i}  \frac{D\left(m,\frac{n}{k}\right)}{e^m } e^{-\frac{e^m k}{n-k-e^m}}\right) \\
&> \frac{k}{n} \left(me^{mb_i}  \frac{D\left(m,\frac{n}{k}\right)}{e^m } e^{-\frac{e^m k}{(n-3)-k}}\right) \\
&\ge \frac{k}{n} \left(me^{mb_i}  \frac{D\left(m,\frac{n}{k}\right)}{e^m } e^{-\frac{e^m k}{0.9n-k}}\right) \\
&\ge \frac{k}{n} \left( m   \frac{D\left(m,\lambda_0\right)}{e^m } e^{-\frac{e}{0.9\lambda_0 - 1}}\right)\\
&= \frac{k}{n} \left( f(\lambda_0) e^{-\frac{e}{0.9\lambda_0 - 1}}\right).
\end{align}

Here, we can let 
\begin{align}
    g(\lambda) = ef(\lambda)e^{-\frac{e}{0.9\lambda-1}},
\end{align}

then $g$ is increasing and 
\begin{align}
    \lim_{\lambda_\to \infty} g(\lambda)=1.
\end{align}

It holds that
\begin{align}
a_i(b_i,\b_{-i})p_i(b_i,\b_{-i}) &= \int_{0}^{b_i} t \pd{a_i(t,\b_{-i})}{t} dt\\
&\ge \frac{g(\lambda)}{2e}\cdot\frac{k}{n} b_i^2.
\end{align}

Similar to the argument of Eqs.~(\ref{eqn:similar_begin}-\ref{eqn:similar_end}), the UIR and BF hold when 
\begin{align}
h_* = g(\lambda_0)\cdot\frac{2kc_\rho(n-1)}{en^2}.
\end{align}

\color{black}

\subsubsection{Proof of U-SP.}

Since the variation term of the mechanism for block size $k$ has the same form as block size 1, we can show that the effects of the variation term do not influence the U-SP property in the same way as Appendix~\ref{app:proof:size1:usp}. Furthermore, because fake transactions have zero valuation and non-negative payment, we only need to prove the following proposition:

\begin{proposition}
    For any user $i$, adding a fake bid will not benefit her utility in the Auxiliary Mechanism $M$ for block size $k$.
\end{proposition}

Actually when $p_i(0, \b_{-i})=0$, the payment function in Myerson's Lemma has an equivalent form \citep{shi}:
\begin{small}
\begin{align}
    a_i(b_i, \b_{-i}) p_i(b_i, \b_{-i}) = a_i(b_i, \b_{-i}) b_i - \int_{0}^{b_i} a_i(t, \b_{-i}) dt.
\end{align}
\end{small}
Therefore, the utility of user $i$ when truthfully bidding in the auxiliary mechanism is:

\begin{align}
    u_i(b_i, \b_{-i}; b_i) = \int_{0}^{b_i} a_i(t, \b_{-i}) dt.
\end{align}

Now we only need to show that when we inject a fake transaction, the probability that a user (bidding arbitrary $t$) is confirmed would not increase, as the following lemma:

\begin{lemma} \label{lem:sampling}
    In a weighted random sampling without replacement, if we add a new item, the probability that any already existing item is chosen does not increase.
\end{lemma}

\proof{Proof.}

Consider the Algorithm~\ref{Algo:draw}. Now we assume there are $n$ items $1, \cdots, n$ with weights $w_1, \cdots, w_n$ and without loss of generality we assume $\sum_{i=1}^{n} = 1$, and define $t_j = \sum_{i=1}^j w_i$, then when $X$ is a \emph{i.i.d.}\ uniform random sequence in $[0,1)$, we can see that

\begin{itemize}
    \item $Draw(X, \mathbf{t}, \emptyset, k)$ randomly samples $k$ items among $\{1, \cdots, n\}$ without replacement.
    \item $Draw(X, \mathbf{t}, [t_{n-1}, 1), k)$ randomly samples $k$ items among $\{1, \cdots, n-1\}$ without replacement.
\end{itemize}

We recall that Algorithm~\ref{Algo:draw} performs random drawing without replacement in the following way. Every round an item in $\{1, \cdots, n\}$ is drawn. If the item is already drawn or does not exist, we draw again; otherwise, we finalize it and add it to $S$.

In the rest of the proof, we prove that $\forall i \in \{1, \cdots, n-1\}$, 
\begin{align*}
&\Pr[i \in Draw(X, \mathbf{t}, [t_{n-1}, 1), k)]  \\
&\quad \ge  \Pr[i \in Draw(X, \mathbf{t}, \emptyset, k)]
\end{align*}
by actually showing
\[
i \in Draw(X, \mathbf{t}, \emptyset, k) \: \implies \: i \in Draw(X, \mathbf{t}, [t_{n-1}, 1), k).
\]
In fact, for fixed $X$, because 
\[
P \subseteq Q \quad \implies \quad P \cup R \subseteq Q \cup R,
\]
after each round of drawing, the $B$ in $Draw(X, \mathbf{t}, \emptyset, k)$ is always a subset of the $B$ in $Draw(X, \mathbf{t}, [t_{n-1}, 1), k)$. Therefore, $Draw(X, \mathbf{t}, [t_{n-1}, 1), k)$ would draw no less rounds than $Draw(X, \mathbf{t}, \emptyset, k)$.

Besides, we see that when $i\neq n$, $i$ is drawn if and only if a $x\in [t_{i-1}, t_i)$ appears by the time the drawing completes, so if $i \in Draw(X, \mathbf{t}, \emptyset, k)$, we indeed have $i \in Draw(X, \mathbf{t}, [t_{n-1}, 1), k)$.

Hence we have shown that $\Pr[i \in Draw(X, \mathbf{t}, [t_{n-1}, 1), k)] \ge \Pr[i \in Draw(X, \mathbf{t}, \emptyset, k)]$.

\hfill $\square$

\endproof

From Lemma~\ref{lem:sampling} we prove that our TFM for block size $k$ is U-SP. 

{For the corresponding constants, we note that in the mechanism of block size $1$, the expected payment of a user bidding $b_i$ is lower bounded by $\sim \frac{b_i^2}{2en}$ and $h\lesssim \frac{2c_\rho}{en}$, and it is $\big(C,O(\frac{1}{1-C})\big)$-U-SP for any $C<1$. In the mechanism of block size $k$, the expected payment of user $i$ is lower bounded by $g({\lambda_0})k\cdot \frac{b_i^2}{2en}$, and $h\lesssim g({\lambda_0})k\cdot\frac{2c_\rho}{en}$. Hence, it can be shown in a similar way that Mechanism~\ref{mec:sizek} is also $\big(C,O(\frac{1}{1-C})\big)$-U-SP for any $C<1$.}

\endproof

\subsection{Proof of Theorem~\ref{thm:almost:mic}} \label{app:proof:almost:mic}

Without loss of generality, we can assume the miner will conduct the deviation in this way: in Stage 1 the miner deletes transactions one by one, and then in Stage 2 inject fake transactions one by one. Then we introduce two lemmas before proving the theorem: firstly analyze the robustness of the miner revenue function $\tilde{r}$, then upper bound the advantage the miner may gain in each stage.

\subsubsection{Robustness analysis of the miner revenue function.}

Firstly, we assume that the mean of $b_i^2$ is close to $c_\rho = \Theta(1)$, which holds with high probability with large $n$ and $\Delta = o(n)$. Here we define $H = L c_\rho$, then we prove the following lemma, showing that as long as the average of $\{b_i^2\}$ is close to $c_\rho$, adding or deleting a transaction would not have a significant impact on the miner revenue:

\begin{lemma} \label{lem:mic:robustness}
    If $\left|\frac{\sum_{i=1}^n b_i^2}{n} - c_\rho\right| < \delta$, and recall that
    \begin{align}
        \tilde{r}(\b) = \frac{Hk}{2n} \left(\sum_{i=1}^{n}b_i^2 - \frac{\sum_{1\le i < j \le n}{b_i^2b_j^2}}{c_\rho(n-1)}\right),
    \end{align}
    then for $n\ge 3$, there exists a constant $C_{L\ref{lem:mic:robustness}}$ s.t. $\forall j \in [n]$,
    \begin{align}
        \left|\tilde{r}(\b_{-j}) - \tilde{r}(\b)\right| \le C_{L\ref{lem:mic:robustness}}\delta \cdot \frac{Hk}{c_\rho n}.
    \end{align}
\end{lemma}

\proof{Proof.}
Without loss of generality we assume $j=n$. Then, we compute that
\begin{small}
\begin{align}
&\frac{\tilde{r}(\b_{-n}) - \tilde{r}(\b)}{\frac{1}{2} H k} \nonumber\\
&= \frac{\sum_{i=1}^{n-1} b_i^2}{n-1} - \frac{\sum_{1\le i < j \le n-1}b_i^2 b_j^2}{c_\rho (n-2) (n-1)} \nonumber\\
&\qquad\quad- \frac{\sum_{i=1}^{n} b_i^2}{n} + \frac{\sum_{1\le i < j \le n}b_i^2 b_j^2}{c_\rho (n-1) n} \\
&= \frac{1}{n(n-1)} \sum_{i=1}^{n-1} b_i^2 - \frac{b_n^2}{n} \nonumber\\
&\qquad\quad- \frac{2 \sum_{1\le i < j \le n-1}b_i^2 b_j^2}{c_\rho n(n-1)(n-2)} + \frac{b_n^2 \sum_{i=1}^{n-1} b_i^2}{c_\rho n(n-1)} \\
&= \frac{1}{n(n-1)} \left(\sum_{i=1}^{n-1} b_i^2 - \frac{2 \sum_{1\le i < j \le n-1}b_i^2 b_j^2}{c_\rho (n-2)}\right) \nonumber\\
&\qquad\quad+ \frac{b_n^2}{n}\left(\frac{\sum_{i=1}^{n-1} b_i^2}{c_\rho (n-1)}- 1\right)
\end{align}
\end{small}
From the assumption we see that $\left|\frac{\sum_{i=1}^{n-1} b_i^2}{c_\rho (n-1)}- 1 \right|= O(\delta / c_\rho)$ and $b_n^2 \le 1$, we have 
\begin{align}
\left|\frac{b_n^2}{n}\left(\frac{\sum_{i=1}^{n-1} b_i^2}{c_\rho (n-1)}- 1\right)\right| = O(\delta / c_\rho n).
\end{align}
Now we only need to prove that $\left|\sum_{i=1}^{n-1} b_i^2 - \frac{2 \sum_{1\le i < j \le n-1}b_i^2 b_j^2}{c_\rho (n-2)}\right| = O(\delta n / c_\rho)$.

In fact, we notice that 
\begin{align}
2 \sum_{1\le i < j \le n-1}b_i^2 b_j^2 = \left(\sum_{i=1}^{n-1} b_i^2 \right)^2 - \sum_{i=1}^{n-1} b_i^4.
\end{align}
Hence, 
\begin{align}
&~~~~\left|\sum_{i=1}^{n-1} b_i^2 - \frac{2 \sum_{1\le i < j \le n-1}b_i^2 b_j^2}{c_\rho (n-2)}\right| \nonumber\\
&= \left| \sum_{i=1}^{n-1} b_i^2 - \frac{\left(\sum_{i=1}^{n-1} b_i^2 \right)^2 - \sum_{i=1}^{n-1} b_i^4} {c_\rho (n-2)} \right| \\
&= \left|\sum_{i=1}^{n-1} b_i^2 \cdot \left(1 - \frac{\sum_{i=1}^{n-1} b_i^2}{c_\rho(n-2)}\right) - \frac{\sum_{i=1}^{n-1} b_i^4}{c_\rho(n-2)}\right| \\
&\le \sum_{i=1}^{n-1} b_i^2 \cdot \left|\left(1 - \frac{\sum_{i=1}^{n-1} b_i^2}{c_\rho(n-2)}\right)\right| + \left| \frac{\sum_{i=1}^{n-1} b_i^4}{c_\rho(n-2)}\right| \\
&= O(n) \cdot O(\delta/c_\rho) + O(1) \\
&= O(\delta n / c_\rho).
\end{align}
\hfill $\square$
\endproof

\subsubsection{Advantage analysis of \MTD.} \label{app:adv:mtd}

Now we analyze the advantage in revenue the miner can get after conducting all the transaction deletions. Intuitively, we first show that for large $n$ and $\delta = \omega(\Delta / n)$, the condition $\left|\frac{\sum_{i=1}^n b_i^2}{n} - c_\rho\right| < \delta$ holds with high probability at each step in the $\Delta = o(n)$ deletions. Then we use Lemma~\ref{lem:mic:robustness} to bound the advantage.

First, we deduce the following concentration lemma.

\begin{lemma}
    For any \emph{i.i.d.}\ random variable $\{b_i\} $ in $[0,1]$ satisfying $\E[b_i^2] = c_\rho$ and given $\delta > 0$, we have
    \begin{align}
    \Pr\left[ \left|\frac{\sum_{i=1}^n b_i^2}{n} - c_\rho\right| \ge \frac{\delta}{2} \right] \le 2\exp\left(-\frac{\delta^2 n}{2}\right).
    \end{align}
\end{lemma}

\proof{Proof.}
Hoeffding's inequality \citep{hoeffding1963probability} states that when $\{x_i\}$ are independent random variables with $l_i \le x_i \le r_i$, and denoting $s_n = \sum_{i=1}^n x_i$, it holds that

\begin{align}
    \Pr[|s_n - \E[s_n]| \ge t] \le 2\exp\left(-\frac{2t^2}{\sum_{i=1}^n (r_i - l_i)^2}\right).
\end{align}

Let $x_i = b_i^2, l_i = 0, r_i = 1, t = \frac{\delta n}{2}$, then $\E[s_n] = c_\rho n$ and we get: 

\begin{align}
    \Pr\left[\left|\sum_{i=1}^n b_i^2 - c_\rho n\right| \ge \frac{\delta n}{2}\right] \le 2\exp\left(-\frac{\frac{1}{2}\delta^2 n^2}{n}\right),
\end{align}

i.e.,

\begin{align}
    \Pr\left[ \left|\frac{\sum_{i=1}^n b_i^2}{n} - c_\rho\right| \ge \frac{\delta}{2} \right] \le 2\exp\left(-\frac{\delta^2 n}{2}\right).
\end{align}

\hfill $\square$
\endproof

Then we upper bound the impact of transaction deletion on the average of $\{b_i^2\}$. Without loss of generality, we assume the miner deletes $b_n, b_{n-1}, \cdots, b_{n-t+1}$ sequentially\footnote{Notice that the argument holds for any subset and order of deletion, via  re-permutations of $\{b_i\}$.} for $t\le \Delta = o(n)$, and we want that $\left|\frac{\sum_{i=1}^n b_i^2}{n} - \frac{\sum_{i=1}^{n-t+1} b_i^2}{n-t+1} \right| \le \frac{\delta}{2}$.

In fact, we have $\b_i \in [0,1]$, so 

\begin{align}
    \frac{\sum_{i=1}^{n} b_i^2}{n-t+1} \le \frac{\sum_{i=1}^{n-t+1} b_i^2}{n-t+1} \le \frac{\left(\sum_{i=1}^{n} b_i^2\right) - t}{n-t+1}
\end{align}

Therefore, for $t \le \Delta$, There exists constant $C_{MIC1}$ s.t. for $n \ge C_{MIC1}\frac{\Delta}{\delta}$ and $n-t+1 \ge 3$, we indeed have 

\begin{align}
    \left|\frac{\sum_{i=1}^n b_i^2}{n} - \frac{\sum_{i=1}^{n-t+1} b_i^2}{n-t+1} \right| \le \frac{\delta}{2}.
\end{align}

Combined with Lemma~\ref{lem:mic:robustness}, we deduce that when $n \ge C_{MIC1}\frac{\Delta}{\delta}$, with probability at least $1- 2 \exp\left( \delta^2 n / 2\right)$, the advantage of \MTD~ with $t$ deletions is at most $O(\delta)\cdot \frac{Hkt}{c_\rho n}$. We also see that when we require $\delta \in (0, 1]$, then because $\Delta \ge 1$, $n-t+1\ge 3$ is guaranteed. Formally:

\begin{theorem}[Our mechanism is almost-\{\MTD\}-proof]\label{thm:adv:deletion}
Denote $B^-_{\Delta} (\b)$ as the family of all bidding vectors generated via deleting at most $\Delta$ bids from $\b$. Then for universal constant $C_{MIC1}>0$ and $\delta \in (0, 1], n \ge C_{MIC1}\frac{\Delta}{\delta}$, we have
\begin{small}
\begin{align}
    &\Pr_{\b}\left[\sup_{\b' \in B_{\Delta}^-(\b)}\left(\tilde{r}(\b') - \tilde{r}(\b) \right) > O(\delta) \frac{H k\Delta}{c_\rho n}\right] \nonumber\\
    &\le 2\exp\left(-\frac{\delta^2 n}{2}\right).
\end{align}
\end{small}
\end{theorem}

\subsubsection{Advantage analysis of \MFT.} 

Finally we analyze the miner advantage of the miner's injection of fake transactions. The advantage a miner can get consists of two parts: increase of the miner revenue $\tilde{r}(\cdot)$, and the utility of fake identities. We notice that the robustness analysis of $\tilde{r}(\cdot)$ not only holds for transaction deletion, but also injection. So we can upper bound the miner advantage in the immediate revenue via very similar arguments. Formally, we have (proof omitted):

\begin{corollary}\label{cor:adv:rev}
Denote $B_{\Delta} (\b)$ as the family of all bidding vectors generated via injecting and deleting a total of at most $\Delta$ bids to/from $\b$. Then for universal constants $C_{M0}, C_{M0'}, C_{MIC2}, C_{MIC3}>0$ and $\delta \in (0, 1], n \ge C_{MIC2}\frac{\Delta}{\delta}$,
\begin{align}
    &\Pr_{\b}\left[\sup_{\b' \in B_{\Delta}(\b)}\left(\tilde{r}(\b') - \tilde{r}(\b) \right)> C_{M0}\delta \frac{H k\Delta}{c_\rho n}\right]\nonumber\\
    &\le C_{M0'}\exp\left(-C_{MIC3}\delta^2 n\right).
\end{align}
\end{corollary}

Hence, we only need to further upper bound the advantage from the utility of fake identities. We notice that the fake transactions do not have intrinsic values, so the valuations of fake transactions are zero.

Now we consider the total utility of fake identities. Because the valuations are zero, their total utility are just the opposite of their payment. So for $b'_j \in \b' \backslash \b$, the utility of identity $j'$ is
\begin{align}
    &\tilde{u}_j(b'_j, \b'_{-j}; 0) \nonumber\\
    &= - a_j (b'_j, \b'_{-j}) \tilde{p}_j(b'_j, \b'_{-j}) \\
    &= -a_j (b'_j, \b'_{-j}) p_j(b'_j, \b'_{-j}) - \theta_j(b'_j, \b'_{-j}).
\end{align}
From Eq.~(\ref{eqn:k:payment})\footnote{let $\lambda_0 = 1.582$ and compute $f(\lambda_0), m$ accordingly.}, and denote that the number of bids in $\b'$ is $n' \in [n-\Delta, n+\Delta]$, we get:
\begin{align}
    a_j (b'_j, \b'_{-j}) p_j(b'_j, \b'_{-j}) = \Theta\left(\frac{k {b'_j}^2}{n}\right).
\end{align}
From $h = \frac{Hk}{n'}$ we get:
\begin{align}
    \theta_j(b'_j,\b'_{-j}) = -\frac{Hk}{2n}{b'_j}^2\left(\frac{\sum_{i\ne j}{b'_{i}}^2}{c_\rho(n'-1)}-1\right) 
\end{align}
Similar to the argument in Appendix~\ref{app:adv:mtd}, as long as $c_\rho = \Theta(1)$ and $\left|\frac{\sum_{i=1}^n b_i^2}{n} - c_\rho\right| < O(\delta)$, we have $\left|\frac{\sum_{i\ne j}{b'_{i}}^2}{c_\rho(n'-1)}-1\right| \le O(\delta / c_\rho)$ for any $(\b'\in B_\Delta(b), j \in \b' \backslash \b)$, which happens with probability at least $1-\exp(-\Theta(\delta^2 n))$. 

In this case, we have:
\begin{align}
    \left|\theta_j(b'_j,\b'_{-j})\right| \le O(\delta) \cdot \frac{Hk}{c_\rho n}. 
\end{align}
Therefore, with probability at least $1-\exp(-\Theta(\delta^2 n))$, for any $\b'\in B_\Delta(b)$
\begin{align}
&\sum_{b'_j \in \b' \backslash \b} \tilde{u}_{j}(b'_{j}, \b'_{-j}; 0) \nonumber\\
&= \sum_{b'_j \in \b' \backslash \b} -a_j (b'_j, \b'_{-j}) p_j(b'_j, \b'_{-j}) - \theta_j(b'_j, \b'_{-j}) \\
& = \sum_{b'_j \in \b' \backslash \b} \left( - \Theta\left(\frac{k {b'_j}^2}{n}\right) + O(\delta) \cdot \frac{Hk}{c_\rho n} \right) \\
&\le \sum_{b'_j \in \b' \backslash \b} O(\delta) \cdot \frac{Hk}{c_\rho n} \\
&\le O(\delta) \cdot \frac{Hk\Delta}{c_\rho n}.
\end{align} 
Combined with Corollary~\ref{cor:adv:rev}, we deduce that for universal constants $C_{M0}, C_{MIC2}, C_{MIC3}>0$, $C_{M0'}>1$ and $\delta \in (0, 1], n \ge C_{MIC2}\frac{\Delta}{\delta}$,
\begin{small}
    \begin{align} \label{eqn:amic:almostThere}
        &\Pr_{\b}\left[\sup_{\b' \in B_{\Delta}(\b)}\left(\tilde{r}(\b') - \tilde{r}(\b) + \sum_{b'_j \in \b' \backslash \b} \tilde{u}_{j}(b'_{j}, \b'_{-j}; 0)\right)\right.\nonumber\\
        &\qquad\quad\left. >\vphantom{\sum_{b'_j \in \b'}} C_{M0}\delta \cdot \frac{Hk\Delta}{c_\rho n}\right] < C_{M0'}\exp(-C_{MIC3}\delta^2 n).
    \end{align}
    \end{small}

 Particularly, we can let $\delta = (\Delta / n)^{1/3}$, then for $n \ge C_{MIC2}^{3/2} \Delta$, 
 
\begin{small}
     \begin{align} 
        &\Pr_{\b}\left[\sup_{\b' \in B_{\Delta}(\b)}\left(\tilde{r}(\b') - \tilde{r}(\b) + \sum_{b'_j \in \b' \backslash \b} \tilde{u}_{j}(b'_{j}, \b'_{-j}; 0)\right) \right.\nonumber\\
        &\qquad\:\left.>\vphantom{\sum_{b'_j \in \b'}}C_{M0}  \frac{Hk\Delta^{4/3}}{c_{\rho}n^{4/3}}\right] < C_{M0'} \exp(-C_{MIC3}\Delta^{2/3} n^{1/3}).
    \end{align}
    \end{small}
Therefore, because $\Delta \ge 1$, for any $\epsilon > 0$ when $n \ge \max \{C_{MIC2}^{3/2} \Delta, C_{MIC3}^{-3} \log^3 \frac{C_{M0'}}{\epsilon}\}$, we have
\begin{small}
\begin{align}
&\Pr_{\b}\left[\sup_{\b' \in B_{\Delta}(\b)}\left(\tilde{r}(\b') - \tilde{r}(\b) + \sum_{b'_j \in \b' \backslash \b} \tilde{u}_{j}(b'_{j}, \b'_{-j}; 0)\right)\right.\nonumber\\
&\qquad\qquad\left.\vphantom{\sum_{b'_j \in \b'}}> C_{M0}  \frac{Hk\Delta^{4/3}}{c_\rho n^{4/3}}\right] < \epsilon.
\end{align}
\end{small}

For $\epsilon \in (0, 1/2)$, we have 
\begin{align*}
    \log \frac{C_{M0'}}{\epsilon} &=  \log C_{M0'} + \log \frac{1}{\epsilon} \\
    &= \log \frac{1}{\epsilon} \left(1 + \frac{\log C_{M0'}}{\log \frac{1}{\epsilon}}\right) \\
    &< \log \frac{1}{\epsilon} \left(1 + \frac{\log C_{M0'}}{\log 2}\right).
\end{align*}

 Just let $C_{M1} = C_{MIC2}^{3/2}$, $C_{M2} = C_{MIC3}^{-3}\left(1 + \frac{\log C_{M0'}}{\log 2}\right)^3$, and from  $H =Lc_\rho$, we have proven Theorem~\ref{thm:almost:mic}.

\subsection{Proof of Theorem~\ref{thm:mic:impossibility}}
\label{app:proof:no:mic}

For convenience we let $t = |\b| - 1$. Denote $M\apr$ and $T(\mathbf{\theta}, \tilde{r})$ is the auxiliary-variation decomposition of an 1-SCP mechanism $\tilde{M}$, then from Lemma~\ref{lem:myerson} and Lemma~\ref{lem:1scp} we know that 
\begin{align} \label{eqn:tmp1}
    \theta_i(b_i,\b_{-i}) - \theta_i(0,\b_{-i}) = \tilde{r}(b_i,\b_{-i}) - \tilde{r}(0,\b_{-i}).
\end{align}

User $i$'s utility in $\tilde{M}$ is 
\begin{align}
&\tilde{u}(b_i,\b_{-i};v_i)\nonumber\\
&= a_i(b_i,\b_{-i}) (v_i - p_i(b_i, \b_{-i})) - \theta_i(b_i, \b_{-i}) \\
&= u(b_i,\b_{-i};v_i) - \theta_i(b_i, \b_{-i}).
\end{align}

From U-BNIC of $\tilde{M}$ we know that $\E_{\b_{-i}}[\tilde{u}(b_i+ \delta,\b_{-i};b_i)] \le \E_{\b_{-i}}[\tilde{u}(b_i, \b_{-i};b_i)]$ and $\E_{\b_{-i}}[\tilde{u}(b_i,\b_{-i};b_i + \delta)] \le \E_{\b_{-i}}[\tilde{u}(b_i + \delta, \b_{-i};b_i + \delta)]$, i.e.,
\begin{align}
   & \E_{\b_{-i}}[u(b_i,\b_{-i};b_i) - \theta_i(b_i, \b_{-i})] \nonumber\\
   & \ge \E_{\b_{-i}}[u(b_i + \delta,\b_{-i};b_i) - \theta_i(b_i + \delta, \b_{-i})] \\
  &  \E_{\b_{-i}}[u(b_i,\b_{-i};b_i+ \delta )- \theta_i(b_i, \b_{-i})] \nonumber \\
   & \le \E_{\b_{-i}}[u(b_i + \delta,\b_{-i};b_i + \delta) - \theta_i(b_i + \delta, \b_{-i})].
\end{align}

From U-BNIC (implied by U-DSIC) of $M$ we get:
\begin{align}
    \E_{\b_{-i}}[u(b_i,\b_{-i};b_i) ] & \ge \E_{\b_{-i}}[u(b_i + \delta,\b_{-i};b_i) ] \\
    \E_{\b_{-i}}[u(b_i,\b_{-i};b_i+ \delta )] & \le \E_{\b_{-i}}[u(b_i + \delta,\b_{-i};b_i + \delta) ].
\end{align}

By integration on $b_i$ for fixed $\b_{-i}$, we know that
\begin{align}
    \E_{\b_{-i}}[\theta_{i}(b_i, \b_{-i})] - \E_{\b_{-i}}[\theta_{i}(0, \b_{-i})] = 0.
\end{align}

From NFL we know that $\theta_{i}(0, \b_{-i}) = 0$, so 
\begin{align}
    \E_{\b_{-i}}[\theta_{i}(b_i, \b_{-i})] = 0.
\end{align}

Combined with Eq.~(\ref{eqn:tmp1}), we know that 
\begin{align}
    & \E_{\b_{-i}}[\tilde{r}(b_i,\b_{-i}) - \tilde{r}(0,\b_{-i})]\\
    &= \E_{\b_{-i}}[\theta_i(b_i,\b_{-i}) - \theta_i(0,\b_{-i})] \\
    & = \E_{\b_{-i}}[\theta_i(b_i,\b_{-i})] = 0.
\end{align}

From assumption we know that $\E_{\b_{-i}}[r(0, \b_{-i})] \le \E_{\b_{-i}}[r(\b_{-i})]$, so we have 
\begin{align}
    &\E_{\b}[\tilde{r}(b_i,\b_{-i})] = \E_{b_i}[\E_{\b_{-i}}[\tilde{r}(b_i,\b_{-i})]] \\
    &= \E_{b_i}[\E_{\b_{-i}}[\tilde{r}(0,\b_{-i})]] \le \E_{b_i}[\E_{\b_{-i}}[\tilde{r}(\b_{-i})]] \\
    &= \E_{\b_{-i}}[\tilde{r}(\b_{-i})].
\end{align}

Hence the expected revenue for $t+1$ users is at most the expected revenue for $t$ users. We notice that when there is zero user the expected revenue is non-positive, so by induction, the expected revenue for arbitrary $n$ users is non-positive.

\subsection{Proof of Theorem~\ref{thm:mir:condition}} \label{app:proof:mir:condition}

\textbf{Necessity.} Let $\forall b_i=1$, then it has already been shown that $\tilde{r}(\b) = \Theta(k)\left(1-\frac{1}{2c_\rho}\right).$ If $c_\rho <\frac{1}{2}$, then in this case $\tilde{r}(\b) < 0$, violating MIR.

\textbf{Sufficiency.} Because $\forall b_i \in [0,1]$, we have $b_i^2 \ge b_i ^4$. Therefore,
\begin{small}
\begin{align}
    &\tilde{r}(\b) = \Theta\left(\frac{k}{n}\right)\cdot \left(\sum_{i=1}^{n}b_i^2 - \frac{\sum_{1\le i < j \le n}{b_i^2b_j^2}}{c_\rho(n-1)}\right) \\
    & \ge \Theta\left(\frac{k}{n}\right)\cdot \left(\sum_{i=1}^{n}b_i^4 - \frac{\sum_{1\le i < j \le n}{b_i^2b_j^2}}{c_\rho(n-1)}\right) \\
    &= \Theta\left(\frac{k}{n}\right)\cdot \left(\left(1-\frac{1}{2c_\rho}\right) \sum_{i=1}^{n}b_i^4 \right. \nonumber\\ 
    & \qquad\qquad\qquad  \left. + \frac{1}{4c_\rho (n-1)}\sum_{i=1}^n \left(b_i ^2 - b_j^2\right) ^2\right).
\end{align}
\end{small}

If $c_\rho \ge \frac{1}{2}$, then $1-\frac{1}{2c_\rho} \ge 0$, so $\tilde{r}(\b)$ is lower bounded by a sum of squares. Therefore, $\tilde{r}(\b)$ is non-negative for any $\b \in [0,1]^n$, proving the MIR property of the mechanism.

\subsection{Proof of Theorem~\ref{thm:rev:conc}} \label{app:thm:rev:conc}

We have that
\begin{small}
\begin{align}
    \tilde{r}&(\b) = \frac{h}{2} \cdot \left(\sum_{i=1}^{n}b_i^2 - \frac{\sum_{1\le i < j \le n}{b_i^2b_j^2}}{c_\rho(n-1)}\right) \\
    &= \frac{h}{2} \cdot \left(\sum_{i=1}^{n}b_i^2 - \frac{1}{2c_\rho(n-1)} \left(\left(\sum_{i=1}^{n}b_i^2\right)^2 - \sum_{i=1}^{n}b_i^4 \right)\right). 
\end{align}
\end{small}

By the Cauchy–Schwarz inequality, we have $\left(\sum_{i=1}^{n}b_i^4\right)\cdot  \left(\sum_{i=1}^{n} 1\right) \ge \left(\sum_{i=1}^{n}b_i^2\right)^2$, i.e., 
\begin{align}
    \sum_{i=1}^{n}b_i^4 \ge \frac{1}{n}\left(\sum_{i=1}^{n}b_i^2\right)^2.
\end{align}

Therefore, 
\begin{align}
    \tilde{r}(\b) &\ge \frac{h}{2}\left(  \sum_{i=1}^{n}b_i^2 - \frac{1}{2c_\rho n} \left(\sum_{i=1}^{n}b_i^2\right)^2\right) \\
    &= \frac{h}{2}  \sum_{i=1}^{n}b_i^2 \left( 1 - \frac{1}{2c_\rho n} \sum_{i=1}^{n}b_i^2 \right) \\
    &= \frac{hc_\rho n}{4} \left(1- \frac{1}{c_\rho^2 n^2} \left(\sum_{i=1}^{n}b_i^2 - c_\rho n\right)^2\right).
\end{align}

We know that $\E[\tilde{r}(\b)] = \frac{hc_\rho n}{4}$, so 
\begin{align}\label{eqn:rev:fraction}
    \frac{\tilde{r}(\b)}{\E[\tilde{r}(\b)]} \ge 1- \frac{1}{c_\rho^2 n^2} \left(\sum_{i=1}^{n}b_i^2 - c_\rho n\right)^2.
\end{align} 

Hoeffding's inequality \citep{hoeffding1963probability} states that when $\{x_i\}$ are independent random variables with $l_i \le x_i \le r_i$, and denoting $s_n = \sum_{i=1}^n x_i$, it holds that
\begin{align}
    \Pr[|s_n - \E[s_n]| \ge t] \le 2\exp\left(-\frac{2t^2}{\sum_{i=1}^n (r_i - l_i)^2}\right).
\end{align}

We let $x_i = b_i^2, l_i = 0, r_i = 1, t = \sqrt{\frac{\lambda n}{2}}$, and get:
\begin{align}
    \Pr\left[\left|\sum_{i=1}^{n}b_i^2 - c_\rho n \right|\ge \sqrt{\frac{\lambda n}{2}}\right] \le 2\exp(-\lambda).
\end{align}

Combined with Eq.~(\ref{eqn:rev:fraction}), we get:
\begin{align}
\Pr\left[ \frac{\tilde{r}(\b)}{\E [\tilde{r}(\b)]} \le 1-\frac{\lambda}{c_\rho^2 n}\right] \le 2\exp(- \lambda).
\end{align}

\ifdefined\EnableBurning
\subsection{Proof of Theorem~\ref{thm:burning}} 
\com{TODO: update}
\label{app:proof:thm:burning}


  Denote $k$ as the constant block size, then from symmetry we know that for $\forall i$,
    \begin{align}
        a_i(\mathbf{0}_{[n]}) = \frac{k}{n}.
    \end{align}
    We construct an auxiliary TFM $M=(\a,\p,r)$ to be U-DSIC with the same allocation rule as $M$. Here we define
    \begin{align}
        p_i(b_i,\b_{-i}) &= \begin{cases}
    \frac{\int_{0}^{b_i} t \pd{a_i(t,\b_{-i})}{t}dt}{a_i(b_i,\b_{-i})},&\:a_i(b_i,\b_{-i})>0\\
	0,&\:a_i(b_i,\b_{-i})=0,\\
  \end{cases} \label{eqn:p:def}\\
  r(\b) &= 0.
    \end{align}

    By the Auxiliary Mechanism method, the mechanism $M=(\a,\p,r)$ is indeed U-DSIC (thus also U-BNIC) and 1-SCP.

    Define 
    \begin{align}
        \theta_i(b_i,\b_{-i}) = a_i(b_i,\b_{-i})(\tilde{p}_i(b_i,\b_{-i})-p_i(b_i,\b_{-i})), \label{eqn:theta:def}
    \end{align} 
    
    then from NFL we get
    \begin{align}
        \theta_i(0,\b_{-i}) = 0. \label{eqn:theta:zero}
    \end{align}

    Recall that over the distribution of $\b_{-i} \sim V_{-i}$, the expected utility of user $i$ for bidding $b_i$ is
    \begin{align}
        \mathbb{E}_{\b_{-i} \sim V_{-i}}[\tilde{u}(b_i,\b_{-i};v_i)]=\mathbb{E}_{\b_{-i} \sim V_{-i}}[u(b_i,\b_{-i};v_i) - \theta_i(b_i,\b_{-i})].
    \end{align}
    
    Therefore,
    
    \begin{align}
        \pd{\mathbb{E}_{\b_{-i} \sim V_{-i}}[\tilde{u}(b_i,\b_{-i};v_i)]}{b_i} = \pd{\mathbb{E}_{\b_{-i} \sim V_{-i}}[u(b_i,\b_{-i};v_i)]}{b_i} - \pd{\mathbb{E}_{\b_{-i} \sim V_{-i}}[\theta_i(b_i,\b_{-i})]}{b_i}.
    \end{align}
    
    Because mechanisms $M$ and $\tilde{M}$ are both U-BNIC, it holds that
    \begin{align}
        \left.\pd{\mathbb{E}_{\b_{-i} \sim V_{-i}}[\tilde{u}(b_i,\b_{-i};v_i)]}{b_i}\right|_{b_i=v_i} = \left.\pd{\mathbb{E}_{\b_{-i} \sim V_{-i}}[u(b_i,\b_{-i};v_i)]}{b_i}\right|_{b_i=v_i} = 0
    \end{align}
    
    We deduce 
    \begin{align}
    \left.\pd{\mathbb{E}_{\b_{-i} \sim V_{-i}}[\theta_i(b_i,\b_{-i})]}{b_i}\right|_{b_i=v_i} = 0,
    \end{align}
    
    i.e.,
    
    \begin{align}
        \pd{\mathbb{E}_{\b_{-i} \sim V_{-i}}[\theta_i(b_i,\b_{-i})]}{b_i} = 0.
    \end{align}
    
    From Eq.~(\ref{eqn:theta:zero}) we get that $\mathbb{E}_{\b_{-i} \sim V_{-i}}[\theta_i(0,\b_{-i})] = 0$, so 

    \begin{align}
        \mathbb{E}_{\b_{-i} \sim V_{-i}}[\theta_i(b_i,\b_{-i})] = 0. \label{eqn:theta:exp}
    \end{align}

    From Strong Budget Feasibility, we get 

    \begin{align}
        \tilde{r}(\b) = \sum_{i=1}^n \left(a_i(b_i,\b_{-i}) p_i(b_i,\b_{-i}) + \theta_i(b_i,\b_{-i})\right).
    \end{align}

    Let $\b_{-1}=0$, from NFL we get $\forall i\ne 1$, $a_i(b_i,\b_{-i}) p_i(b_i,\b_{-i}) = 0$ and $\theta_i(b_i,\b_{-i})=0$. Therefore, we get
    \def\zero{\mathbf{0}}
    \begin{align}
        \tilde{r}(b_1,\zero) = a_1(b_1,\zero) p_1(b_1,\zero) + \theta_1(b_1,\zero).\label{eqn:revenue:one:1}
    \end{align} 

    Since $\tilde{M}$ is 1-SCP, from Lemma~\ref{lem:1scp} we have

    \begin{align}
        \tilde{r}(b_1,\zero)-\tilde{r}(0,\zero) = \theta_1(b_1,\zero) - \theta_1(0,\zero).
    \end{align}

    Let $b_1=0$ in Eq.~(\ref{eqn:revenue:one:1}) we get $\tilde{r}(0,\zero)=0$, and from Eq.~(\ref{eqn:theta:zero}) we get $\theta_1(0,\zero)=0$. Therefore,

        \begin{align}
        \tilde{r}(b_1,\zero) = \theta_1(b_1,\zero).\label{eqn:revenue:one:2}
    \end{align} 

    Combine Eq.~(\ref{eqn:revenue:one:1}) and Eq.~(\ref{eqn:revenue:one:2}). We have that

    \begin{align}
        a_1(b_1,\zero) p_1(b_1,\zero) = 0.
    \end{align}

    From the definition of $p_i(\cdot)$ in Eq.~(\ref{eqn:p:def}), we get:

    \begin{align}
    \int_{0}^{b_i} t \pd{a_1(t,\zero)}{t}dt = 0.
    \end{align}

    Since $\tilde{M}$ is 1-SCP, from Lemma~\ref{lem:1scp}, $a_1(\cdot, \zero)$ is monotonic non-decreasing, so $\pd{a_1(t,\zero)}{t} \ge 0$, and therefore we deduce that $\pd{a_1(t,\zero)}{t} = 0$. Hence,

    \begin{align}
        a_1(b_1,\zero) = \frac{k}{n}. \label{eqn:const:one}
    \end{align}

    Now we use induction. We assume that for given $s=s_0(\le n-1)$, $\forall i \le s$, $a_i (b_1, \cdots, b_s, 0, \cdots, 0) = \frac{k}{n}$, and prove that it also holds for $s=s_0+1$.

    Actually, from symmetry and block size $k$, we can get 
    \begin{align}
    a_{s+1} (b_1, \cdots, b_s, 0, 0, \cdots, 0) = \frac{k}{n}. \label{eqn:burning:induction}
    \end{align}.
    
    Furthermore, from symmetry and Eq.~(\ref{eqn:const:one}) , we get 
    \begin{align}
        a_{s+1} (0, \cdots, 0, b_{s+1}, 0, \cdots, 0) = \frac{k}{n}.
    \end{align}

    From competitiveness, it holds that $a_{s+1} (b_1, \cdots, b_s, b_{s+1}, 0, \cdots, 0) \le a_{s+1} (0, \cdots, 0, b_{s+1}, 0, \cdots, 0)$, so
    \begin{align}
        a_{s+1} (b_1, \cdots, b_s, b_{s+1}, 0, \cdots, 0) \le  \frac{k}{n}
    \end{align}
    
    On the other hand, from Lemma~\ref{lem:1scp}, $a_{s+1}$ is monotonic non-decreasing with regard to $b_{s+1}$, so $a_{s+1} (b_1, \cdots, b_s, b_{s+1}, 0, \cdots, 0) \ge a_{s+1} (b_1, \cdots, b_s, 0, 0, \cdots, 0)$. Together with Eq.~(\ref{eqn:burning:induction}), we get

    \begin{align}
        a_{s+1} (b_1, \cdots, b_s, b_{s+1}, 0, \cdots, 0) \ge  \frac{k}{n}.
    \end{align}

    Therefore, $a_{s+1} (b_1, \cdots, b_s, b_{s+1}, 0, \cdots, 0) = \frac{k}{n}$ for $\forall b_{s+1} \in [0,1]$. By symmetry, it also holds for any index $i\le s+1$.

    By induction, we get that in general,

    \begin{align}
        a_i(b_i,\b_{-i}) = \frac{k}{n}, \forall \b, \forall i.
    \end{align}

    This proves Eq.~(\ref{eqn:burn:1}). 
    
    Further from Eq.~(\ref{eqn:p:def}), we get $p_i(b_i,\b_{-i}) \equiv 0$. 

    From Eq.~(\ref{eqn:theta:def}) and by StrongBF and symmetry, we get

    \begin{align}
        \mathbb{E}_{\b \sim V} [\tilde{r}(\b)] &= \mathbb{E}_{\b \sim V}\left[\sum_{i=1}^n \left(a_i(b_i,\b_{-i}) p_i(b_i,\b_{-i}) + \theta_i(b_i,\b_{-i})\right) \right] \\
        &= \sum_{i=1}^n \mathbb{E}_{\b \sim V}\left[a_i(b_i,\b_{-i}) p_i(b_i,\b_{-i}) + \theta_i(b_i,\b_{-i})\right] \\
        &= n \mathbb{E}_{\b \sim V}\left[a_1(b_1,\b_{-1}) p_1(b_1,\b_{-1}) + \theta_1(b_1,\b_{-1})\right] \\
        &= n \mathbb{E}_{\b \sim V}\left[ \theta_1(b_1,\b_{-1})\right] \\
        &= n \mathbb{E}_{b_1 \sim V_1}\left[\mathbb{E}_{\b_{-1} \sim V_{-1}}[ \theta_1(b_1,\b_{-1})]\right].
    \end{align}

    According to Eq.~(\ref{eqn:theta:exp}) we get $\mathbb{E}_{\b_{-1} \sim V_{-1}}[ \theta_1(b_1,\b_{-1})] = 0$, so

    \begin{align}
        \mathbb{E}_{\b \sim V} [\tilde{r}(\b)] = 0.
    \end{align}

    This proves Eq.~(\ref{eqn:burn:2}).

\fi

\theendnotes
 

		
		


		
		
	\end{document}